\documentclass[final, 11pt]{article}
\usepackage[margin=1in]{geometry}

\usepackage{xspace}
\makeatletter
\newcommand{\specificthanks}[1]{\@fnsymbol{#1}}
\makeatother

\usepackage{paralist}
\usepackage{bbold}

\usepackage{graphicx}
\usepackage{grffile}

\usepackage{pythonhighlight}

\usepackage{algorithm}
\usepackage[noend]{algpseudocode}

\newif\ifdraft

\usepackage{pgfplots}
\usepackage{amsmath,amsfonts,amssymb,mathtools}
\usepackage{subcaption}

\usepackage{graphicx,float}
\usepackage[short]{optidef}

\newcommand{\etal}{\textit{et~al.}\xspace}

   \usepackage[amsmath, thmmarks]{ntheorem}
   \newtheorem{theorem}{Theorem}[section]

   \newtheorem{lemma}{Lemma}[section]

   \newtheorem{definition}{Definition}[section]
   
   \newtheorem{problem}{Problem}
   \newtheorem{assumption}{Assumption}

\author{
Harb, Elfarouk \thanks{Supported in part by NSF CCF-1910149. We thank Vasilis Livanos and Chandra Chekuri for helpful feedback, discussions, and manuscript improvement. We thank Raimundo Saona for help with replicating some results in \cite{csz-pstbs-18}. We are particularly indebted to Sariel Har-Peled for several ideas and valuable feedback on the manuscript. In particular, the $O(\log^\ast n)$ load analysis is due to Sariel; the author had a looser analysis of $O(\log \log n)$.} \\
  University of Illinois at Urbana-Champaign  \\
  \texttt{eyfmharb@gmail.com}
}

\newcommand{\Alg}{{\textsf{ALG}}\xspace}
\newcommand{\constant}{0.6724\xspace}
\newcommand{\constantSemiOnline}{0.89\xspace}

\newcommand{\IID}{{\textsf{IID}}\xspace}

\newcommand{\pth}[2][\!]{\mleft({#2}\mright)}

\newcommand{\pbrcx}[1]{\left[ {#1} \right]}

\usepackage{mleftright}

\newcommand{\Prob}[1]{\mathop{\mathbb{P}}\mleft[{#1}\mright]}

\newcommand{\ProbCond}[2]{\mathop{\mathbb{P}}\mleft[{#1} \middle\vert\; #2 \mright]}

\usepackage{txfonts}

\newcommand{\Ex}[2][\!]{\mathbb{E} #1\pbrcx{#2}}

\newcommand{\cardin}[1]{\left| {#1} \right|}

\usepackage[title]{appendix}

\renewcommand{\Re}{\mathbb{R}}

\newcommand{\Poisson}{{\textsf{Poisson}}\xspace}
\newcommand{\Bern}{{\textsf{B}}\xspace}
\newcommand{\Uniform}{{\textsf{Uniform}}\xspace}
\newcommand{\Bin}{{\textsf{Bin}}\xspace}

\newcommand{\Term}[1]{\textsf{#1}}

\newcommand{\CDF}{\Term{cdf}\xspace}

\newcommand{\PMF}{\Term{pmf}\xspace}

\newcommand*\dif{\mathop{}\mathrm{d}}

\usepackage{hyperref}
\usepackage{xcolor}
\hypersetup{
    colorlinks,
    linkcolor={red!50!black},
    citecolor={blue!50!black},
    urlcolor={blue!80!black}
}

\usepackage{cleveref}

   \usepackage[amsmath,thmmarks]{ntheorem}
   \theoremseparator{.}

   \theoremstyle{plain}

   \theoremstyle{plain}
   \theoremheaderfont{\sf} \theorembodyfont{\upshape}
   \newtheorem*{remark:unnumbered}{Remark}
   
   \newtheorem{remark}{Remark}

\newcommand{\myqedsymbol}{\rule{2mm}{2mm}}

\theoremheaderfont{\em}
\theorembodyfont{\upshape}
\theoremstyle{nonumberplain}
\theoremseparator{}
\theoremsymbol{\myqedsymbol}
\newtheorem{proof}{Proof:}
\newtheorem{proofof}{Proof of\!}

\newcommand{\HLink}[2]{\hyperref[#2]{#1~\ref*{#2}}}
\newcommand{\HLinkSuffix}[3]{\hyperref[#2]{#1\ref*{#2}{#3}}}

\newcommand{\figlab}[1]{\label{fig:#1}}
\newcommand{\figref}[1]{\HLink{Figure}{fig:#1}}

\newcommand{\algolab}[1]{\label{alg:#1}}
\newcommand{\algoref}[1]{\HLink{Algorithm}{alg:#1}}

\newcommand{\tablab}[1]{\label{tab:#1}}
\newcommand{\tabref}[1]{\HLink{Table}{tab:#1}}

\newcommand{\seclab}[1]{\label{sec:#1}}
\newcommand{\secref}[1]{\HLink{Section}{sec:#1}}

\providecommand{\eqlab}[1]{}
\renewcommand{\eqlab}[1]{\label{equation:#1}}
\newcommand{\Eqref}[1]{\HLinkSuffix{Eq.~(}{equation:#1}{)}}

\newcommand{\lemlab}[1]{\label{lemma:#1}}
\newcommand{\lemref}[1]{\HLink{Lemma}{lemma:#1}}

\newcommand{\apdxlab}[1]{{\label{apdx:#1}}}
\newcommand{\apdxref}[1]{\HLink{Appendix}{apdx:#1}}

\newcommand{\Xlemref}[1]{\noexpand{\noexpand\HLink{Lemma}{lemma:#1}}}
\newcommand{\Xthmref}[1]{\noexpand{\noexpand\HLink{Theorem}{theo:#1}}}

\newcommand{\ICDFY}[2]{\mathsf{q}_{#1}\pth{#2}}

\usepackage{enumitem}
\usepackage{booktabs}

\setlist[enumerate]{align=left}
\usepackage{mdwlist}
\usepackage{dsfont}

\title{New Prophet Inequalities via Poissonization and Sharding}

\renewenvironment{equation*}{\[}{\]\ignorespacesafterend}

\begin{document}

\pagenumbering{gobble}

\maketitle

\begin{abstract}
This work introduces \emph{sharding} and \emph{Poissonization} as a unified framework for analyzing prophet inequalities. Sharding involves splitting a random variable into several independent random variables, shards, that collectively mimic the original variable's behavior. We combine this with Poissonization, where these shards are modeled using a Poisson distribution. Despite the simplicity of our framework, we improve the competitive ratio analysis of a dozen well studied prophet inequalities in the literature, some of which have been studied for decades. This includes the \textsc{Top-$1$-of-$k$} prophet inequality,  prophet secretary inequality, and semi-online prophet inequality, among others. This approach not only refines the constants but also offers a more intuitive and streamlined analysis for many prophet inequalities in the literature. Furthermore, it simplifies proofs of several known results and may be of independent interest for other variants of the prophet inequality, such as order-selection. 
\end{abstract}
\newpage 
\pagenumbering{arabic}

\section{Introduction, Related Work, and Contributions. }

The field of optimal stopping theory concerns the optimization settings
where one makes decisions in a sequential manner, given imperfect
information about the future, with an objective to maximize a reward
or minimize a cost. The classical problem in the field is known as the \emph{prophet inequality} problem \cite{ks-sfv-77, ks-sapfv-78}. In this problem, a gambler is presented with $n$ \emph{non-negative} independent random variables $X_1, \ldots X_n$ with known distributions. We  assume, without loss of generality, that the random variables are continuous. In iteration $t$, a random realization value $v_t$ is drawn from the distribution of $X_t$ and presented to the gambler. The gambler may accept $v_t$, concluding the game, or irrevocably reject $v_t$ to proceed to iteration $t+1$. Note that the random variable ordering is chosen adversarially by an almighty adversary that knows the gambler's algorithm. The goal of the gambler is to maximize their expected reward, where the expectation is taken across all possible realizations of $X_1, \ldots, X_n$. The gambler is compared to a \emph{prophet} who is allowed to make their decision after seeing all realizations (i.e., can always select $\max(v_1, \ldots v_n)$) regardless what realizations occur. In other words, the prophet receives a value $Z$ with expectation $\Ex{Z} = \Ex{\max(X_1, \ldots X_n)}$. An algorithm \Alg
is \emph{$\alpha$-competitive}, for $\alpha \in [0,1]$, if
$\Ex{\Alg} \geq \alpha \cdot \Ex{Z}$, and $\alpha$ is called the
\emph{competitive ratio}.

The classic prophet inequality asserts the existence of a $1/2$-competitive algorithm and, moreover, that this is tight. The first algorithm to give the $1/2$ analysis is due to Krengel and Sucheston \cite{ks-sfv-77, ks-sapfv-78}. Later, Samuel-Cahn \cite{s-ctsrm-84} gave a simple and elegant algorithm that sets a single
threshold $\tau$ as the median of the distribution of $Z=\max_i X_i$, and accepts the first value (if any) above $\tau$. She showed that the algorithm is $1/2$ competitive and, moreover, this is tight. Kleinberg and Weinberg \cite{kw-mpiam-19} also showed that setting $\tau = \Ex{\max_i X_i}/2$ also gives a $1/2$-competitive algorithm. 

The preceding discussion assumes only the independence of the distributions of $X_1, \ldots, X_n$. For \IID\footnote{Independent and identically distributed} non-negative random variables $X_1, \ldots, X_n$, Hill and Kertz \cite{hk-csrse-82} initially gave a $(1 - 1/e)$-competitive algorithm. This was improved by Abolhassani, Ehsani, Esfandiari, Hajiaghayi, and Kleinberg \cite{aeehk-b1op-17} in STOC 2017 into a $\approx 0.738$ competitive algorithm. This was improved to the tight $\approx 0.745$ in a result due to Correa, Foncea, Hoeksma, Oosterwijk, and Vredeveld \cite{cfhov-ppmot-21}. This constant is tight due to a matching upper bound, and hence the \IID special case is also resolved. Throughout the paper, unless explicitly stated otherwise, random variables are not assumed to be \IID.

Numerous variants of the prophet inequality problem are known. We list some below. 
\begin{problem}
    \textsc{Random-Order}: The variant of the prophet inequality problem where the random variables realizations arrive in the order of a random permutation (i.e., the order is not adversarial). This is also known as the \textsc{prophet secretary} problem. 
\end{problem}

\begin{problem} \textsc{Top-$1$-of-$k$}: This is a generalization of the single-choice prophet inequality. In this model, an adversary arranges the random variables $X_1, ..., X_n$ adversarially. The gambler may choose up to $k \geq 2$ outcomes, going beyond the single-choice limitation. The gambler's final reward is the \textbf{highest} value among the selected outcomes. 
\end{problem}

\begin{remark}
    We contrast this to the $k$-cardinality constraint model, akin to the \textsc{Top-$1$-of-$k$} model, which focuses on selecting up to $k$ values to maximize the \emph{sum} of chosen values. The algorithm reward is compared against a prophet achieving the sum of the top-$k$ values in each realization \cite{hks-aomdpi-07, jwj-tgmpiosk-22}. Other generalizations exist like maximizing under matroid constraints \cite{kw-mpi-12}. We do not discuss these variants.
\end{remark}

\begin{problem}
    \textsc{Order-Selection}: In this variation of the prophet inequality problem, the gambler is allowed to determine the order in which the random variables are presented to them. 
\end{problem}

\begin{problem} \textsc{Semi-Online}: In this variant, the variables' actual values are kept hidden from the gambler. The gambler can make $n$ adaptive queries, each inquiring if ``$X_i \geq \tau_i$?'', where $\tau_i$ is chosen by the gambler. Each random variable is eligible for only \emph{one} query. After all $n$ queries have been exhausted, the gambler selects the variable that holds the highest conditional expectation.
\end{problem}
\begin{problem} \textsc{Semi-Online-Load-Minimization} (\textsc{SOLM}): This variant resembles the \textsc{Semi-Online} setting but allows the flexibility of multiple queries per variable. However, there is still a limit of $n$ queries in total. The objective is to achieve a competitive ratio of $1-o(1)$, while minimizing the maximum number of times a single variable is queried, referred to as the \emph{load}.
\end{problem}

\paragraph{Prophet Secretary} In the \textsc{Random-Order} variation, Esfandiari, Hajiaghayi, Liaghat, and Monemizadeh  \cite{ehlm-ps-17}  initially gave a $1-\frac{1}{e} \approx 0.632$ competitive algorithm. Azar, Chiplunkar, and Kaplan \cite{ack-pss1b-18} later refined this to \(1-\frac{1}{e}+\frac{1}{400} \approx 0.634\) at EC 2018. While the improvement is small, the case-by-case analysis introduced was non-trivial, exposing the intricacies of the problem. Correa, Saona, and Ziliotto \cite{csz-pstbs-18} further improved the competitive ratio to $\approx 0.669$ by adopting the notion of \emph{discrete} blind strategies at SODA 2019. This required less case-by-case analysis. 
Meanwhile, current impossibility results show that no algorithm can achieve a competitive ratio better than $0.7235$ \cite{gmts-pisrofos-23}.

\paragraph{\textsc{Top-$1$-of-$k$}} Assaf and Samuel-Cahn \cite{as-srpimm-00} first introduced the \textsc{Top-$1$-of-$k$} variant in the context of prophet inequalities, building on the seminal work by Gilbert and Mosteller \cite{gm-rms-66}. They proposed a simple and elegant algorithm with a competitive ratio of $k/(k+1)$ for any $k \geq 2$, and noted that for $k=2$, one cannot do better than $0.8$. In a followup paper, Assaf, Samuel-Cahn, and Goldstein \cite{ags-rpiwm-02} offered a highly non-trivial tighter analysis for $k \geq 2$, finding competitive ratios of roughly $0.731$ for $k=2$, $0.8479$ for $k=3$, and $0.9108$ for $k=4$. These ratios are defined by a recursive differential equations, making it hard to understand their behavior for larger $k$. 

Later, Ezra, Feldman, and Nehama \cite{efn-pso-18} revisited the problem and improved the lower bound for large $k$ to $1-1.5e^{-k/6}$, showing a new exponential relationship with $k$. They also proved an upper bound of $1-\frac{1}{(2k+2)!}$ for any $k$. However, their improvements did not affect the lower bounds for smaller $k$ values initially found by Assaf, Samuel-Cahn, and Goldstein using recursive differential equations. 

\textit{Note on small $k$:} Exact competitive ratios for small $k$ in \textsc{Top-$1$-of-$k$} are important, as studied in previous works \cite{ags-rpiwm-02, jwj-tgmpiosk-22}. These ratios are especially relevant in applications where \textsc{Top-$1$-of-$k$} serves as a small buffer, which typically has a small size. For example, with $k=4$ and \IID random variables, we demonstrate an algorithm with a competitive ratio of approximately 0.98, almost matching the prophet's performance. This paper, therefore, first concentrates on smaller values of $k$ (say $k=2, 3, 4$), providing a tight analysis for them using our framework, before making the analysis slightly looser to address the asymptotic behavior of the algorithm as $k\to \infty$.

\paragraph{Order-selection} The order-selection problem has had more progress than random-order. Specifically, since a random-order is a valid order for order-selection, then the result of Correa \etal \cite{csz-pstbs-18} of $\approx 0.669$ remained the state of the art. This was improved recently in FOCS 2022 to a $0.7251$-competitive algorithm by Peng and Tang \cite{bz-ospi-22}. In a followup work at EC 2023, Bubna and Chiplunkar \cite{bc-piosb-22} showed that the analysis of Peng and Tang \cite{bz-ospi-22} method cannot be improved, and gave an improved $0.7258$ competitive algorithm (i.e., improvement in the $4\textsuperscript{th}$ digit) for order-selection using a slightly different approach. The also proved no algorithm can do better than $0.7254$ in the random order model. This finally created a \textit{separation} from the random order model: there is a strict advantage of order-selection over random-order. Thus, the optimal order-selection strategy is \emph{not} a random permutation. The separation result was also established independently by Giambartolomei, Mallmann-Trenn and Saona in \cite{gmts-pisrofos-23} around the same time.
\paragraph{Semi-Online}  Hoefer and Schewior  \cite{hs-ttsopi-23} introduced the semi-online prophet inequalities variants, focusing exclusively on the \IID case, and deferred the more complex general case (i.e., Non-\IID) versions for future work. For the \IID \textsc{Semi-Online} problem, they proposed an algorithm with a competitive ratio of $0.869$, significantly outperforming the $\approx 0.745$ ratio of the classical \IID prophet inequality. Furthermore, they showed that no algorithm could exceed a $0.9799$ competitive-ratio \footnote{In a private correspondence, the authors of \cite{hs-ttsopi-23} confirmed they knew (post publication) of a hardness example which shows an improved upper bound of $\approx 0.92$. The author of this paper has not seen that hardness example.}. They also established that the \textsc{Semi-Online-Load-Minimization} problem for \IID random variables is solvable with an $O(\log(n))$ load.

\subsection*{Contributions} 

Our main contribution is the introduction of a new framework, \emph{Poissonization} and \emph{sharding}, to analyze and improve upon prophet inequalities. \textit{These concepts are simple, yet powerful. We show that they unify and improve upon the analysis of several prophet inequalities, that have been studied using more specialized methods for decades.} Moreover, this framework considerably simplifies numerous proofs of known results in the literature, making them more accessible. 

\paragraph{Poissonization} Here, we outline the key idea of ``Poissonization'', we defer the technical details to the main body. The original idea of ``Poissonization'' refers to the following. Suppose we have $n$ Bernoulli random variables $X_1, ..., X_n$ with probability $p$. Let $S_n=\sum_{i=1}^n X_i$, and suppose that $np$ is ``small''. Then the standard Poissonization argument says that  $S_n$ ``behaves'' the same as a Poisson random variable $T_n \sim \Poisson(np)$. Known generalizations of this exist. For example, Le Cam's theorem states that if $X_i \sim \Bern(p_i)$, and $\lambda=\sum_{i=1}^n p_i$, then $S=\sum_i X_i$ ``behaves'' the same as $T_n\sim \Poisson(\lambda)$. The error (in terms of the variational distance) of the approximation is guaranteed to be at most $\leq 2\sum_{i=1}^n p_i^2 $, and hence if all the $p_i$ are ``small'' (say $p_i=c_i/n$ for some constant $c_i$), then the approximation is good. 

Poisson distributions have several desirable properties including the memorylessness property, closed additivity (If $X\sim \Poisson(\lambda_1), Y\sim \Poisson(\lambda_2)$, then $X+Y \sim \Poisson(\lambda_1+\lambda_2)$), and a simple \PMF \footnote{Probability mass function} (If $X \sim \Poisson( \lambda)$, then $\Prob{X = k} = e^{-\lambda} \lambda^k /k!$). Hence, when the error is small, we would prefer to work with the Poisson random variables in computing probabilities, rather than the original sum of Bernoulli random variables.

For our case, we need a higher order generalization of Poissonization. In particular, our random variables will be $k$-dimensional $X_i \in \Re^k$, and we want a similar Poissonization result on $S_n = \sum_i X_i$ in terms of a $k$-dimensional Poisson random variable.

\paragraph{Sharding} Here, we briefly introduce the idea of sharding. We  defer showing more examples of using sharding to the main body. Suppose we are given $n$ random variables $X_1, ..., X_n$ that are not necessarily \IID. The idea of sharding is to first ``break'' each $X_i$ into $K\geq 2$ \IID random variables $\{Y_{i, j}\}_{1\leq j \leq K}$. If the \CDF \footnote{Cumulative distribution function} of $X_i$ is $F$, then $Y_{i,j}$ has \CDF $F^{1/K}$. Finally (and  importantly), we take $K\to \infty$. Hence, it can be thought that each random variable was finely ``broken'' into small shards or splinters. 

Shards \textbf{collectively} behave similar to \IID random variables. In addition, the distribution of $\max(Y_{i, 1}, ..., Y_{i, K})$ is precisely the distribution of $X_i$: 
\[
\Prob{\max(Y_{i, 1}, ..., Y_{i, K})\leq \tau}=\Prob{Y_{i, 1}\leq \tau}^K = F^{1/K}(\tau)^K = F(\tau).
\] 

By using a Poissonization argument on the shards $\{Y_{i,j}\}$, we are able to derive a closed form \textbf{exact} formula for the probability that there are $k$ shards above some threshold $\tau$ (i.e., the probability that $k$ of $Y_{i,j}$ are $\geq \tau$). Finally, we bound the competitive ratio of the algorithm in terms of events on the \emph{shards}, instead of on $X_1, ..., X_n$. 

\subsection*{New results.}
Below we present the main \textit{new} results obtained using our framework. \tabref{tab:my_label} provides a summary of these improved results, excluding simplified results. \textit{The common denominator in all the results is the application of the Poissonization and sharding framework}. We believe that Poissonization and sharding will become a central tool in tackling prophet inequality type problems, despite the framework's simplicity. In particular, we believe our analysis might be of independent interest for similar problems such as the prophet inequality with order-selection. We sketch some ideas for achieving that in the conclusion and leave it for future work to extend the analysis we have here for the order selection problem. 

\begin{enumerate}[leftmargin=5pt]
    \item For the \textsc{Top-$1$-of-$k$} model, our results significantly improve the long standing bounds of Assaf and Samuel-Cahn \cite{as-srpimm-00, ags-rpiwm-02}, demonstrating that even for $k=2$, both the upper and lower bounds by them are not optimal. For $k=2$, we improve the lower bound from approximately $0.731$ to $0.781$ and the upper bound from $0.8$ to approximately $0.794$, almost resolving the model. For general $k$, we refine the bound by Ezra, Feldman, and Nehama \cite{efn-pso-18} to a $1-e^{-kW(\frac{\sqrt[k]{k!}}{k})}$ competitive algorithm, where $W$ is the Lambert $W$ function\footnote{The Lambert $W$ function $W(z)$ satisfies $W(z)e^{W(z)} = z$.}.
\end{enumerate}

\begin{theorem}{(Proof in \secref{best1ofk})}
    There exists an algorithm for the \textsc{TOP-$1$-of-$2$} problem with a competitive ratio of $0.781$. No algorithm can achieve a competitive ratio higher than $0.794$. For any $k$, there exists an algorithm for \textsc{TOP-$1$-of-$k$} with a competitive ratio of at least $1-e^{-kW(\frac{\sqrt[k]{k!}}{k})}$, which asymptotically approaches $1-e^{-kW(1/e) + o(k)}$ as $k$ increases.
\end{theorem}

\begin{enumerate}[resume,leftmargin=5pt]
    \item In the case of \IID random variables for the \textsc{Top-$1$-of-$k$} model, we improve the results for both small and large $k$. Specifically, let $\zeta_k$ be the unique positive solution to 
    \begin{equation*}
    1-e^{-x} = \sum_{i=0}^{k-1} e^{-x}\frac{x^i}{i!} + \sum_{i=k}^\infty\sum_{j=0}^k  e^{-x} \frac{x^i}{i!} \frac{{i+1 \brack j}}{(i+1)!},
    \end{equation*}
    where ${r \brack s}$ represents the (unsigned) Stirling number of the first kind. We show that there exists an algorithm for the \textsc{Top-$1$-of-$k$} problem with \IID random variables that achieves a competitive ratio of $1-e^{-\zeta_k}$, significantly improving the previous bounds for $k=2, 3,$ and $4$. Additionally, for general $k$, we present the first algorithm with a \emph{super-exponential} competitive ratio of at least $1-k^{-k/5}$, improving upon the previous exponential bound.
\end{enumerate}

\begin{theorem}{(Proof in \secref{best1ofk})}
    For $k=2,3,4$, there is an algorithm for the \textsc{TOP-$1$-of-$k$} problem with \IID random variables that achieves competitive ratios of at least $0.883, 0.946,$ and $0.9816$, respectively. For any $k \geq 1$, there is an algorithm that achieves a competitive ratio of $1 - k^{-k/5}$.
\end{theorem}

\begin{enumerate}[resume,leftmargin=5pt]
    \item For the prophet secretary problem, previously studied in \cite{ehlm-ps-17, ack-pss1b-18, csz-pstbs-18}, we raise the lower bound from $0.669$ to $0.6724$. The improved algorithm uses a continuous blind strategy.
\end{enumerate}

\begin{theorem}{(Proof in \secref{main})}
    There is an algorithm for the prophet secretary problem that achieves a competitive ratio of at least $\constant$. 
\end{theorem}

\begin{enumerate}[resume,leftmargin=5pt]
    \item For the \IID \textsc{Semi-Online} problem, we improve the lower bound from $0.869$ to  $0.89$. This improvement is achieved by adopting an adaptive strategy that progressively lowers the threshold over time, combined with a novel \textit{discrete clock} analysis using dynamic programming.
\end{enumerate}

\begin{theorem}{(Proof in \secref{iidsemionline}) }
    There exists an algorithm for the \IID \textsc{Semi-Online} problem that has a competitive ratio of at least $\constantSemiOnline$.
\end{theorem}

\begin{enumerate}[resume,leftmargin=5pt]
    \item For the \textsc{Semi-Online-Load-Minimization} (\textsc{SOLM}) problem, both for \IID and Non-\IID settings, we improve the upper bound. Previously, the Non-\IID \textsc{SOLM} was an open question, and for \IID variables, an $O(\log n)$ load was established. We demonstrate that with $O(\log^\ast n)$ load, it's possible to achieve a $1-o(1)$ competitive ratio for both \IID and Non-\IID variables.
\end{enumerate}

\begin{theorem}{(Proof in \secref{loadminimization})}
    There is an algorithm for the \textsc{Semi-Online-Load-Minimization} (\textsc{SOLM}) problem that uses $O(\log^\ast n)$ load for both \IID and Non-\IID random variables.
\end{theorem}

\begin{table}
    \centering
    \begin{tabular}{p{2.8cm} p{2.5cm} p{4.3cm} p{5.6cm}}

        \toprule  
        \toprule
        Problem & Bound type & Known results & New result \\
        \midrule 
        Prophet Secretary & Lower bound & $0.669$  \cite{csz-pstbs-18} & $0.6724$,  \lemref{prophet:secretary:improvement}.\\
        \midrule
        \textsc{Top-$1-$of-$k$} & Lower bound& $1-1.5e^{-k/6}$ \cite{efn-pso-18} & $1-e^{-kW(\frac{\sqrt[k]{k!}}{k})}$, \lemref{general:best1ofk:lb:improvement}.  \\
        \midrule 
        \textsc{\IID Top-$1-$of-$k$} & Lower bound& $1-1.5e^{-k/6}$ \cite{efn-pso-18} & $1-k^{-k/5}$, \lemref{best1ofk:iid:improvement:generalk}.  \\

        \midrule 
        \textsc{\IID Semi-Online} &Lower bound& $0.869$  \cite{hs-ttsopi-23} & $0.89$, \secref{iidsemionline}. \\
        \midrule  
        \textsc{SOLM} & Upper bound& $O(\log n)$ for \IID r.vs \cite{hs-ttsopi-23}  & $O(\log^\ast n)$ for general case,  \secref{loadminimization}.   \\
        \midrule 
        \textsc{Top-$1-$of-$2$} & Lower bound& $0.731$ \cite{ags-rpiwm-02} &  $0.781$, \lemref{noniid:firstimprovement} and \lemref{noniid:secondimprovement}.  \\
        \midrule
        \textsc{Top-$1-$of-$2$} &Upper bound & $0.8$  \cite{as-srpimm-00} & $0.7943$,  \lemref{betterub:best1of2}.  \\
        \midrule   
        \textsc{\IID Top-$1-$of-$2$} &Lower bound& $0.745$  \cite{cfhov-ppmot-21} & $0.883$, \lemref{best1of2:iid:finalalgorithm}. \\
        \midrule 
        \textsc{\IID Top-$1-$of-$3$}& Lower bound & $0.8479$ \cite{ags-rpiwm-02} &  $0.9463$,  \lemref{best1ofk:iid:firstimprovement}.  \\
        \midrule 
        \textsc{\IID Top-$1-$of-$4$} & Lower bound& $0.9108$ \cite{ags-rpiwm-02} & $0.9816$,  \lemref{best1ofk:iid:firstimprovement}.  \\
        \bottomrule 
    \end{tabular}
    \caption{Summary of some new results.}
    \label{tab:my_label}
    \tablab{tab:my_label}
\end{table}
\begin{remark}
    Following a preprint of our paper, Har-Peled, Harb, and Livanos \cite{oracle-augmented} introduced a new variant for prophet inequalities termed as oracle-augmented prophet inequalities. Using our framework, they developed an optimal single-threshold algorithm for this new model. This marks yet another prophet inequality that leverages our framework, reinforcing the case that this is a unifying analytic framework for prophet inequalities.
\end{remark}

\subsection*{Simplified Results.}

We also present new, considerably simpler proofs for several established results in the literature:

\begin{enumerate}[resume,leftmargin=5pt]
    \item At \lemref{prooffromthebook}, a ``proof from the book'' is provided for the $1 - \frac{1}{e}$ competitive single-threshold algorithm for the prophet secretary problem. This proof is notably simple, boiling down to the calculation of an elementary sum, combined with our framework. 
    
    \item For discrete blind strategies, we offer simpler proofs of key lemmas initially presented in \cite{csz-pstbs-18} at \lemref{simplified-upper} and \lemref{discrete_blind_lb}. The original arguments were complex, utilizing Schur-convex minimization. Our proofs are elementary and from first principles using our framework. 
    
    \item At Section 3, \hyperref[eta10]{an alternative simpler proof} is provided for achieving the competitive ratio of $\approx 0.745$ for the standard \IID prophet inequality. The original tight $\approx 0.745$ \cite{cfhov-ppmot-21} is quite technical, although known simplifications under mild assumptions exist in Sahil Singla's PhD thesis \cite{s-couup-18}. 
\end{enumerate}

\subsection*{Outline of Framework.}
We outline our framework, illustrated with a motivating example. For a more detailed formalization and additional examples, refer to \secref{warmupiid} and \secref{sharding}. Let ${\bf X} = X_1, ..., X_n$ denote a sequence of continuous independent random variables with cumulative distribution functions (\CDF) $F_1, ..., F_n$. We use $\cardin{\beta \leq {\bf X} \leq \alpha} = \cardin{\{i : \beta \leq X_i \leq \alpha\}}$ to represent the count of variables in ${\bf X}$ falling within the interval $[\beta, \alpha]$. Instead of directly sampling from $F_i$, each $X_i$ is divided into $K$ shards ${\bf Y_i} = Y_{i,1}, ..., Y_{i,K}$ with \CDF $F_i^{1/K}$, and we set $X_i = \max_j(Y_{i,j})$. This results in a new sequence of $Kn$ variables ${\bf S} = {\bf Y}_1 \cdot ... \cdot {\bf Y}_n$, with $\cdot$ indicating concatenation.

A key observation is that for any threshold $\tau$ and integer $t$,
\begin{equation*}
    \Prob{\cardin{\tau \leq {\bf X} < \infty} \geq t} \leq \Prob{\cardin{\tau \leq {\bf S} < \infty} \geq t}.
\end{equation*}
This inequality holds because if at least $t$ variables in ${\bf X}$ exceed $\tau$, then at least $t$ shards must also exceed $\tau$. However, the converse may not always be true due to the possibility of multiple shards from the same $X_i$ surpassing $\tau$. Nonetheless, the highest value shard in $\bf S$ corresponds to a real value from some $X_i$. We also note that $\Prob{\cardin{\tau \leq {\bf X} < \infty} \geq 1} = \Prob{\cardin{\tau \leq {\bf S} < \infty} \geq 1}$; if any shard is above $\tau$, then at least one $X_i \geq \tau$, and vice versa.

Define $H_{i,j} = 1$ if and only if $Y_{i,j} \geq \tau$. The distribution of $\sum_{j} H_{i,j}$ follows a binomial distribution $\Bin(K, p_i)$, where $p_i = 1 - \Prob{X_i \leq \tau}^{1/K}$. As $K \to \infty$, $p_i$ approaches zero, allowing for a Poisson approximation with rate $Kp_i = K(1 - \Prob{X_i \leq \tau}^{1/K})$. This rate converges to $-\log(\Prob{X_i \leq \tau})$ as $K \to \infty$. Therefore, the sum $U_{\tau} = \sum_i \sum_{j} H_{i, j}$ can be approximated by a Poisson distribution with rate
\begin{equation*}
    \lambda_{\tau} = \sum_{i=1}^n -\log(\Prob{X_i \leq \tau}) = -\log\left(\prod_{1 \leq i \leq n} \Prob{X_i \leq \tau}\right) = -\log(\Prob{Z \leq \tau}),
\end{equation*}
where $Z = \max(X_1, ..., X_n)$. Going backwards, by setting a threshold $\tau$ such that 
\begin{equation*}
    \lim_{K \to \infty} \left( \sum_{i=1}^n \sum_{j=1}^K \Prob{Y_{i, j} \geq \tau} \right) = q,
\end{equation*}
we find that $\Prob{Z \leq \tau} = e^{-q}$.

Applying the same technique with a larger threshold $\beta > \tau$ results in a similar Poisson random variable $U_\beta$, which counts the shards exceeding $\beta$, but with a smaller Poisson rate $\lambda_\beta$. The difference $U_{\tau, \beta} = U_\tau - U_\beta$ represents the count of shards within the interval $[\tau, \beta]$, and follows a Poisson distribution with rate $\lambda_\tau - \lambda_\beta$. Specifically, as $K \to \infty$, we have $\Prob{\cardin{\tau \leq {\bf S} \leq \beta} = t} = \Prob{U_{\tau, \beta} = t}$. Crucially, in the limit as $K \to \infty$, $U_{\tau, \beta}$ and $U_\beta$ become independent, a property we shall establish through coupling.

Stochastic dominance, or majorization, forms the last piece of the puzzle of our framework. This is a quite well known tool for bounding competitive ratios. For any algorithm $\Alg$ for any variant of the prophet inequality, if a constant $c\in [0, 1]$ exists such that $\Prob{\Alg \geq x} \geq c \Prob{Z \geq x}$ for all $x \geq 0$, then majorization asserts $c$ as a lower bound on $\Alg$'s competitive ratio. By selecting a threshold $\tau$ satisfying $\Prob{Z \leq \tau} = \alpha$, it follows trivially\footnote{Given the assumption of continuous random variables, there can be no point masses on $\tau$.} that $\Prob{Z \geq \tau} = 1 - \alpha$. The objective then becomes to establish a lower bound for $\Prob{\Alg \geq \tau}$ using $\alpha$. Using the sharding framework, the count of shards above $\tau$ is modeled by a Poisson distribution with rate $-\log(\alpha)$. At this point, the application of our framework diverges based on the specific problem at hand. For each problem, we give an event $\xi_\alpha$ on the \textbf{shards} that implies that $\Alg$ running on $X_1, ..., X_n$ receives a reward with a value at least $\tau$. Thus, $\Prob{\Alg \geq \tau} \geq \Prob{\xi_\alpha}$. We emphasize this is only done for the sake of the analysis of the algorithm; at no point are we actually running the algorithm on the shards. 

\paragraph{Simple Application.} To illustrate the framework, consider the Samuel-Cahn algorithm for the standard prophet inequality, which sets a threshold $\tau$ satisfying $\Prob{Z \leq \tau} = 1/2$ and accepts the first value (if any) from $X_1, ..., X_n$ exceeding $\tau$. We briefly demonstrate its $1/2$ competitiveness using our framework. Our new proof is not significantly simpler in this case than the previous proof, we just provide it as an example of our framework.

\begin{lemma}
    The Samuel-Cahn algorithm is $1/2$ competitive.
\end{lemma}
\begin{proof}
    Sharding the variables $X_1, ..., X_n$ into $Y_{i,1}, ..., Y_{n,K}$, the framework implies $\sum_{i=1}^n \sum_{j=1}^K \Prob{Y_{ij} \geq \tau} = \lambda_\tau = -\log(1/2) = \log(2)$ as $K \to \infty$. Thus, the number of shards exceeding $\tau$ follows a Poisson distribution with rate $\log(2)$. We employ stochastic dominance to compare $\Prob{\Alg \geq \ell}$ and $\Prob{Z \geq \ell}$, which depends on $\ell$'s value.

\paragraph{Case 1: $\ell \in [0, \tau]$.} The algorithm accepts a value $\geq \ell$ iff at least one shard exceeds $\tau \geq \ell$, corresponding to an actual $X_i$ realization. Therefore,
\begin{equation}
    \Prob{\Alg \geq \ell} = 1-e^{-\lambda_\tau}\frac{\lambda_\tau^0}{0!}= 1 - e^{-\log(2)} = \frac{1}{2} \geq \frac{1}{2} \Prob{Z \geq \ell}. \eqlab{half:first}
\end{equation}

\paragraph{Case 2: $\ell \in (\tau, +\infty)$.} Define 
\begin{equation*}
    \lambda_\ell = \lim_{K \to \infty} \sum_{i=1}^n \sum_{j=1}^K \Prob{Y_{i,j}\geq \ell}.
\end{equation*}
Given $\ell > \tau$, it follows $\lambda_\ell \leq \lambda_\tau$. The probability $\Prob{Z \geq \ell} = 1 - e^{-\lambda_\ell}$, since at least one shard must be above $\ell$ for $Z \geq \ell$. To lower bound $\Prob{\Alg \geq \ell}$, consider an event implying $\Alg \geq \ell$: no shards with value in $[\tau, \ell]$ and at least one shard with value exceeding $\ell$. This event guarantees $\Alg \geq \ell$ as at least one $X_i \geq \ell$, and no $X_j$ within $[\tau, \ell]$ prevents us from choosing such $X_i$ (as there are no shards in $[\tau, \ell]$). The Poisson rate for shards in $[\tau, \ell]$ is $\lambda_\tau - \lambda_\ell$, and for shards exceeding $\ell$, it's $\lambda_\ell$. Thus, 
\begin{equation}
        \Prob{\Alg \geq \ell} \geq \underbrace{e^{-(\lambda_\tau - \lambda_\ell)}}_{\text{no shards in }[\tau, \ell]} \cdot \underbrace{(1-e^{-\lambda_\ell})}_{\text{at least one shard in }[\ell, \infty)} = e^{-(\lambda_\tau - \lambda_\ell)} \Prob{Z \geq \ell} \geq e^{-\lambda_\tau}\Prob{Z \geq \ell} = \frac{1}{2}\Prob{Z \geq \ell}.\eqlab{half:second}
\end{equation}
By combining \Eqref{half:first} and \Eqref{half:second} through stochastic dominance, the algorithm is $1/2$ competitive. 
\end{proof}

This example serves as a primer to our framework. Different problems will have different events on the shards of varying complexity that imply $\Alg \geq \ell$. The simplicity of this framework belies its strength. Despite the simplicity, \textit{it improves the competitive ratio of more than a dozen prophet inequalities that individually had different and specialized analysis}. Moreover, for several problems, our framework establishes tight competitive ratios within constrained algorithm classes, like single thresholds algorithms. The framework also significantly simplifies proofs for known results in the literature, unifying them and making them more accessible. 

\paragraph{New results exceeds parameter optimization.} Many of the results in this field work in two steps. The first step is deriving some parametric formula for the competitive-ratio, which is typically problem specific. The second step then involves optimizing the parameters to obtain the best (i.e., highest) possible lower bound on the value of the function. Unfortunately, the optimization part can be quite tedious and technical, and in most of the cases, no analytical closed form solutions exist for the maximizer. Hence, numerical solvers are often used to find a set of parameters that are ``good enough''. It is of course plausible that such parameters are suboptimal, and that a ``better'' optimizer would find a slightly better solution, with a better competitive ratio. 

Our main contribution is a new way to perform the first part of the
above analysis; deriving the actual parametric formula. While we still have to dabble in some parameter optimization to derive our bounds, this is neither our main contribution, nor was a major thrust of the work. Furthermore, without the new ideas, no parametric optimization can lead to our main improved results.

\paragraph{Organization} \secref{recap} introduces notation, assumptions, and recaps existing results and techniques. \secref{warmupiid} is a warmup section that uses the ideas of Poissonization in re-deriving the classical $\approx 0.745$ prophet inequality for \IID random variables. \secref{sharding} is yet another much needed warmup section that introduces the idea of sharding, and reproves several known results in the literature using our technique. \secref{best1ofk}  presents our first new major result, giving the improved analysis for the \textsc{Top-$1$-of-$k$} model both for \IID and Non-\IID random variables. \secref{main} presents our second major result, giving the improved analysis for the Non-\IID prophet secretary. \secref{iidsemionline} gives the improved $0.89$ competitive algorithm for the \IID \textsc{Semi-Online} problem using discrete clocks. \secref{loadminimization} introduces the $O(\log^\ast n)$ load result for the \textsc{Semi-Online-Load-Minimization} problem.  Finally, we add concluding remarks and potential future work directions in \secref{conclusion}. 

\section{Notation, assumptions, and existing results.}
\seclab{poissonization}
\seclab{recap}

\paragraph{Notation} In contexts where the dimension $k$ is clear, $e_i$ represents the $i$-th standard basis vector in $\mathbb{R}^k$, characterized by zeroes in all coordinates except for a $1$ in the $i$-th position. $\mathbb{S}_n$ refers to the symmetric group of $n$ elements. The notation $[n]$ signifies the set $\{1, \ldots, n\}$, and $\log^{(t)}n$ denotes the $t$-fold iterated logarithm function. For example, $\log^{(2)}n$ represents $\log \log n$. The iterated logarithm function $\log^\ast(n)$ is defined as the smallest integer $t$ for which $\log^{(t)}n \leq 1$.

\paragraph{Continuity Assumptions} Consider $X_1, \ldots, X_n$ as independent, \textit{non-negative} random variables. In this paper, $Z = \max(X_1, \ldots, X_n)$ is used to denote the maximum value among these $n$ random variables. The notation $\Alg$ is used to represent the algorithm's reward, albeit it is sometimes used interchangeably to refer to the algorithm itself for convenience.

\begin{assumption}
For all problems under consideration, it is assumed, without loss of generality, that the random variables $X_1, \ldots, X_n$ are continuous.
\end{assumption}

Refer to \cite{csz-pstbs-18} for a rationale on why this assumption, enabled by stochastic tie-breaking, does not lose generality.

\paragraph{Folklore set up for prophet secretary} In the Prophet Secretary problem, a random permutation $\sigma$ from the symmetric group $\mathbb{S}_n$ is chosen uniformly at random. The values are revealed to a gambler in the sequence $X_{\sigma(1)}, \ldots, X_{\sigma(n)}$. At each iteration $t$, the gambler is presented with the value $X_{\sigma(t)}$ and must decide whether to accept this value as their final reward, thereby concluding the game, or to irrevocably reject $X_{\sigma(t)}$ in favor of proceeding to the next iteration $t+1$. Should the gambler fail to select a value by the conclusion of round $n$, their reward defaults to zero.

An alternative "folklore" version within the community exists for the prophet secretary problem.\footnote{If the reader is aware of pertinent references, the author would be grateful for the information.} This paper adopts this version, which is included here for completion sake. In this version, each random variable $X_i$ samples a value $v_i$ from its distribution and is assigned a "time of arrival" $t_i$, selected uniformly at random from the interval $[0,1]$. Denoting $\pi$ as the permutation satisfying $t_{\pi(1)} \leq \ldots \leq t_{\pi(n)}$, the values are then presented in the sequence $v_{\pi(1)}, \ldots, v_{\pi(n)}$, ordered according to their times of arrival. Given that any permutation $\pi$ of the arrival order of $X_1, \ldots, X_n$ occurs with a probability $1/n!$, this scheme is equivalent to drawing a random permutation.

A subtle technical point arises in that the algorithm remains unaware of the chosen times of arrival in this setup; it is only informed of the values of the realizations. However, the algorithm can simulate the scheme by generating $n$ independent random times of arrival $t_1, \ldots, t_n \sim \text{Uniform}(0,1)$. After sorting these times of arrival such that $a_1 \leq \ldots \leq a_n$, the algorithm maps the $i$-th realization it processes to the time of arrival $a_i$. Defining $T_i$ as a random variable representing the time of arrival for $X_i$, we claim that this simulation mimics the scheme where each random variable independently selects its time of arrival $t_i$.

\begin{lemma}
\lemlab{folklore}
\lemlab{pairwise_indep}
    For any variable $X_i$, let $t_i\sim \Uniform(0, 1)$, and $T_i$ be the time of arrival in the simulated process. For any $t \in [0,1]$, we have $\Prob{t_i\leq t}=\Prob{T_i \leq t}=t$. In addition, $T_1, ..., T_n$ are independent. 
\end{lemma}
\begin{proof}
    See \apdxref{A}. 
\end{proof}
\begin{assumption}
    We assume without loss of generality that the algorithm for the prophet secretary has access to the time of arrival of a realization drawn uniformly and independently at random from the interval $[0,1]$. 
\end{assumption}

\paragraph{Types of Thresholds} Threshold-based algorithms work by establishing a series of thresholds $\tau_1, \ldots, \tau_n$, often set in a descending order. A realization $v_i$ is accepted if and only if $v_i \geq \tau_i$, and all preceding realizations $v_1, \ldots, v_{i-1}$ fall below their respective thresholds $\tau_1, \ldots, \tau_{i-1}$. This means that $v_i$ is the first realization to surpass its threshold.

Two primary forms of thresholding methods are prevalent in the literature. The first method, known as \textit{maximum based thresholding}, involves setting each $\tau_i$ to correspond with the $q_i$-quantile of the distribution of $Z = \max_i X_i$, such that the probability $\Prob{Z \leq \tau_i}$ equals $q_i$. These $q_i$ values are carefully selected, often arranged in a non-increasing sequence. Samuel Cahn \cite{s-ctsrm-84} sets a single threshold $\tau=\tau_1 = \ldots = \tau_n$ such that $\Prob{Z\leq \tau}=1/2$ (i.e., the median of $Z$). Since then, several results have adopted this idea, including the result of Correa \etal on discrete blind strategies \cite{csz-pstbs-18}.

Similarly, \textit{summation based thresholding} sets a threshold $\tau$ such that the expected number of realizations at least $\tau$ sums to $s_i$ (i.e., $\sum_{i=1}^n \mathbb{P}(X_i \geq \tau) = s_i$). One paper that uses a variation of this idea is the work of \cite{ethlm-psmpc-19}.

One of the key contributions of this paper is relating these two kinds of thresholding techniques via Poissonization and sharding. In practice, these are not necessarily the only two types of threshold setting techniques that can work. For example, one can certainly set thresholds such that (say) $\sum_{i=1}^n \Prob{X_i \geq \tau}^2 = q_i$. However, theoretical analysis of such techniques are highly non-trivial as one often needs to bound both $\Prob{Z\geq \tau}$ and $\Prob{\Alg \geq \tau}$.  With maximum based thresholding, often the bound on $\Prob{Z\geq \tau}$ is trivial, because we choose $\tau$ as a quantile of the maximum, but bounding $\Prob{\Alg \geq \tau}$ is more cumbersome. On the other hand, summation based thresholding typically have simpler analysis for $\Prob{\Alg \geq \tau}$, but bounding $\Prob{Z\geq \tau}$ is harder and is distribution specific. 

\paragraph{Standard Stochastic Dominance/Majorization Argument} To establish a lower bound for the competitive ratio of a thresholding algorithm using a descending sequence of thresholds $\tau_1 > \ldots > \tau_n$, a common approach involves the concept of majorization or stochastic dominance. We outline the approach below. Consider the expected values represented by the integrals:
\begin{equation*}
    \mathbb{E}[\Alg] = \int_{0}^{\infty} \mathbb{P}(\Alg \geq \ell) \dif \ell,~~~~~~\mathbb{E}[Z] = \int_{0}^{\infty} \mathbb{P}(Z \geq \ell) \dif \ell.
\end{equation*}
By setting $\tau_0 = \infty$ and $\tau_{n+1} = 0$, if we can guarantee that for every $\nu$ in the interval $[\tau_{i}, \tau_{i-1}]$, there exists a constant $c_i \in [0,1]$ such that $\Prob{\Alg \geq \nu} \geq c_i \Prob{Z \geq \nu} $, then the following inequality can be derived
\[
\mathbb{E}[\Alg] = \sum_{i=1}^{n+1} \int_{\tau_{i}}^{\tau_{i-1}} \Prob{\Alg \geq \ell} \, \dif \ell \geq \sum_{i=1}^{n+1} c_i \int_{\tau_{i}}^{\tau_{i-1}} \Prob{Z \geq \ell} \, \dif \ell \geq \min(c_1, \ldots, c_{n+1})\mathbb{E}[Z].
\]
Consequently, the competitive ratio of $\Alg$ is lower-bounded by $c = \min(c_1, \ldots, c_{n+1})$. This technique is a cornerstone of various lower bounds on prophet inequalities, including our own, and is commonly referred to as \textit{majorizing} $\Alg$ by $Z$. This technique is helpful as it simplifies the process of lower bounding the 
competitive ratio by instead comparing $\Prob{\Alg \geq \ell} $ versus $\Prob{Z \geq \ell} $ within a fixed interval, rather than directly dealing with the expectations as a whole.

\paragraph{Recap of Discrete Blind Strategies} The concept of discrete blind strategies, as introduced by Correa \etal \cite{csz-pstbs-18}, employs maximum based thresholding for the \textit{prophet secretary} problem. Initially, the algorithm chooses a decreasing function $\alpha: [0,1] \to [0,1]$. With $\ICDFY{Z}{q}$ denoting the threshold for which $\Prob{Z \leq \ICDFY{Z}{q}} = q$, the algorithm commits to the first realization $v_i$ satisfying $v_i \geq \ICDFY{Z}{\alpha(i/n)}$ (i.e., if $v_i$ is within the top $\alpha(i/n)$ percentile of $Z$). Defining $T$ as a random variable representing the time a realization is chosen (if any), Correa \etal derive the crucial inequality for any $k \in [n]$ \cite{csz-pstbs-18}
\[
\frac{1}{n} \sum_{i=1}^k \left(1-\alpha\left(\frac{i}{n}\right)\right)  \leq \Prob{T \leq k} \leq  1-\left(\prod_{i=1}^{k} \alpha\left(\frac{i}{n}\right)\right)^{1/n}.
\]
The proof of this inequality is involved, utilizing principles of Schur-convexity an infinite number of times to establish the upper bound, and $n$ times for the lower bound. In \secref{sharding} we present a straightforward and elementary proof of the above inequalities, along with even more tighter bounds.

Next, they use these bounds on $\Prob{T \leq k}$ to deduce a lower bound for $\Prob{\Alg \geq \ICDFY{Z}{\alpha(i/n)}}$. Coupled with the straightforward relation $\Prob{Z \geq \ICDFY{Z}{\alpha(i/n)}} = 1 - \alpha(i/n)$, this enables them to majorize blind strategies against $Z$, thereby establishing a lower bound on the competitive ratio in relation to $\alpha$ as $n\to \infty$. By optimizing across various $\alpha$ functions, they achieve a competitive ratio of approximately $0.669$. For further details, refer to \cite{csz-pstbs-18}.

\paragraph{Probability Background} In the context of a measurable space $(\Omega, \mathcal{F})$ equipped with probability measures $P, Q$, the \textit{total variational distance} between $P$ and $Q$ is defined as
\[
d(P, Q) = \frac{1}{2} \lVert P-Q \rVert_1 = \sup_{A \in \mathcal{F}} \lvert P(A) - Q(A) \rvert.
\]
A random variable $X \in \mathbb{R}^k$ is termed \textit{categorical}, parameterized by success probabilities $p \in \mathbb{R}^k$, if $X$ can assume values in $\{{\bf 0}, e_1, \ldots, e_k\}$\footnote{As a reminder, $e_i$ is the $i$-th standard basis vector in $\mathbb{R}^k$} where $\mathbb{P}(X=e_i) = p_i$ for $i = 1, \ldots, k$, and $\mathbb{P}(X={\bf 0}) = 1 - \sum_i p_i$. The Poisson distribution, denoted by $\Poisson(\lambda)$, is characterized by a rate parameter $\lambda$, with a variable $X \sim \Poisson(\lambda)$ taking values in $\mathbb{N}_{\geq 0}$ and having $\Prob{X=k} = e^{-\lambda} \frac{\lambda^k}{k!}$. A \textit{multinomial Poisson distribution}, denoted $\Poisson(\lambda_1, \ldots, \lambda_k)$, is understood as a $k$-dimensional random variable where each coordinate is an \textbf{independent} Poisson random variable, with $X \sim \Poisson(\lambda_1, \ldots, \lambda_k)$ in $\mathbb{N}_{\geq 0}^k$ satisfying $\mathbb{P}(X = (n_1, \ldots, n_k)) = \prod_{i=1}^k e^{-\lambda_i} \frac{\lambda_i^{n_i}}{n_i!}$.

\paragraph{Poissonization via Coupling} Coupling is a powerful technique  for estimating the variational distance between two random variables. Generally, to bound the variational distance between variables $X, Y$, it suffices to construct a joint random vector $W$ whose marginals are precisely $X$ and $Y$.

The result required here concerns the coupling of multi-dimensional random variables, an extension of the single-dimensional case known as Le Cam's theorem \cite{c-atpbd-60}, with the needed higher-dimensional generalizations found in \cite{w-cma-86}. The proof, standard in the coupling literature \cite{h-ptcm-12}, is reformulated below to follow our notation.

\begin{lemma}{\cite{w-cma-86}}
\lemlab{distance}
    Given $n$ independent categorical random variables $Y_1, \ldots, Y_n$, each parameterized by $p_1, \ldots, p_n \in \mathbb{R}^k$, and defining $S_n = \sum_{i=1}^n Y_i$ with $\lambda = \sum_i p_i$, let $T_n \sim \Poisson(\lambda_1, \ldots, \lambda_k)$. Denoting $\hat{p}_i = \sum_{j=1}^k p_{i,j}$, it follows that
    \[
    d(S_n, T_n) \leq 2\sum_{i=1}^n \hat{p}_i^2.
    \]
\end{lemma}

\section{Poissonization Warmup: The \IID Prophet Inequality.}\label{eta10}
\seclab{warmupiid}
In this section, we focus on the classical prophet inequality for \IID random variables, for which there exists an algorithm achieving a competitive ratio of $\approx 0.745$. This discussion aims to lay the groundwork for understanding Poissonization. The problem is exactly  the standard \IID  prophet inequality, as the random permutation of \IID variables does not alter the problem's nature. We proceed under the assumption that $n \to \infty$. This assumption is not required but serves to simplify the exposition here. Hence, our input is $(t_i, v_i)$, for $i=1, ..., n$, where $t_i \sim \Uniform(0, 1)$ is the time of arrival, and $v_i$ is the value of the $i$-th random variable $X_i$. 

\paragraph{Canonical Boxes}
\begin{figure}
    \centering
    \includegraphics[width=0.4\textwidth]{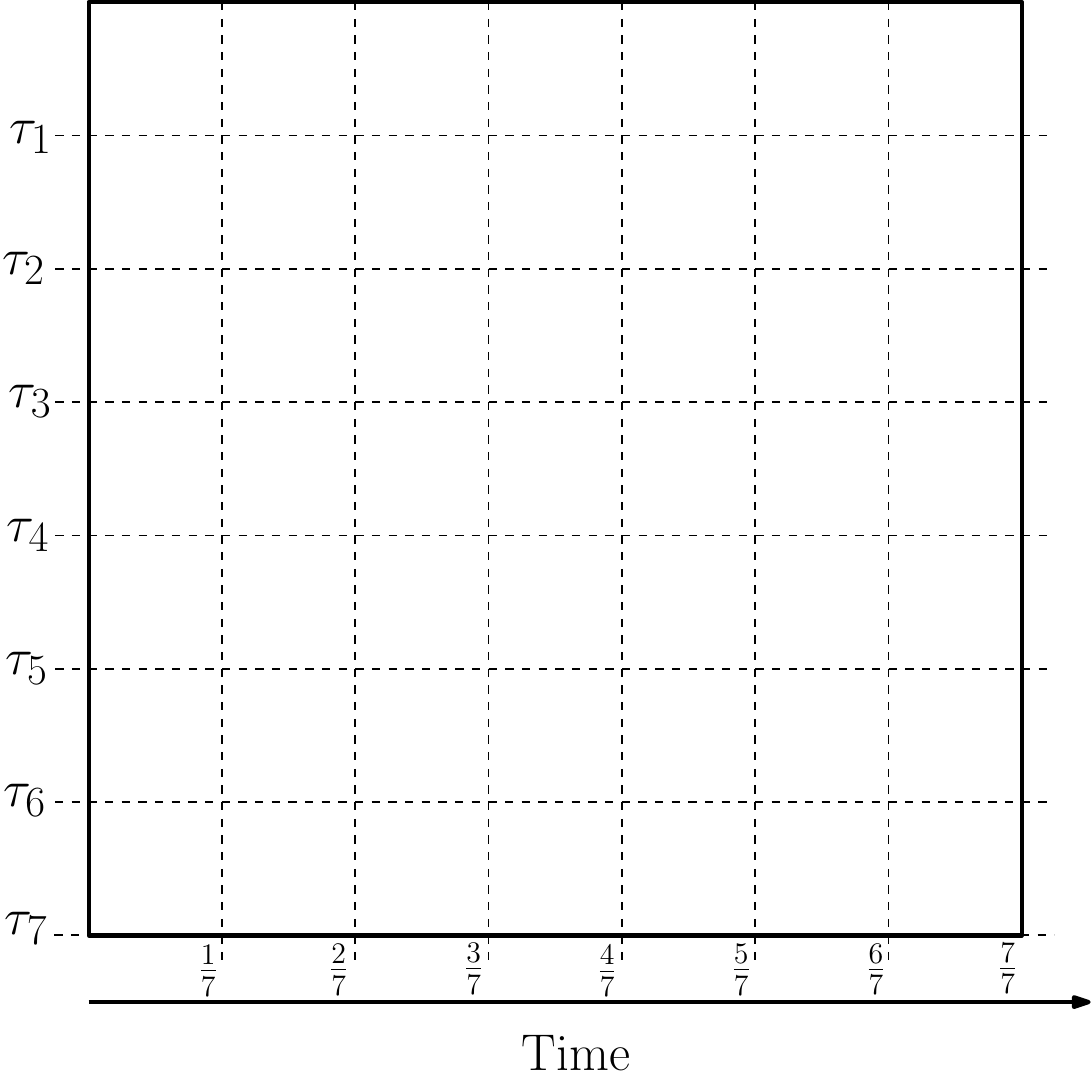}
    \caption{Level $7$ canonical boxes of $\tau_7 = \Xi(q)$}
    \figlab{fig:k7poissonization}
\end{figure}
Because the random variables are continuous, for any $q \in [0,n]$, there exists a threshold $\tau$ satisfying $\sum_{i=1}^n \Prob{X_i \geq \tau} = q$, as guaranteed by the intermediate value theorem.
\begin{definition}
When the random variables are implicitly clear, we use $ \Xi(q)$ to denote the threshold such that on expectation, $q$ realizations are above it.  
\end{definition}

In the subsequent analysis, think of $k\to \infty$ and $q=O(1)$ as a constant to be determined. We fix a threshold $\Xi(q)$ and break ``arrival time'' into $k$ buckets, the $i$-th between time $\frac{i-1}{k}$ and $\frac{i}{k}$. In addition, we define $k+1$ thresholds $\tau_0, \tau_1, \ldots, \tau_k$ such that $\tau_i = \Xi(\frac{q\cdot i}{k})$ (with $\Xi(0)=\infty$). 

\begin{definition}
    The \textit{level $k$ canonical boxes} of $\Xi(q)$ are defined as the $k^2$ sets $\square_{i,j} = \{(t, v) \mid \frac{i-1}{k} \leq t \leq \frac{i}{k} \text{ and } \tau_j \leq v \leq \tau_{j-1}\}$. Refer to  \figref{fig:k7poissonization}.
\end{definition}

Assuming the random variables arrival times are $t_1, ..., t_n$ and their values are $v_1, ..., v_n$,
\begin{definition}
    A realization $v_i$ is said to \textit{arrive} or \textit{fall} in $\square_{r, s}$ if $(t_i, v_i) \in \square_{r, s}$.
\end{definition}

Our objective is to derive a succinct, closed-form expression for $S \in \mathbb{R}^{k \times k}$, where $S_{i,j}$ represents the count of realizations arriving in $\square_{i,j}$. This will be achieved by coupling the distribution with a multinomial Poisson distribution $T \in \mathbb{R}^{k \times k}$, which mimics $S$ as $n, k \to \infty$ (i.e., $\lVert S-T \rVert_1 \to 0$ as $n, k \to \infty$).

\begin{lemma}
\lemlab{replace:poisson}
    Fix $q=O(1)$ and consider the level-$k$ canonical boxes of $\Xi(q)$. Let $S_n \in \Re^{k\times k}$ count the number of realizations in the canonical boxes $\{\square_{i,j}\}$. Let $T_n \in \Re^{k\times k}$ be a multinomial Poisson random variable with each coordinate rate being $\frac{q}{k^2}$. Then 
    \[
    d(S_n, T_n) \leq \frac{2q^2}{n}.
    \]
    In particular, as $k,n\to \infty$, then for any (simple) region $\circledcirc \subseteq [0,1]\times [\Xi(q), \infty]$, the probability we have $r$ realizations fall in $\circledcirc$ is $e^{-\mu{(\circledcirc)}} \frac{\mu{(\circledcirc)}^r}{r!}$ where $\mu{(\circledcirc)}=\sum_{i=1}^n \Prob{X_i \text{ arrives in } \circledcirc}$
\end{lemma}
\begin{proof}
    See \apdxref{A}.
\end{proof}

\begin{remark}
    The proof of  \lemref{replace:poisson} is extendable to non-\IID random variables under the condition that each $\Prob{ X_i \geq \tau_k}$ is sufficiently ``small''. This condition is typical in the proofs of coupling results, such as Le Cam's theorem. For instance, if $\Prob{X_i \geq \tau_k} \leq 1/K$ for some $K \to \infty$, then the variational distance also approaches zero. The proof is mostly the same as previously described.
\end{remark}

\paragraph{Plan of Attack} With the aid of \lemref{replace:poisson}, and by letting $k, n \to \infty$, we find that for any region $\circledcirc$ above $\Xi(q)$, the probability of observing $j$ realizations within this region is given by $e^{-\mu(\circledcirc)} \frac{\mu(\circledcirc)^j}{j!}$, where $\mu(\circledcirc)$ denotes the area (measure) of the region $\circledcirc$. This simplification paves the way to express the competitive ratio of an algorithm in terms of an integral.

\paragraph{Algorithm} We consider algorithms described by an \textbf{increasing} curve $C:[0,1]\to \Re_{\geq 0}$ with $C(1)\leq q = O(1)$. At a given time $t_i$, a realization $(t_i, v_i)$ is accepted if and only if $v_i \geq \Xi(C(t_i)) = \tau_C(t_i)$, meaning that the threshold $\tau_C(t)$ at time $t$ is set so that the expected number of realizations above it is equal to $C(t)$. Given such a function $C$, how do we find the competitive ratio of an algorithm following $\tau_C$?

\begin{figure}
    \centering
    \includegraphics[width=0.85\textwidth]{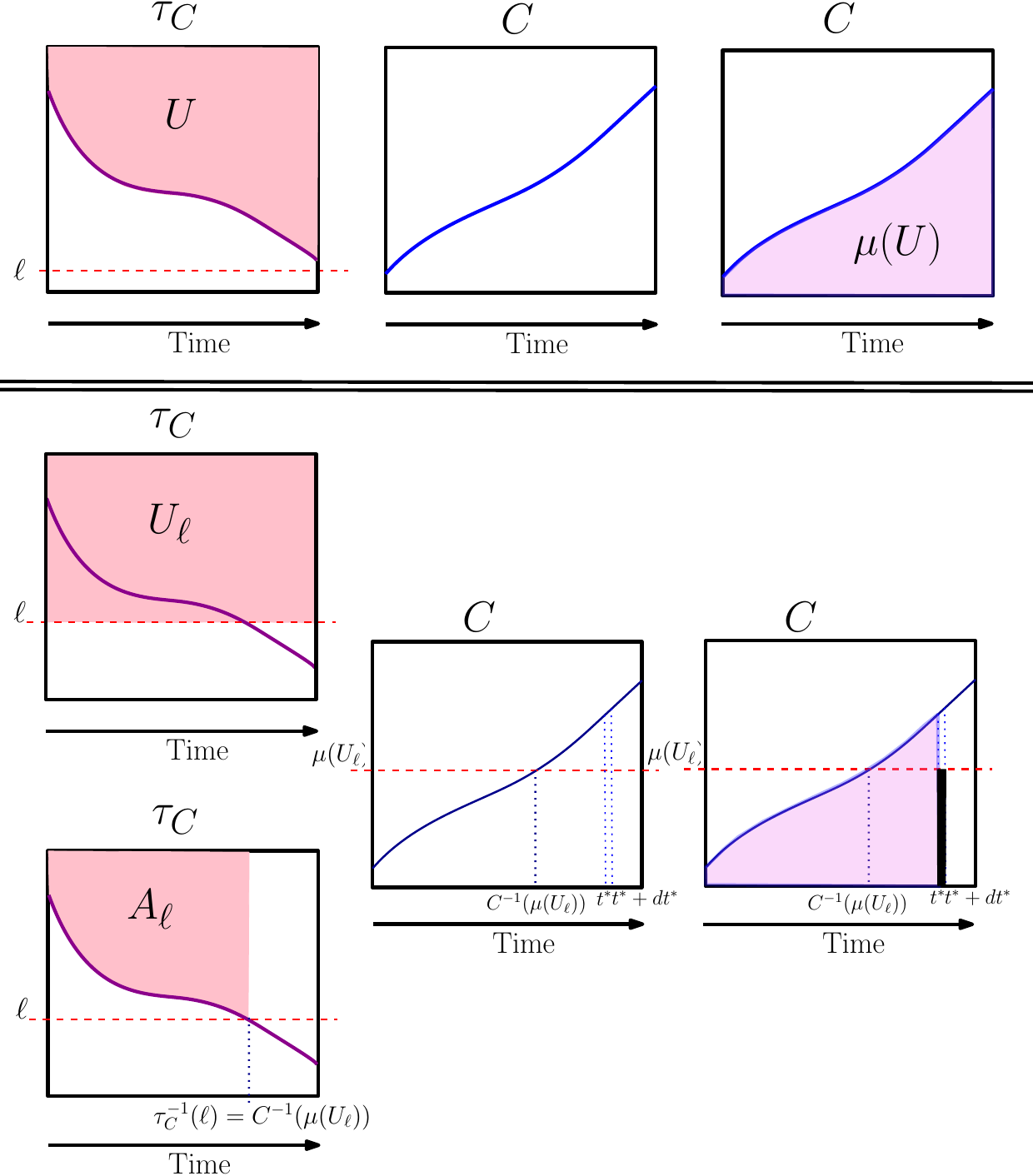}
    \caption{The two cases of \lemref{iid:monstrosity}. Curves are labeled either as $C$ or $\tau_C$.  }
    \figlab{fig:poissinization-iid-1}
\end{figure}

\begin{lemma}
\lemlab{iid:monstrosity}
    The competitive ratio $c$ of the algorithm that follows curve $C:[0,1]\to \Re_{\geq 0}$ satisfies
    \begin{equation}
            c \geq \min \left\{ 1-e^{-\int_0^1 C(t) \dif t}, \inf_{0 < \ell' \leq C(1)} \left( \frac{ 1-e^{-\int_{0}^{C^{-1}(\ell')} C(t) \dif t} + \int_{C^{-1}(\ell')}^1 \ell' e^{-\int_{0}^t C(y) \, \dif y} \, \dif t}{1-e^{-\ell'}} \right) \right\}. \eqlab{monstrosity}
    \end{equation}
\end{lemma}
\begin{proof}
    Throughout this proof, refer to \figref{fig:poissinization-iid-1} for a visual aid. Recall that $C$ is an increasing function. For convenience, we extend the notation such that $C^{-1}(\ell') = 1$ for any $\ell' > C(1)$, and $C^{-1}(\ell') = 0$ for $\ell' < C(0)$. Let $\Alg$ denote the outcome of the strategy that adheres to the threshold $\tau_C$. We proceed by stochastic dominance. 

\paragraph{Case 1: $\ell \in [0, \tau_C(1)]$.} For what follows, see the first row of  \figref{fig:poissinization-iid-1}.  We establish a trivial \textit{upper bound} for $\Prob{Z \geq \ell}$ as $\Prob{Z \geq \ell} \leq 1$. 

Define the set $U = \{(t,v) \mid 0 \leq t \leq 1, \tau_C(t) \leq v\}$. The measure of $U$, denoted $\mu(U)$, can be determined by integrating over the curve $C$. 
\begin{align*}
    \mu(U) &= \sum_{i=1}^n \Prob{X_i \text{ arrives in }U} \\ 
    & = \sum_{i=1}^n \int_0^1 \Prob{X_i \geq \tau_C(t)} \dif t  \\
    & = \int_0^1 \sum_{i=1}^n \Prob{X_i \geq \tau_C(t)} \dif t \\ 
    & = \int_0^1 C(t) \dif t. 
\end{align*}
Thus, the probability that $\Alg$ exceeds $\ell$ is given by
\begin{align*}
        \Prob{\Alg \geq \ell} &= 1-\Prob{U \text{ has no arrivals}} \\
        & =  1 - e^{-\mu(U)} \\
        & = 1- e^{-\int_0^1 C(t) \dif t}. 
\end{align*}
Combining this with the upper bound on $\Prob{Z\geq \ell}$, we have the first main inequality 
\begin{equation}
    \Prob{\Alg \geq \ell} \geq (1-e^{-\int_0^1 C(t) \dif t }) \Prob{Z \geq \ell}
    \eqlab{iid:case1:bound}. 
\end{equation}

\paragraph{Case 2: $\ell \in [\tau_C(1), \infty)$. } For what follows, see the second row of  \figref{fig:poissinization-iid-1}.  Consider $U_\ell=\{(t, v)|0\leq t \leq 1, \ell \leq v < \infty\}$. We can compute the measure of $U_\ell$ as before 
\begin{align*}
    \mu(U_\ell) &= \sum_{i=1}^n \Prob{X_i \text{ arrives in } U_\ell}\\
    &= \sum_{i=1}^n \Prob{X_i \geq \ell} .
\end{align*}
Thus, it follows
\[
\Prob{Z \geq \ell} = 1-\Prob{U_\ell \text{ has no arrivals}}= 1-e^{-\mu(U_\ell)}.
\]

Next, we lower bound $\Prob{\Alg \geq \ell}$. Consider the region $A_\ell = \{(t, v) \mid 0 \leq t \leq C^{-1}(\ell), \tau_C(t) \leq v\}$. The measure of $A_\ell$, $\mu(A_\ell)$, is calculated by integrating over the curve $C$ up to $\tau_C^{-1}(\ell)$ as follows. 

\begin{align*}
    \mu(A_\ell)&=\sum_{i=1}^n \Prob{X_i \text{ falls in }A_\ell} \\
    &= \sum_{i=1}^n \int_{0}^{\tau_C^{-1}(\ell)}\Prob{X_i \geq \tau_C(t)} \dif t \\
    &= \sum_{i=1}^n \int_{0}^{C^{-1}(\sum_{i=1}^n \Prob{X_i \geq \ell})}\Prob{X_i \geq \tau_C(t)} \dif t \\
        &= \sum_{i=1}^n \int_{0}^{C^{-1}(\mu(U_\ell))}\Prob{X_i \geq \tau_C(t)} \dif t \\
&= \int_0^{C^{-1}(\mu(U_\ell))} \sum_{i=1}^n \Prob{X_i \geq \tau_C(t)} \dif t \\ 
& = \int_0^{C^{-1}(\mu(U_\ell))} C(t) \dif t. 
\end{align*}

Consider the time $t^\ast$ where the algorithm accepts a value. First, consider if $t^\ast \in [0, C^{-1}(\mu(U_\ell))]$. In that case, the algorithm accepts a value above $\ell$ if and only if the region $A_\ell$ is non-empty (i.e., contains a realization). This happens with probability 
\begin{equation}
    1-e^{-\mu(A_\ell)} = 1 - e^{-\int_{0}^{C^{-1}(\mu(U_\ell))} C(t) \dif t} \eqlab{first:of:2:iid:proof}  .
\end{equation}

On the other hand, if $t^\ast \in [C^{-1}(\mu(U_\ell)), 1]$, the algorithm accepts a value above $\ell$ at time $t^\ast$ if and only if the area from time $0$ to time $t^\ast$ above $\tau(t)$ is empty (no realizations arrive in that region) and the interval $[t^\ast, t^\ast + \dif t^\ast]$ witnesses a realization above $\ell$. This probability for this happening is
\begin{equation}
   e^{-\int_{0}^{t^\ast} C(t) \dif t} \mu(U_\ell) \dif t^\ast. \eqlab{second:of:2:iid:proof}
\end{equation}
Combining \Eqref{first:of:2:iid:proof} and integrating \Eqref{second:of:2:iid:proof} over $t^\ast \in [C^{-1}(\mu(U_\ell)), 1]$, it follows correspondingly that the probability that $\Alg$ exceeds $\ell$ is expressed as
\begin{equation*}
    \Prob{\Alg \geq \ell} = 1 - e^{-\int_{0}^{C^{-1}(\mu(U_\ell))} C(t) \, \dif t} + \int_{C^{-1}(\mu(U_\ell))}^1 \mu(U_\ell) e^{-\int_{0}^{t^\ast} C(y) \dif y} \dif t^\ast.
\end{equation*}

Combining this with the value of $\Prob{Z \geq \ell}$, we have proved the inequality. 

\begin{equation*}
    \Prob{\Alg \geq \ell} \geq  \frac{ \left( 1 - e^{-\int_{0}^{C^{-1}(\mu(U_\ell))} C(t) \, \dif t} + \int_{C^{-1}(\mu(U_\ell ))}^1 \mu(U_\ell ) e^{-\int_{0}^{t^\ast} C(y) \dif y} \dif t^\ast \right) }{1-e^{-\mu(U_\ell)}} \Prob{Z \geq \ell}.
\end{equation*}
Letting $\ell' = \mu(U_\ell) \in (0, C(1)]$, this rewrites into the second main inequality
\begin{equation}
    \Prob{\Alg \geq \ell} \geq  \pth{\frac{1-e^{-\int_{0}^{C^{-1}(\ell')} C(t) \dif t } + \int_{C^{-1}(\ell')}^1 \ell' e^{- \int_{0}^t C(y) \dif y } \dif t}{1-e^{-\ell'}}} \Prob{Z \geq \ell}. 
\eqlab{iid:case2:bound}
\end{equation}

By applying the majorization technique discussed earlier, combining \Eqref{iid:case1:bound} and minimizing \Eqref{iid:case2:bound} for $\ell'\in (0, C(1)]$, we establish a lower bound for the competitive ratio $c$ as follows
\begin{align*}
    c &\geq \min \left\{ 1-e^{-\int_0^1 C(t) \dif t}, \inf_{0 < \ell' \leq C(1)}  \pth{\frac{1-e^{-\int_{0}^{C^{-1}(\ell')} C(t) \dif t } + \int_{C^{-1}(\ell')}^1 \ell' e^{- \int_{0}^t C(y) \dif y } \dif t}{1-e^{-\ell'}}}  \right\}.
\end{align*}

\end{proof}
Simple curves $C$ achieve good competitive ratios for \Eqref{monstrosity}. Recall that the optimal threshold algorithm for the \IID case achieves a competitive ratio of approximately $0.745$.

If we consider simple step function curves---specifically, using $\Xi(c_1)$ for some constant $c_1$ from time $0$ to $\frac{1}{m}$, transitioning to $\Xi(c_2)$ for some constant $c_2$ from time $\frac{1}{m}$ to $\frac{2}{m}$, and continuing in this manner---allows us to evaluate the expression in \Eqref{monstrosity}. This is because it simplifies the evaluation to a mere summation, as the integrals are transformed into summations. With $m=10$, we show that there exists a curve that yields a competitive ratio of approximately $0.7406$ for appropriately chosen $c_1, ...., c_{10}$, almost matching the \IID ratio of approximately $0.745$. See \apdxref{B} for the code and exact values of $c_1, ..., c_{m}$ we use.. Nevertheless, we demonstrate that a function $C^\ast$ exists that analytically achieves an exact competitive ratio of $\approx 0.745$.

\begin{lemma}
    There exists a threshold function $C^\ast(t)$ that gives a competitive ratio of $c \approx 0.745$ for the \IID prophet inequality.
\end{lemma}
 
\begin{proof}
We relax the optimization from \Eqref{monstrosity}. Let $\tau = C^{-1}(\ell') \in [0, 1]$, and define 
\[
\phi(\tau, \ell') = 1-e^{-\int_0^\tau C(t) \dif t} + \int_{\tau}^1 \ell' e^{-\int_0^t C(y) \dif y} \dif t - c (1-e^{-\ell'}).
\]

We relax the optimization to requiring 
\begin{equation*}
    \inf_{\substack{0\leq \tau\leq 1 \\ 0< \ell'}} \phi(\tau, \ell') \geq 0. 
\end{equation*}
We first optimize for $\ell' > 0$. Define $g(z) = \frac{1}{c}\int_z^1 e^{-\int_0^x C(y) \dif y} \dif x$. Then $g'(z)=-\frac{1}{c}e^{-\int_0^z C(y) \dif y}$. Then we have 
\[
\frac{\partial \phi}{\partial \ell' } = \int_\tau^1 e^{-\int_0^x C(y) \dif y}\dif x - ce^{-\ell'} = cg(\tau) - ce^{-\ell'}.
\]
Setting this to $0$, and substituting into $\phi$, we derive 
\[
\Phi(\tau) = \min_{\ell' > 0} \phi(\tau, \ell') =  1+cg'(\tau) - c\log(g(\tau)) g(\tau) - c + c g(\tau).
\]
The remainder of the proof follows \cite{s-couup-18} in showing that the differential equation $\Phi(\tau)=0$ for $\tau \in [0, 1]$ is satisfied for $c\approx 0.745$ (the \IID constant) for some $g^\ast(.)$. Finally, we have 
\[
C^\ast(z) = - \frac{\partial^2 g^\ast / \partial z^2}{\partial g^\ast / \partial z}. 
\]
The function $C^\ast(t)$ for $c=0.74544$ is shown in \figref{Cx}. 
\end{proof}

\begin{figure}\centering
 \includegraphics[width=0.5\textwidth]{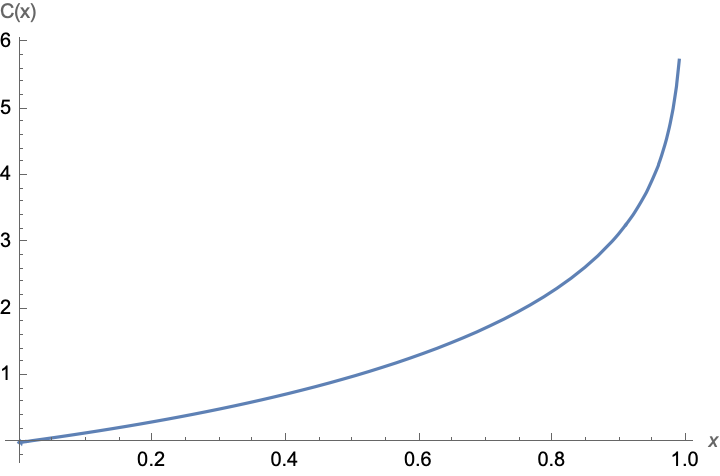}
 \caption{The function $C^\ast(x)$ for $c=0.745$ solved using numerical methods, truncated at $x=0.99$.} 
 \figlab{Cx}
 \end{figure}

\paragraph{Independence of $n$} This above section shows that algorithms that are based on thresholds of the form $\sum_i \Prob{X_i \geq \tau} =q_i $ are comparable to algorithm that choose their thresholds based on the maximum distribution (i.e., quantiles of $Z$), at least for the \IID case. One interesting fact about the result above is that the curve is independent of $n$. This is because we are approximating a continuous curve, that is independent of $n$. In particular, the $m=10$ thresholds holds for all sufficiently large $n$. 

\newpage 
\section{Sharding Warmup: Standard Prophet Inequalities.}
\seclab{sharding}

\paragraph{Sharding without Time of Arrival} Returning to the non-\IID scenario, our goal is to apply strategies akin to those used for the \IID case, which rely on summation thresholds. In the \IID context, small probabilities $\Prob{X_i \geq \tau_k}$ allow for the application of Poissonization. However, this does not carryover to the non-\IID setting. This is because of ``superstars'', random variables $X_i$ with ``large'' $\Prob{X_i \geq \tau_k}$. Indeed, it is no longer sufficient to use a Poisson distribution to count the number of arrivals in a region because of the non \IID nature of the random variables. What can we do then?

The main idea is to think about ``breaking'' each random variable $X_i$ with \CDF $F_i$ into $K$ \emph{shards}. More formally, we consider the  \IID random variables $Y_{i,1}, \ldots , Y_{i,K}$ with \CDF $F_i^{1/K}$\footnote{This is a valid \CDF because $F^{1/K}(-\infty)=0, F^{1/K}(\infty)=1$ and $F^{1/K}$ is still monotonic for positive integer $K$.}. This is an idea that was implicitly used in \cite{ethlm-psmpc-19}. One can easily see that the distribution of $\max(Y_{i,1}, \ldots, Y_{i,K})$ is the same as $X_i$, and so sampling from $X_i$ is the same as sampling from the shards, and taking the maximum-valued shard as the representative for $X_i$.

An important nuance concerning shards is that the shard with the maximum value in $\{Y_{i, j}\}_{\substack{1\leq i \leq n\\ 1\leq j \leq K}}$ must corresponds to an \textit{actual realization} of some $X_i$. This is because no other shard surpasses it, meaning some $X_i$ would adopt its value. Conversely, not all shards represent actual realizations of $X_i$; some may be dominated by other shards from the same variable.

As $K$ approaches infinity, the likelihood of a shard exceeding a threshold diminishes since $1 - F^{1/K}(\tau)$ approaches zero\footnote{We will assume without loss of generality that $F_i(\tau) > 0$ for all $i$. If $F_i(\tau) = 0$, then $K(1 - \Prob{X_i \leq \tau}^{1/K}) = K$. However, for our analysis, we require that $\lim_{K \to \infty} \sum_i K(1 - \Prob{X_i \leq \tau}^{1/K}) = q$ for a \textbf{constant} $q$. This condition cannot be satisfied if $F_i(\tau) = 0$.}. Consequently, the coupling argument for the \IID scenario remains applicable here. By revisiting the argument from  \lemref{sharding-break}, a similar conclusion regarding the variational distance approaching zero as $K \to \infty$ can be drawn, irrespective of $n$. However, the relationship between summation based thresholds on the \textit{shards} $\{Y_{i,j}\}$ and maximum-based thresholds for the actual realizations $\{X_i\}$ is not clear. The connection is made in the following lemma.

\begin{lemma}
\lemlab{sharding-break}
    Let $\tau$ be a summation based threshold on the \textbf{shards} such that \begin{equation*}
        \sum_{i=1}^n \sum_{j=1}^K \Prob{Y_{i,j} \geq \tau} = q. 
    \end{equation*} Then as $K\to \infty$, we have $\Prob{Z \leq \tau} = e^{-q}$.
\end{lemma}
\begin{proof}
Because $Y_{i,j}$ are \IID for fixed $i$, then we have $\sum_{i=1}^n K\Prob{Y_{i,1} \geq \tau} =q$. However, recall that $\Prob{Y_{i,1} \geq \tau}=1-\Prob{X_i \leq \tau}^{1/K}$. Hence, we are choosing a threshold such that 
\[
\sum_{i=1}^n K(1-\Prob{X_i \leq \tau}^{1/K}) =q.
\]
What happens when we take $K\to \infty$? The limit of $K(1-x^{1/K})$ as $K\to \infty$ can be evaluated with L'Hôpital's rule
\begin{align*}
        \lim_{K \to \infty} K(1-x^{1/K}) &= \lim_{K \rightarrow \infty} K(1 - x^{1/K}) \\
        & = \lim_{K \rightarrow
       \infty} \frac{1 - \exp(\log(x)/K)}{1/K} \\
       & = \lim_{K \rightarrow
       \infty} \frac{ \log(x)\exp(\log(x)/K)/K^2 }{-1/K^2}\\ 
       & = - \log
    x.
\end{align*}
And so we have that for $K\to \infty$, $\sum_{i=1}^n -\log \Prob{X_i \leq \tau} = q$. This implies $-\log \Prob{Z\leq \tau} = q$. In other words, we chose a threshold such that $\Prob{Z\leq \tau} = e^{-q}$.
\end{proof}

Hence, we retrieve maximum based thresholds, but with a twist: we now have an alternative view in terms of shards.  Specifically, if we choose maximum-based threshold $\tau_j$ such that $\Prob{Z\leq \tau_j}=\alpha_j$, then the number of \textbf{shards} above $\tau_j$ follows a Poisson distribution with rate $\log \frac{1}{\alpha_j}$. This is only possible because the probability of each shard being above $\tau_j$ is small (i.e., $\to 0$ as $K\to \infty$).  

To signify the importance of this view and to warmup, we reprove several known results in the literature with this new point of view. None of these results are \textbf{needed} for our new results, however, they provide a much needed warmup for the sharding machinery.  \textit{We crucially emphasize here that sharding is only done for the analysis of the algorithm. }

\paragraph{Sharding with time of arrival} Consider the scenario where the variables $X_1, \ldots, X_n$ are associated with times of arrival. In this variant, each shard independently selects a random time of arrival, uniformly distributed over the interval $[0, 1]$. The arrival time of $X_i$ is determined by the arrival time of the shard with the maximum value. Given the independence in the selection of arrival times by each shard, it consequently follows that the arrival times of $X_i$ are independently determined. This independence allows the extension of the sharding analysis to scenarios incorporating time of arrival, such as the prophet secretary problem. An example is provided to clarify this extension.

\begin{lemma}\label{eta8}
\lemlab{prooffromthebook}
    For the prophet secretary problem, consider the single threshold algorithm that chooses $\tau$ such that $\Prob{Z\leq \tau}=1/e$ and accepts the first value (if any) above $\tau$. Then the algorithm has a $1-1/e$ competitive ratio. 
\end{lemma}
\begin{proof}
    We shard the $n$ random variables. Invoking \lemref{sharding-break}, we establish that \begin{equation*}
        \lim_{K\to \infty}\sum_{i=1}^n \sum_{j=1}^K \Prob{Y_{ij} \geq \tau} = q = -\log(1/e) = 1.
    \end{equation*}

\paragraph{Case 1: $\ell \in [0, \tau]$.} The condition for the algorithm to accept a value $\geq \ell$ is met if at least one shard surpasses $\tau \geq \ell$. Consequently,
    \begin{equation}
            \frac{ \Prob{\Alg \geq \ell}}{\Prob{Z\geq \ell}} \geq \Prob{\Alg \geq \ell} = 1-e^{-q}=1-\frac{1}{e}. \eqlab{single:threshold:prophet:sec:1}
    \end{equation}

\paragraph{Case 2: $\ell \in [\tau, \infty)$.} Consider the scenario where the region $A = [0, 1] \times [\tau, \ell]$ contains $\beta$ shards, and region $B = [0, 1] \times [\ell, \infty)$ contains at least one shard. If the highest value shard in $B$ arrives before all $\beta$ shards in region $A$, the algorithm will accept a value exceeding $\ell$. This is because the highest value shard in $B$ represents an actual realization of some $X_i$, and potentially, all $\beta$ shards in $A$ could correspond to realizations of some $X_j$ that prevent us from selecting a value above $\ell$. Let $\ell' = \sum_{i=1}^n \sum_{j=1}^K \Prob{Y_{ij} \geq \ell}$. The Poisson rate of shards in $A$ is $q - \ell' = 1-\ell'$, while that in $B$ is $\ell'$. Thus,
    \begin{equation}
            \frac{ \Prob{\Alg \geq \ell}}{\Prob{Z\geq \ell}} = \frac{\Prob{\Alg \geq \ell}}{1-e^{-\ell'}} \geq \frac{\sum_{\beta=0}^{\infty} (1-e^{-\ell'})e^{-(q-\ell')} \frac{(q-\ell')^\beta}{\beta!} \frac{1}{\beta+1}  }{1-e^{-\ell'}} = \frac{e-e^{\ell'}}{e-e\ell' } \eqlab{single:threshold:prophet:sec:2}. 
    \end{equation}
The expression on the right-hand side of \Eqref{single:threshold:prophet:sec:2} is an increasing function in $\ell' \in (0, 1]$, reaching its minimum for $\ell' \to 0$, with a value of $1 - \frac{1}{e}$. By stochastic dominance, and combining \Eqref{single:threshold:prophet:sec:1} and \Eqref{single:threshold:prophet:sec:2}, we conclude that the algorithm has a competitive ratio of $1 - \frac{1}{e}$.
\end{proof}

\begin{figure}
\centering
\begin{subfigure}{.5\textwidth}
    \centering
    \includegraphics[width=0.8\textwidth]{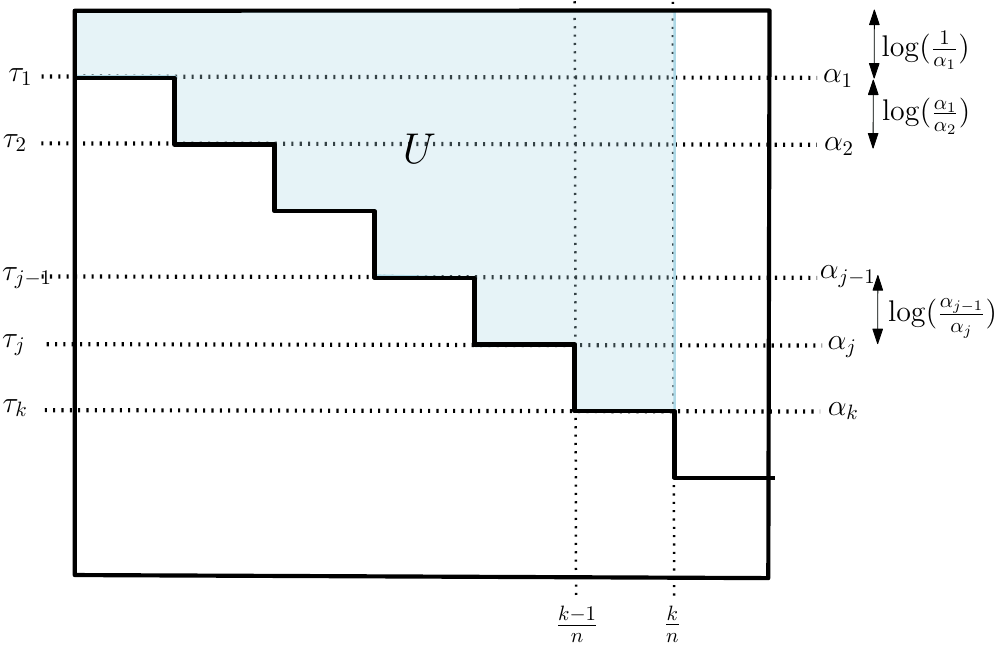}
    \caption{Proof of \lemref{simplified-upper}. The region $U$ is the light\\   blue region. }
    \figlab{simplified-upper}
\end{subfigure}
\begin{subfigure}{.5\textwidth}
    \centering
    \includegraphics[width=0.8\textwidth]{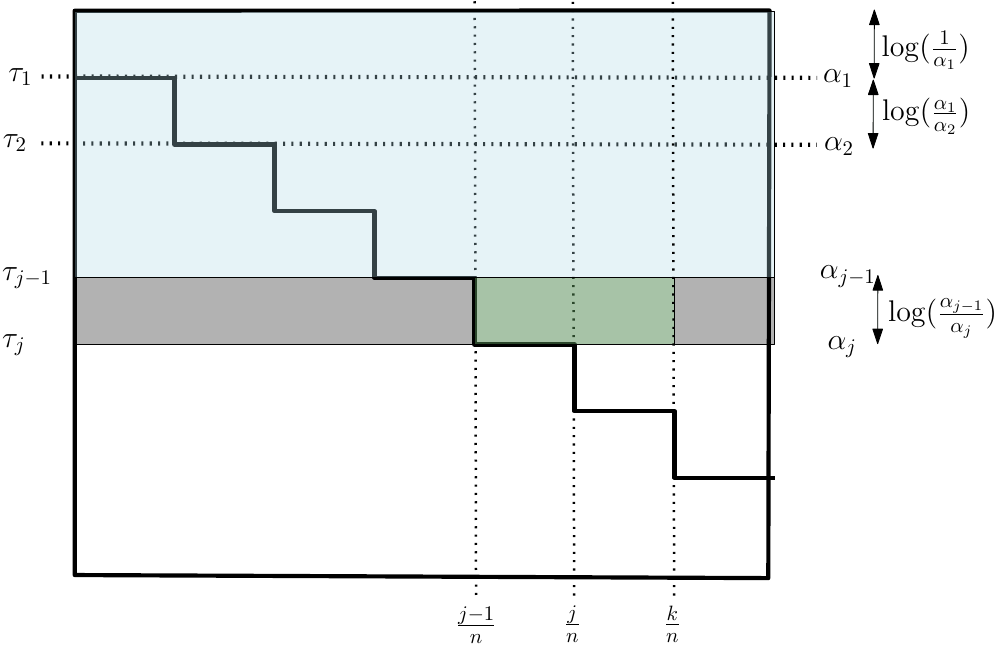}
    \caption{The light blue region is $A_j$, the gray (and green) region is $B_j$, and the green region is where $(t^\ast, v^\ast)$ arrives.}
    \figlab{lower-bound}
\end{subfigure}
\label{fig:test}
\caption{}
\end{figure}

Next, we re-prove the following results that were proven in \cite{csz-pstbs-18} for the prophet secretary variant via a nontrivial argument that applies a Schur-convexity inequality an infinite number of times. The short proof below establishes the same results via the new shards point of view. 

\begin{lemma}{\cite{csz-pstbs-18}}\label{eta9}
\lemlab{simplified-upper}
    Let $T\in [n]$ be a random variable for the time that the algorithm following thresholds $\tau_1 \geq \ldots \geq \tau_n$ selects a value (if any) with $\Prob{Z\leq \tau_j}=\alpha_j$ for the prophet secretary problem. Then for any $k\in [n]$
    \[
    \Prob{T>k} \geq  \pth{\prod_{j=1}^k \alpha_j }^{1/n}.
    \]
\end{lemma}
\begin{proof}
    Refer to \figref{simplified-upper} throughout this proof. Define the descending threshold function $\tau(t) = \tau_{\lceil tn \rceil }$. Let $U=\{(t, v) \mid 0\leq t\leq k/n, v \geq \tau(t)\}$. The condition $T>k$ is equivalent to the absence of realizations (in terms of $X_i$) within $U$. Consider the event $\xi$, characterized by the absence of \emph{shards} in $U$. This event implies the absence of realizations (in terms of $X_i$s) in $U$, and thus $\Prob{T>k} \geq \Prob{\xi}$. Setting $\alpha_0=1$, the measure of the region $U$ can be expressed as a telescoping sum
\begin{align*}
    \mu(U) &= \sum_{i=1}^k \frac{k-i+1}{n}\left( \log\left(\frac{1}{\alpha_i}\right) - \log\left(\frac{1}{\alpha_{i-1}}\right)\right) \\
    &= \sum_{i=1}^k \frac{1}{n} \log\left(\frac{1}{\alpha_i}\right)\\
    &= \frac{1}{n} \log \left(\frac{1}{\prod_{i=1}^k \alpha_i }\right). 
\end{align*}
Hence, the probability of event $\xi$ is given by:
\[
\Prob{\xi} = e^{-\mu(U)} = \left(\prod_{i=1}^k \alpha_i\right)^{1/n}. 
\]
\end{proof}

\cite{csz-pstbs-18} also prove the following inequality. We can also prove the same inequality via an event on the \textbf{shards} that implies $T\leq k$ and whose probability is the RHS. 
\begin{lemma}
\lemlab{discrete_blind_lb}{\cite{csz-pstbs-18}}
    Let $T\in [n]$ be a random variable for the time that the algorithm following thresholds $\tau_1 \geq \ldots \geq \tau_n$ selects a value (if any) with $\Prob{Z\leq \tau_j}=\alpha_j$ for the prophet secretary problem. Then for any $k\in [n]$
    \[
    \Prob{T\leq k}\geq \frac{1}{n}\sum_{j=1}^k \pth{1-\alpha_j}.
    \]
\end{lemma}
\begin{proof}
    Refer to \figref{lower-bound} throughout this proof. Define the threshold function $\tau(t) = \tau_{\lceil tn \rceil }$. Formally, we consider the event $\xi$, characterized by the existence of some $1 \leq j \leq k$ for which the region $A_j = \{(t, v) \mid 0 \leq t \leq 1, \tau_{j-1} \leq v < \infty\}$ is empty of shards, whereas the region $B_j = \{(t, v) \mid 0 \leq t \leq 1, \tau_{j} \leq v \leq \tau_{j-1}\}$ has at least one shard, with the highest value shard $(t^\ast, v^\ast)$ in $B_j$ arriving within the interval $t=(j-1)/n$ to $t=k/n$.

Informally, this event signifies that the region $[\tau_1, \infty)$ contains a shard, and the \emph{maximum value shard} among them is present between $t=0$ and $t=k/n$, \textbf{or} the region $[\tau_1, \infty)$ is devoid of shards, the region $[\tau_2, \tau_1]$ contains shards, with the maximum shard situated between $t=1/n$ and $t=k/n$, \textbf{or} the region $[\tau_2, \infty)$ lacks shards, while the region $[\tau_3, \tau_2]$ contains shards, with the highest value shard arriving between $t=2/n$ and $t=k/n$, and so forth. This event implies $T \leq k$ as it guarantees the presence of at least one realization from $X_i$ exceeding $\tau(t)$ before time $k/n$. The probability of this event can be simplified by telescoping sums:
\begin{align*}
    \Prob{\xi} &= \sum_{i=0}^{k-1} e^{-\log\pth{\frac{1}{\alpha_{i+1}}}}\pth{1 - e^{-\pth{\log\left(\frac{1}{\alpha_{i+1}}\right) - \log\left(\frac{1}{\alpha_i}\right)}}}\frac{k-i}{n}, \\ 
    &=\sum_{i=0}^{k-1} \alpha_{i}\left(1 - \frac{\alpha_{i+1}}{\alpha_{i}}\right)\frac{k-i}{n} \\
    &= \sum_{i=0}^{k-1} (\alpha_{i} - \alpha_{i+1})\frac{k-i}{n} \\
    &= \frac{1}{n} \sum_{i=0}^{k-1} 1 - \alpha_{i+1} \\
    & = \frac{1}{n}\sum_{i=1}^k 1 - \alpha_i.
\end{align*}
\end{proof}
\section{\textsc{Top-1-of-$k$}.}
\seclab{best1ofk}

\begin{algorithm}
\caption{$0.776$ competitive algorithm for \textsc{Top-$1$-of-$2$}. }\algolab{best1of2:noniid:firstimprovement}
\begin{algorithmic}
    \State Choose $\tau_1, \tau_2$ such that $\Prob{Z \leq \tau_i}=e^{-c_i}$ for constants $c_1, c_2$ as described in \lemref{noniid:firstimprovement}. 
    \State Set $r \gets 1$
    \For{$i = 1, \ldots, n$}
        \If{$X_i \geq \tau_r$}
            \State $r \gets r + 1$
            \State Accept $X_i$. 
            \State If we have accepted $2$ items, break.
        \EndIf
    \EndFor
\end{algorithmic}
\end{algorithm}

\paragraph{Improved algorithm for Non-\IID \textsc{Top-$1$-of-$2$}.} We give an improved algorithm for the \textsc{Top-$1$-of-$2$} problem. This improves the result by Assaf and Samuel-Cahn \cite{as-srpimm-00} from $\approx 0.731$ to $0.776$. We then improve this to $\approx 0.781$. See \algoref{best1of2:noniid:firstimprovement}. The algorithm is a simple two-threshold algorithm. We select thresholds $\tau_1=\Xi(c_1), \tau_2 = \Xi(c_2)$ on the shards, for some constants $c_1 > c_2$. Specifically, we choose thresholds $\tau_1, \tau_2$ such that 
\[
\lim_{K\to \infty}\sum_{i=1}^n \sum_{j=1}^K \Prob{Y_{i,j}\geq \tau_i} = c_i. 
\]
The algorithm accepts the first value (if any) above $\tau_1$, and updates the threshold to $\tau_2$. It finally accepts any value (if any) above $\tau_2$, and terminates.

\begin{figure}
    \centering
    \includegraphics[width=0.9\textwidth]{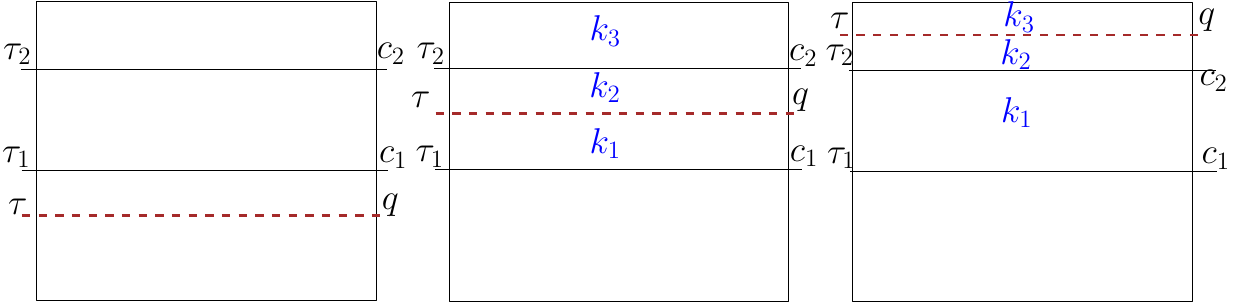}
    \caption{The 3 cases of the analysis for \textsc{Top-$1$-of-$2$} algorithm from left to right.}
    \figlab{best1ofkproof}
\end{figure}

\begin{lemma}
\lemlab{noniid:firstimprovement}
    The competitive ratio $c$ of \algoref{best1of2:noniid:firstimprovement}  is 
    \begin{equation}
            c \geq \min\left( 1-e^{-c_1},e^{-c_1-c_2} \left(e^{c_1} \left(e^{c_2}-1\right)+2 \sqrt{\left(e^{c_1}-1\right) \left(e^{c_2}-1\right)}-e^{c_2}+2\right)  ,  e^{-c_2} + e^{-c_1} c_2  \right) .\eqlab{best1of2}
    \end{equation}
    In particular, choosing $c_1=1.49721, c_2=0.364197$ yields $c\geq 0.776245$. 
\end{lemma}
\begin{figure}\centering
 \includegraphics[width=0.4\textwidth]{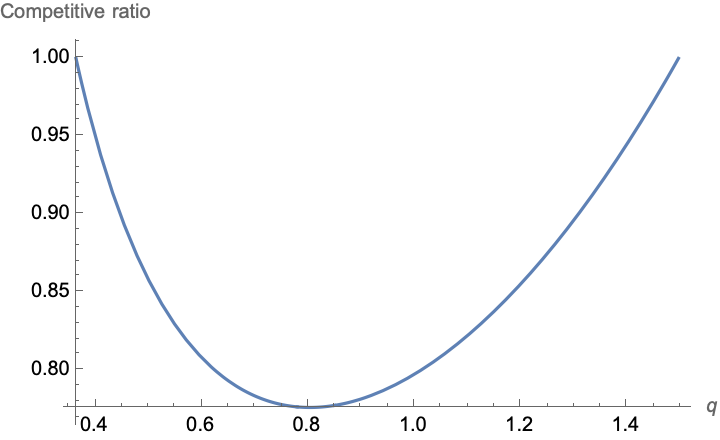}
 \caption{Value of $\left( \frac{e^{-(c_1 - q)}(1-e^{-q}) + (1-e^{-(c_1-q)})(1-e^{-c_2}) }{1-e^{-q}} \right)$ as $q$ ranges from $c_2$ to $c_1$} 
 \figlab{qRange}
 \end{figure}

\begin{proof}
    Refer to \figref{best1ofkproof} throughout this analysis, where we employ stochastic dominance.

\paragraph{Case 1: $\tau \in [0, \tau_1)$.} If there exists a shard exceeding $\tau_1$, the algorithm will choose a value greater than $\tau$. Hence, we obtain
\begin{equation}
    \frac{\Prob{\Alg \geq \tau}}{\Prob{Z\geq \tau}} \geq \Prob{\Alg \geq \tau_1} =  1-e^{-c_1} \eqlab{eq:second:one}.
\end{equation}
\paragraph{Case 2: $\tau \in [\tau_1, \tau_2]$.} Define $q = \lim_{K\to\infty}\sum_i \sum_j \Prob{Y_{i,j} \geq \tau}$, where $q \in [c_2, c_1]$. It follows that $\Prob{Z \geq \tau} = 1 - e^{-q}$. Let $k_1$ be the number of shards with values in $[\tau_1, \tau)$, $k_2$ be the number of shards with values  in $[\tau, \tau_2)$, and $k_3$ be the number of shards with value $[\tau_2, \infty)$. Consider the following event on the shards that implies $\Alg \geq \tau$: if $k_1=0$ and $k_2 + k_3 \geq 1$, \textit{or} $k_1 \geq 1$ and $k_3 \geq 1$, then $\Alg \geq \tau$. In the first scenario, the presence of at least one shard above $\tau$ corresponds to an actual realization of $\{X_i\}$. With $k_1=0$, this realization is selected by the algorithm. If $k_1 \geq 1$ and $k_3 \geq 1$, a shard above $\tau_2 \geq \tau$ corresponds to an actual realization of $\{X_i\}$. In a worst-case scenario, one of the $k_1$ shards from $[\tau_1, \tau]$ corresponds to actual realizations in $\{X_i\}$ and arrives first, prompting the algorithm to raise the threshold to $\tau_2$ and ultimately select a value $\geq \tau$. Therefore,
\begin{equation}
    \frac{\Prob{\Alg \geq \tau}}{\Prob{Z\geq \tau}} \geq \frac{e^{-(c_1 - q)}(1-e^{-q}) + (1-e^{-(c_1-q)})(1-e^{-c_2}) }{1-e^{-q}} = g(q) \eqlab{gq:two}. 
\end{equation}

For $q \in [c_2, c_1]$, we find $g(c_1) = g(c_2) = 1$ and $g$ is minimized when $g'(q) = 0$ (See \figref{qRange}). We derive that
\begin{equation*}
    g'(q) = \frac{e^{-c_1 - c_2 + q} \left(-e^{c_1 + c_2} + e^{c_1} + e^{c_2} - 2 e^q + e^{2 q}\right)}{\left(1 - e^q\right)^2},
\end{equation*}
leading to the condition for $g'(q) = 0$ as
\begin{equation*}
    -e^{c_1 + c_2} + e^{c_1} + e^{c_2} - 2 e^q + e^{2 q} = 0.
\end{equation*}
Solving for $q$, we find $g(q)$ is minimized at
\begin{equation*}
    q = \log \left(\sqrt{e^{c_1 + c_2} - e^{c_1} - e^{c_2} + 1} + 1\right).
\end{equation*}
Substituting this into \Eqref{gq:two}, we obtain the second main inequality as
\begin{equation}
    \frac{\Prob{\text{Alg} \geq \tau}}{\Prob{Z \geq \tau}} \geq e^{-c_1 - c_2} \left(e^{c_1} \left(e^{c_2} - 1\right) + 2 \sqrt{(e^{c_1} - 1)(e^{c_2} - 1)} - e^{c_2} + 2\right) \eqlab{eq:second:two}.
\end{equation}

\paragraph{Case 3: $\tau \in [\tau_2, \infty)$.} For this case, $\Prob{Z \geq \tau} = 1 - e^{-q}$, where $q = \sum_i \sum_j \Prob{Y_{i, j} \geq \tau}$ and $q \in (0, c_2]$. Consider the event on the shards that implies $\Alg \geq \tau$. Let $k_1$ be the count of shards within $[\tau_1, \tau_2)$, $k_2$ within $[\tau_2, \tau)$, and $k_3$ within $[\tau, \infty)$. If $k_1 = 0$, $k_2 \in \{0, 1\}$, $k_3 \geq 1$, or $k_1 \geq 1$, $k_2 = 0$, $k_3 \geq 1$, then the algorithm secures a value at least $\tau$. In the first scenario, with at most one shard below $\tau$ and at least one above, the algorithm chooses a value above $\tau$. In the second scenario, if $k_1 \geq 1$ and $k_2 = 0$, then, in the worst case, one of the $k_1$ shards corresponds to an actual realization, prompting the algorithm to increase its threshold and accept the first realization above $\tau_2$, as $k_3 \geq 1$. Therefore, we have
\begin{align*}
    \frac{\Prob{\Alg \geq \tau}}{\Prob{Z \geq \tau}} &\geq \frac{ \left[ e^{-(c_1-c_2)} e^{-(c_2-q)}(1+c_2-q) + (1-e^{-(c_1-c_2)})e^{-(c_2-q)}  \right] (1-e^{-q})  }{1-e^{-q}} \\
    & =  \left[ e^{-(c_1-c_2)} e^{-(c_2-q)}(1+c_2-q) + (1-e^{-(c_1-c_2)})e^{-(c_2-q)}  \right]\\
    & = e^{-c_2 + q} + e^{-c_1 + q} (c_2 - q).
\end{align*}
Observing that $e^{-c_2 + q} + e^{-c_1 + q} (c_2 - q)$ is increasing in $q \in (0, c_2)$ due to its positive derivative, we find
\[
\min_{0 < q \leq c_2} e^{-c_2 + q} + e^{-c_1 + q} (c_2 - q) = e^{-c_2} + e^{-c_1} c_2,
\]
leading to
\begin{equation}
    \frac{\Prob{\Alg \geq \tau}}{\Prob{Z \geq \tau}} \geq e^{-c_2} + e^{-c_1} c_2. \eqlab{eq:second:three}
\end{equation}
Combining \Eqref{eq:second:one}, \Eqref{eq:second:two}, and \Eqref{eq:second:three} by stochastic dominance yields the result. 

By selecting $c_1 = 1.49721$ and $c_2 = 0.364197$ in accordance with the above analysis, we deduce the competitive ratio for the \textsc{Top-$1$-of-$2$} problem as at least $0.776$.  
\end{proof}

The analysis can be extended to incorporate three thresholds. 

\begin{algorithm}
\caption{$0.781$ competitive algorithm for \textsc{Top-$1$-of-$2$}. }\algolab{best1of2:noniid:secondimprovement}
\begin{algorithmic}
    \State Choose $\tau_1, \tau_2, \tau_3$ such that $\Prob{Z \leq \tau_i}=e^{-c_i}$ for constants $c_1>c_2>c_3$ as described in \lemref{noniid:secondimprovement}. 
    \State Set $r \gets 1$
    \For{$i = 1, \ldots, n$}
        \If{$X_i \geq \tau_r$}
            \State \begin{math}r \gets \min_j \left\{ ~ \{k : \tau_k > X_i\} ~\cup~ \{3\}  \right\}\end{math}.
            \State Accept $X_i$. 
            \State If we have accepted $2$ items, break.
        \EndIf
    \EndFor
\end{algorithmic}
\end{algorithm}

\begin{lemma}
\label{eta1}
\lemlab{noniid:secondimprovement}
    For $c_1=1.51921, c_2=0.380251, c_3=0.0386845$, \algoref{best1of2:noniid:secondimprovement} is a $0.781$ competitive algorithm  for the \textsc{Top-$1$-of-$2$} problem. 
\end{lemma}
\begin{proof}
    Specifically, the algorithm employs thresholds $\tau_1, \tau_2, \tau_3$, defined such that
\begin{equation*}
    \lim_{K \to \infty} \left(\sum_{i=1}^n K\left(1 - \Prob{X_i \leq \tau_i}^{1/K}\right)\right) = c_i,
\end{equation*}
for distinct constants $c_1 > c_2 > c_3$. Initially, the algorithm uses $\tau_1$ and upon encountering a value $v$ exceeding $\tau_1$, accepts it, and switches to $\tau_{n(v)}$, where $n(v) = \min \left\{ \{j : \tau_j > v \} \cup \{3\} \right\}$. This process selects the next threshold higher than $v$, or defaults to the last threshold otherwise. This generalizes upon \algoref{best1of2:noniid:firstimprovement}.

Repeating the sharding analysis for three thresholds, with $c_1=1.51921, c_2=0.380251, c_3=0.0386845$ from \lemref{noniid:firstimprovement}, the competitive ratio is at least $\min(C_1, C_2, C_3, C_4)$, where
\begin{align*}
    C_1 &= 1-e^{-c_1}~, \\
    C_2 &= \inf_{c_2 \leq \ell \leq c_1} \left\{ \frac{e^{-c_1-c_2+l} \left(-e^{c_1}-e^{c_2}+e^{c_1+c_2}+e^l\right)}{e^l-1} \right\}~,\\
    C_3 &= \inf_{c_3 \leq \ell \leq c_2} \left\{ \frac{e^{-c_1-c_2-c_3+l} \left(e^{c_2+l}+e^{c_1+c_3+l}-e^{c_2+c_3+l}-e^{2 c_2}-e^{c_1+c_3}+e^{2 c_2+c_3}\right)}{e^l-1} \right\}~,\\
    C_4 &= \inf_{0 < \ell \leq c_1} \left\{ e^l \left(-e^{-c_1} (l+1)+e^{-c_2}+e^{-c_1+c_2-c_3}+e^{-c_1} c_3\right)  \right\}. 
\end{align*}

\begin{figure}\centering
 \includegraphics[width=0.4\textwidth]{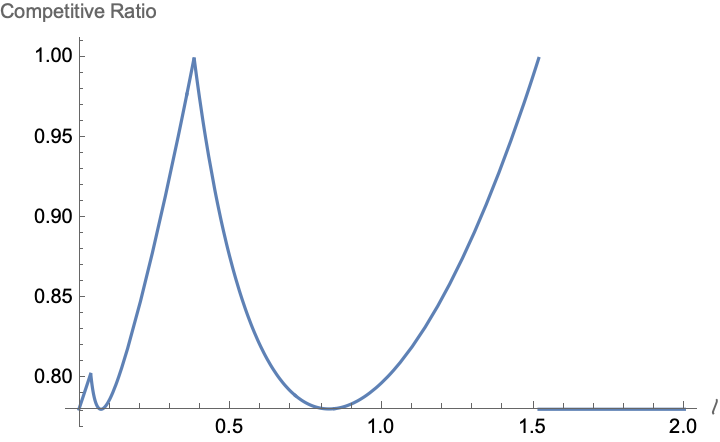}
 \caption{Plot of $\min(C_1, C_2, C_3, C_4)$ for $\ell$ ranging from $0$ to $2$.} 
 \figlab{lrangethree}
\end{figure}

Figure \figref{lrangethree} depicts the four functions of $\ell$. It can be shown analytically by standard calculus that the minima satisfy:
\begin{align*}
    C_1 &= 1-e^{-c_1},  \\
    C_2 &= e^{-c_1-c_2} \left(e^{c_1} \left(e^{c_2}-1\right)-e^{c_2}+2 \sqrt{\left(e^{c_1}-1\right) \left(e^{c_2}-1\right)}+2\right),\\
    C_3 &= \left(e^{c_2} \left(e^{c_2}-2\right) \left(e^{c_3}-1\right)+e^{c_1+c_3}+2 \sqrt{e^{c_2} \left(e^{c_2}-1\right) \left(e^{c_3}-1\right) \left(e^{c_2}+e^{c_1+c_3}-e^{c_2+c_3}\right)}\right) \times \\
    &~~~~~~~\left(\cosh \left(c_1+c_2+c_3\right)-\sinh \left(c_1+c_2+c_3\right)\right),\\
    C_4 &= e^{-c_1} c_3-e^{-c_1}+e^{-c_2}+e^{-c_1+c_2-c_3}.
\end{align*}
Substituting $c_1, c_2, c_3$ yields the result.
\end{proof}

\paragraph{An improved upper bound for $k=2$.} We now improve the upper bound for \textsc{Top-$1$-of-$2$} from $0.8$ by Assaf and Samuel-Cahn \cite{as-srpimm-00} to $0.79424$. Contrary to the $n=3$ random variables instance utilized by Assaf and Samuel-Cahn, our construction involves $n=4$ random variables. 

\begin{lemma}
\lemlab{betterub:best1of2}
    No algorithm can achieve a competitive ratio $>0.7943$ for the \textsc{Top-$1$-of-$2$} model. 
\end{lemma}
\begin{proof}
\label{eta2} We define $4$ random variables defined as 
\begin{equation*}
    X_i = \begin{cases}
        b_i & \text{with probability } p_i \\
        0 & \text{otherwise}
    \end{cases}
    \text{for } i=1,2,3,4.
\end{equation*}
We will require $b_1<b_2<b_3<b_4$ and $p_1=1$ (i.e $X_1=b_1$ always). The prophet value can be computed as 
\begin{equation*}
    \Ex{Z} = \sum_{1\leq i \leq 4} \pth{\prod_{i+1\leq j \leq 4} (1-p_j)} p_i b_i.
\end{equation*}
Next, we consider all possible algorithms for the instance. In total, we need to consider $6$ algorithms $\mathcal{A}_1, ..., \mathcal{A}_6$. 

$\mathcal{A}_1$ decides to accept $X_1$, and accept the next non-zero value it encounters (if any). Hence, the expected value it receives is 
\begin{equation*}
    \Ex{\mathcal{A}_1} = p_2b_2 + (1-p_2)p_3b_3 + (1-p_2)(1-p_3)p_4b_4 + (1-p_2)(1-p_3)(1-p_4)b_1. 
\end{equation*}
$\mathcal{A}_2$ decides to accept $X_1$, and accept the next non-zero value it encounters (if any) starting from $X_3$. Hence, the expected value it receives is 
\begin{equation*}
    \Ex{\mathcal{A}_2} = p_3b_3 + (1-p_3)p_4b_4 + (1-p_3)(1-p_4)b_1. 
\end{equation*}
$\mathcal{A}_3$ decides to accept $X_1$, and wait until $X_4$. Hence, the expected value it receives is 
\begin{equation*}
    \Ex{\mathcal{A}_3} = p_4b_4 + (1-p_4)b_1. 
\end{equation*}
$\mathcal{A}_4$ decides to wait until $X_3, X_4$ to use its 2 slots. Hence, the expected value it receives is 
\begin{equation*}
    \Ex{\mathcal{A}_4} = \Ex{\max(X_3, X_4)} =  p_4b_4 + (1-p_4)p_3 b_3. 
\end{equation*}
$\mathcal{A}_5$ skips $X_1$. It $X_2=0$, then it just gets the maximum of $X_3, X_4$. Otherwise, if $X_2>0$, then it accepts $X_2$, then it accepts the first non-zero value from $X_3, X_4$ (if any). Hence, the expected value it receives is 
\begin{equation*}
    \Ex{\mathcal{A}_5} = p_2 \left[ p_3b_3 + (1-p_3)p_4 b_4 + (1-p_3)(1-p_4)b_2 \right] + (1-p_2)\Ex{\mathcal{A}_4}. 
\end{equation*}
Finally, $\mathcal{A}_6$ skips $X_1$. If $X_2=0$, then it just gets the maximum of $X_3, X_4$. Otherwise, if $X_2>0$ it accepts $X_2$, then it waits for $X_4$ regardless of $X_3>0$. Hence, the expected value it receives is 
\begin{equation*}
    \Ex{\mathcal{A}_6} = p_2 \left[ p_4 b_4 + (1-p_4)b_2 \right] + (1-p_2)\Ex{\mathcal{A}_4}. 
\end{equation*}
Next, we set $b_i=c_i \beta$ for $i=1,2,3$ and $b_4=1$. We also set $p_1=1$ as mentioned earlier, and $p_4=c_4 \beta$. We then take the limit of $\Ex{\mathcal{A}_i}/\Ex{Z}$ as $\beta \to 0$. We get that the competitive ratios are 
\begin{align*}
    \alpha_1 = \lim_{\beta \to 0} \frac{\Ex{\mathcal{A}_1}}{\Ex{Z}} &= \frac{\left(c_2-c_4\right) p_2+c_1 \left(p_2-1\right) \left(p_3-1\right)+\left(c_4-c_3\right) \left(p_2-1\right) p_3+c_4}{c_1 \left(p_2-1\right) \left(p_3-1\right)-c_2 p_2 \left(p_3-1\right)+c_3 p_3+c_4} \\
    \alpha_2 = \lim_{\beta \to 0} \frac{\Ex{\mathcal{A}_2}}{\Ex{Z}} &= \frac{-\left(c_1-c_3+c_4\right) p_3+c_1+c_4}{c_1 \left(p_2-1\right) \left(p_3-1\right)-c_2 p_2 \left(p_3-1\right)+c_3 p_3+c_4}\\
    \alpha_3 = \lim_{\beta \to 0} \frac{\Ex{\mathcal{A}_3}}{\Ex{Z}} &= \frac{c_1+c_4}{c_1 \left(p_2-1\right) \left(p_3-1\right)-c_2 p_2 \left(p_3-1\right)+c_3 p_3+c_4}\\
    \alpha_4 = \lim_{\beta \to 0} \frac{\Ex{\mathcal{A}_4}}{\Ex{Z}} &= \frac{c_3 p_3+c_4}{c_1 \left(p_2-1\right) \left(p_3-1\right)-c_2 p_2 \left(p_3-1\right)+c_3 p_3+c_4} \\
    \alpha_5 = \lim_{\beta \to 0} \frac{\Ex{\mathcal{A}_5}}{\Ex{Z}} &= \frac{c_2 p_2+p_3 \left(c_3-\left(c_2+c_4\right) p_2\right)+c_4}{c_1 \left(p_2-1\right) \left(p_3-1\right)-c_2 p_2 \left(p_3-1\right)+c_3 p_3+c_4} \\
    \alpha_6 = \lim_{\beta \to 0} \frac{\Ex{\mathcal{A}_6}}{\Ex{Z}} &=  \frac{c_2 p_2-c_3 \left(p_2-1\right) p_3+c_4}{c_1 \left(p_2-1\right) \left(p_3-1\right)-c_2 p_2 \left(p_3-1\right)+c_3 p_3+c_4}
\end{align*}
Our objective is to determine the parameters $p_2, p_3, c_1, c_2, c_3, c_4$ that minimize $\max(\alpha_1, \alpha_2, \alpha_3, \alpha_4, \alpha_5, \alpha_6)$. By setting $c_1 = 1$, $c_2 = 2.04632458$, $c_3 = 2.9369093$, $c_4 = 0.8905847$, $p_2 = 0.4466646$, and $p_3 = 0.1487470$, we achieve a competitive ratio below $0.79424$. We verified these computations with the optimal dynamic program for this instance. 
\end{proof}

\begin{remark}
The reader may question the decision of the author to limit the counterexample search to merely four random variables, or some of the arbitrary decisions we made such as setting $p_1=1$ or $b_i=c_i\beta$ for $i=1,2,3$. In theory, the search could be broadened to include more variables by leveraging a Mixed-Integer Linear Programming (MILP) optimizer. However, in practice, the computational complexity became a significant concern, particularly because of the stiffness of the expression. Specifically, the solver\footnote{We use Gurobi\cite{gurobi} under the academic license.} failed to identify a comparable instance with five random variables to the one with four variables, despite running for 12 hours and having access to 128 GB of memory. Theoretically, one can achieve this by simply setting one $p_i=0$ and replicating the parameters from the four-variable case, yet the solver struggled to find such solution. A lot of the baked assumptions we made were guided by analytic educated guesses. This showcases the limitations faced when expanding the scope of the variable search. We do not claim these are the best possible parameters, yet we believe they are almost optimal at least for $n=4$ random variables. We do not know if increasing the number of random variables would help.
\end{remark}

\paragraph{\textsc{Top-$1$-of-$k$} for Non-\IID random variables. }\label{eta3} Next, we present our result for \textsc{Top-$1$-of-$k$} for Non-\IID random variables, and general $k$. 

\begin{lemma}
\lemlab{general:best1ofk:lb:improvement}
    There is an algorithm for the \textsc{Top-$1$-of-$k$} that achieves a competitive ratio of at least $1-e^{-kW(\frac{\sqrt[k]{k!}}{k})}$. This is asymptotically $1-e^{-kW(1/e) + o(k)}$ as $k\to \infty$. 
\end{lemma}
\begin{proof}
   We shard the variables $X_1, ..., X_n$ into $\{Y_{i,j}\}$. We set a single threshold $\tau_1$ such that
\[
\sum_{i=1}^n \sum_{j=1}^K \Prob{Y_{i,j} \geq \tau} = c.
\]
For some constant $c$. Again, we proceed by stochastic dominance. If $\tau \in [0, \tau_1]$, then 
\begin{equation}
    \frac{\Prob{\Alg \geq \tau}}{\Prob{Z\geq \tau}} \geq 1-e^{-c}\eqlab{general:first}.
\end{equation}

Finally, if $\tau \in [\tau_1, \infty)$, and $q=\sum_i \sum_j \Prob{Y_{i,j}\geq \tau} = q$ with $0<q\leq c$. 
Consider the number of shards with value between $\tau_1$ and $\tau$. If this number is at most $k-1$ and there is a shard above $\tau$, then the algorithm would successfully reach a shard above $\tau$ corresponding to an actual realization and take it. Hence, we have 
\[
\frac{\Prob{\Alg \geq \tau}}{\Prob{Z \geq \tau}} \geq \frac{(1-e^{-q})\sum_{i=0}^{k-1} e^{-(c-q)} \frac{(c-q)^i}{i!} }{1-e^{-q}} = \sum_{i=0}^{k-1} e^{-(c-q)} \frac{(c-q)^i}{i!} = f_k (c, q).
\]
Note that $f_k(c, q)$ is minimized for $q\to 0$ in $q\in (0, c)$. By Taylor approximation on the function $f(x)=e^{x-c}$, we have for some $\xi \in (0, c]$
\[
\sum_{i=k}^\infty e^{-c}\frac{c^i}{i!} = \frac{f^{(k)}(\xi)c^k}{k!} \leq \frac{c^k}{k!}.
\]
Hence, 
\begin{equation}
    \frac{\Prob{\Alg \geq \tau}}{\Prob{Z \geq \tau}}  \geq 1-\frac{c^k}{k!} \eqlab{general:second}.
\end{equation}

Finally, combining \Eqref{general:first} and \Eqref{general:second}, the competitive ratio of the algorithm is at least $\min\left(1-e^{-c}, 1-\frac{c^k}{k!}\right)$. We set $1-e^{-c}=1-\frac{c^k}{k!}$, which has a solution of $c=kW(\frac{\sqrt[k]{k!}}{k})$ where $W$ is the Lambert $W$ function. To see this, let $z=W(\frac{\sqrt[k]{k!}}{k})$. Then by definition of $W$, we have 
\begin{equation*}
    z e^z = \frac{\sqrt[k]{k!}}{k}.
\end{equation*}
It thus follows that
\begin{align*}
    1-e^{-c} &= 1-e^{-k z} \\
    &= 1- z^k \pth{z e^{z} }^{-k} \\
    &= 1- z^k \pth{\frac{\sqrt[k]{k!}}{k}}^{-k} \\
    &= 1- \frac{ k^k z^k }{k!} \\
    & = 1-\frac{c^k}{k!}.
\end{align*}
We conclude by noting that $\lim_{k\to \infty} W(\frac{\sqrt[k]{k!}}{k}) = W(1/e)$, so this ratio behaves asymptotically as $1-e^{-kW(1/e) + o(k)}$.  
\end{proof}

\paragraph{\textsc{Top-$1$-of-$k$} for \IID random variables. }\label{eta4} Next, we present our result for \textsc{Top-$1$-of-$k$} and \IID random variables.
First, we prove that we can assume $n\to \infty$ without loss of generality. 
\begin{lemma}
    Let $\{X_i\}_{1\leq i \leq n}$ be \IID random variables with \CDF $F$. Let $\{Y_{i, j}\}_{\substack{1\leq i\leq n \\ 1\leq j \leq K}}$ be \IID random variable with \CDF $F^{1/K}$. Let $\textsc{OPT}(X_1, \ldots, X_n)$ denote the expected value of the optimal algorithm running on $X_1, ..., X_n$ (in this order), and similarly  $\textsc{OPT}(Y_{1,1}, ..., Y_{n, K})$. Then we have 
    \begin{equation*}
        \textsc{OPT}(X_1, \ldots, X_n) \geq \textsc{OPT}(Y_{1,1}, ..., Y_{n, K}),
    \end{equation*}
    and
    \begin{equation*}
        \Ex{\max_{1\leq i \leq n} X_i} = \Ex{\max_{1\leq i \leq n, 1\leq j\leq K} Y_{i, j}}.
    \end{equation*}
    In other words, it is worse to run on $nK$ instances of \IID random variables with \CDF $F^{1/K}$ instead of $n$ instances of \IID random variables with \CDF $F$. 
\end{lemma}
\begin{proof}
    For the first claim, there is an algorithm $A$ running on $X_1, ..., X_n$ that can simply simulate the behavior of the optimal algorithm $B$ of $\{Y_{i,j}\}$. When $A$ observes the value of $X_i$. It samples $Z_1, ..., Z_K$ from the conditional distribution $F^{1/K}$ given $\max(Z_1, ..., Z_K)=X_i$. The algorithm then feeds the values of $Z_1, ..., Z_K$ into $B$. If $B$ accepts any of the random variables, then $A$ accepts $X_i$ which has value at least that of what $B$ accepted. The proof concludes by coupling $\{Y_{i,j}\}_{1\leq j\leq K}$ with $X_i$ through $X_i = \max(Y_{i, 1}, ..., Y_{i, K})$ since they have the same distribution. 

    The last statement follows because
    \[
    \Prob{\max_i X_i \leq \tau} = \prod_i \Prob{X_i\leq \tau} =  F(\tau)^n = F^{Kn/K}(\tau) = \prod_{i, j} \Prob{Y_{i,j}\leq \tau} = \Prob{\max_{i,j}Y_{i,j}\leq \tau}. 
    \]
\end{proof}
By taking $K\to \infty$, we can assume the number of random variables $\to \infty$ without loss of generality, and hence the Poissonization results follow.

\begin{algorithm}
\caption{$1-e^{-\zeta_k}$ competitive algorithm for \IID \textsc{Top-$1$-of-$k$}. }\algolab{best1ofk:iid:firstimprovement}
\begin{algorithmic}
    \State Choose $\tau$ such that $\sum_{i=1}^n \Prob{X_i \geq \tau}=\zeta_k$.
    \For{$i = 1, \ldots, n$}
        \If{$X_i > \tau$}
            \State Accept $X_i$. 
            \State $\tau \gets X_i$.
            \State If we have accepted $k$ items, break.
        \EndIf
    \EndFor
\end{algorithmic}
\end{algorithm}

\paragraph{Algorithm} See \algoref{best1ofk:iid:firstimprovement}. Let $\zeta_k$ be the unique positive solution of 
    \begin{equation*}
    1-e^{-x} = \sum_{i=0}^{k-1} e^{-x}\frac{x^i}{i!} + \sum_{i=k}^\infty\sum_{j=0}^k  e^{-x} \frac{x^i}{i!} \frac{{i+1 \brack j}}{(i+1)!},
    \end{equation*}
    where ${i \brack k}$ is the (unsigned) Stirling number of the first kind. The algorithm  sets a single threshold $\tau$ such that $\sum_{i=1}^n \Prob{X_i \geq \tau} = \zeta_k$. Every time the algorithm observes a value $v$ above $\tau$, it accepts it, and updates its new threshold to $v$. This continues until we can no longer accept values (we consumed the $k$ slots) or run out of random variables. 
\begin{lemma}
\lemlab{best1ofk:iid:firstimprovement}
     \algoref{best1ofk:iid:firstimprovement} is $1-e^{-\zeta_k}$ competitive for \IID \textsc{Top-$1$-of-$k$}.  In particular, for $k=2,3,4$, the competitive ratios are at least  $0.8520, 0.9463, 0.9816$ respectively. 
\end{lemma}
\begin{proof}
We apply stochastic dominance to compare $\Prob{\Alg \geq \ell}$ and $\Prob{Z \geq \ell}$.

\paragraph{Case 1: $\ell \in [0, \tau]$.} If at least one value exceeds $\tau$, the algorithm will select a value greater than or equal to $\tau$, and hence greater than or equal to $\ell$. This leads to
\begin{equation}
    \frac{\Prob{\Alg \geq \ell}}{\Prob{Z \geq \ell}} \geq \Prob{\Alg \geq \tau} \geq 1 - e^{-\zeta_k}. \eqlab{first:of:renyi}
\end{equation}

\paragraph{Case 2: $\ell \in (\tau, \infty]$.} Let $\ell' = \sum_{i=1}^n \Prob{X_i \geq \ell}$, with $\ell' \in (0, \zeta_k)$. Consider the subset of random variables $Y$ with values in $[\tau, \ell]$, denoted $X_{i_1}, \ldots, X_{i_Y}$. Define $R_j = 1$ if $X_{i_j} = \max(X_{i_1}, \ldots, X_{i_Y})$ and $0$ otherwise. Let $M = \sum_{j=1}^Y R_j$ be the count of right-to-left maxima. If $M < k$ and there is at least one value above $\ell$, the algorithm will choose a value $\geq \ell$. Rényi \cite{lefttoright} demonstrated that the number of permutations of $[i]$ with exactly $j \leq i$ left-to-right maxima, which is the same number of permutations with $j$ right-to-left maxima, is equal to ${i \brack j}$, the (unsigned) Stirling number of the first kind.

Formally, for $U_{\tau, \ell} = \{i : \tau \leq X_i \leq \ell \}$ and $U_{\ell} = \{i : \ell \leq X_i \}$, if $|U_{\ell}| \geq 1$ and $|U_{\tau, \ell}| = i$, then the probability of selecting a value above $\ell$ is at least
\begin{equation*}
    \sum_{j=0}^k \frac{{i+1 \brack j}}{(i+1)!}.
\end{equation*}
Therefore,
\begin{equation}
    \frac{\Prob{\Alg \geq \ell}}{\Prob{Z \geq \ell}} = \Prob{\Alg \geq \ell \mid |U_\ell|\geq 1} \geq \sum_{i=0}^{k-1} e^{-(\zeta_k - \ell')}\frac{(\zeta_k - \ell')^i}{i!} + \sum_{i=k}^\infty \sum_{j=0}^{k}  e^{-(\zeta_k - \ell')} \frac{(\zeta_k - \ell')^i}{i!}\frac{{i+1 \brack j}}{(i+1)!}. \eqlab{rhsrenyi}
\end{equation}
The derivative of \Eqref{rhsrenyi} with respect to $\ell'$ is
    \begin{equation}
        e^{-(\zeta_k -\ell')} \frac{(\zeta_k -\ell')^{k-1}}{(k-1)!} + \sum_{i=k}^\infty \sum_{j=0}^k  e^{-(\zeta_k -\ell')} \frac{(\zeta_k -\ell')^{i-1} (\zeta_k - \ell' - i)}{i!} \frac{{i+1 \brack j}}{(i+1)!}, \eqlab{rhsrenyi2} 
    \end{equation}
    and thus, assuming $k \geq \zeta_k \implies \zeta_k - \ell'-i \leq \zeta_k-k \leq 0$, 
    \begin{align*}
        \Eqref{rhsrenyi2} &\geq e^{-(\zeta_k -\ell')} \frac{(\zeta_k -\ell')^{k-1}}{(k-1)!} + \sum_{i=k}^\infty  e^{-(\zeta_k -\ell')} \frac{(\zeta_k -\ell')^{i-1} (\zeta_k - \ell' - i)}{i!}\\
        & =  e^{-(\zeta_k -\ell')} \frac{(\zeta_k -\ell')^{k-1}}{(k-1)!} - e^{-(\zeta_k - \ell')} \frac{(\zeta_k - \ell')^{k-1}}{(k-1)!} \\
        & = 0.
    \end{align*}

This implies that the right-hand side is minimized as $\ell' \to 0$. Therefore, the second main inequality becomes
\begin{equation}
    \frac{\Prob{\Alg \geq \ell}}{\Prob{Z \geq \ell}} \geq  \sum_{i=0}^{k-1} e^{-\zeta_k}\frac{\zeta_k^i}{i!} + \sum_{i=k}^\infty\sum_{j=0}^k  e^{-\zeta_k} \frac{\zeta_k^i}{i!} \frac{{i+1 \brack j}}{(i+1)!}. \eqlab{rhsrenyifinal}
\end{equation}
By stochastic dominance, combining \Eqref{first:of:renyi} and \Eqref{rhsrenyifinal}, gives the competitive ratio as
\begin{equation*}
    \min \left\{ 1-e^{-\zeta_k}, \sum_{i=0}^{k-1} e^{-\zeta_k}\frac{\zeta_k^i}{i!} + \sum_{i=k}^\infty\sum_{j=0}^k  e^{-\zeta_k} \frac{\zeta_k^i}{i!} \frac{{i+1 \brack j}}{(i+1)!} \right\} = 1-e^{-\zeta_k},
\end{equation*}
as defined by $\zeta_k$.
\end{proof}

\begin{algorithm}
\caption{$0.883$ competitive algorithm for \IID \textsc{Top-$1$-of-$2$}. }\algolab{best1of2:iid:finalalgorithm}
\begin{algorithmic}
    \State Define $\tau(t)$ threshold function as described in \lemref{best1of2:iid:finalalgorithm}. 
    \State Let $t_1, ..., t_n$ be time of arrivals of $X_1, ..., X_n$ chosen uniformly and independently from $[0, 1]$.  
    \For{$i = 1, \ldots, n$}
        \If{$X_i > \tau(t_i)$}
            \State Accept $X_i$. 
            \For{$j = i+1, \ldots, n$}
                \If{$X_j > X_i$}
                    \State Accept $X_j$
                    \State Exit (as we accepted two values)
                \EndIf 
            \EndFor
            \State Exit (We only accepted one value). 
        \EndIf
    \EndFor
\end{algorithmic}
\end{algorithm}

\paragraph{Refined Analysis for \textsc{IID Top-$1$-of-$2$}.}
We enhance the prior analysis by allowing dynamic initial thresholds over the interval $[0, 1]$. Consider a sequence of $m$ descending thresholds $\tau_1 > \tau_2 > \dots > \tau_m$, where $\tau_i$ is determined by
\begin{equation*}
    \sum_{j=1}^n \Prob{X_j \geq \tau_i} = c_i,
\end{equation*}
for a set of constants $c_1 < c_2 < \dots < c_m$. The threshold function over time is defined as $\tau(t) = \tau_{\lceil tm \rceil}$, with the algorithm adopting $\tau(t)$ until a value $v > \tau(t)$ is observed. Upon which, $v$ is accepted, followed by the next value (if any) exceeding $v$. See \algoref{best1of2:iid:finalalgorithm}. 

\begin{lemma}
\lemlab{best1of2:iid:finalalgorithm}
Given constants $0=c_0 < c_1 < c_2 < \dots < c_m$, and for any $\ell' \in (c_{j-1}, c_j]$, define
\begin{equation*}
    f_j(c_1, \dots, c_m, \ell') = 1 - \exp\pth{-\frac{1}{m}\sum_{r=1}^{j-1}c_r} + \sum_{k=j}^m \left( \exp\pth{-\frac{1}{m}\sum_{r=1}^{k-1}c_r} \cdot \frac{e^{-\frac{(m+1)c_k}{m}} \left( m e^{c_k} \ell' \left( e^{\frac{c_k}{m}} - 1 \right) \left( 2c_k - \ell' \right) + c_k \ell' e^{\frac{k c_k}{m}} \left( \ell' - c_k \right) \right)}{m c_k^2} \right).
\end{equation*}
The competitive ratio $c$ of \algoref{best1of2:iid:finalalgorithm} is at least
\begin{equation}
    c\geq \min \left\{ 1 - \exp\pth{-\frac{1}{m}\sum_{i=1}^m c_i}, \min_{1 \leq j \leq m} \inf_{\ell' \in [c_{j-1}, c_j]} \frac{f_j(c_1, \dots, c_m, \ell')}{1 - e^{-\ell'}} \right\}.  \eqlab{tighter:best1of2:apdx}
\end{equation}
In particular, we report $m=10$ thresholds in the \apdxref{E} that give a competitive ratio of at least $0.883$ to \Eqref{tighter:best1of2:apdx}. 
\end{lemma}

\begin{proof}
Denote by $\Alg$ the outcome of the strategy adhering to threshold $\tau$. We employ stochastic dominance to compare $\Prob{\Alg \geq \ell}$ and $\Prob{Z \geq \ell}$.

\textbf{Case 1:} $\ell \in [0, \tau_m]$. Define the region $U = \{ (t, v) : 0 \leq t \leq 1, v \geq \tau(t) \}$. The measure of $U$ is given by
\begin{align*}
    \mu(U) &= \sum_{i=1}^n \Prob{X_i \text{ arrives in }U } \\
    &= \sum_{i=1}^n \int_0^1 \Prob{X_i \geq \tau(t)} \dif t \\
    &= \sum_{i=1}^n \sum_{j=1}^m \frac{1}{m} \Prob{X_i \geq \tau_j} \\
    &= \sum_{j=1}^m \frac{1}{m} \sum_{i=1}^n \Prob{X_i \geq \tau_j} \\
    &= \sum_{j=1}^m \frac{1}{m} c_j, 
\end{align*}
leading to 
\begin{equation}
    \Prob{\Alg \geq \ell} \geq 1 - e^{-\mu(U)} = 1 - e^{-\frac{1}{m}\sum_{j=1}^m c_j} \geq (1 - e^{-\frac{1}{m}\sum_{j=1}^m c_j}) \Prob{Z \geq \ell}. \eqlab{best1of2:iid:curvethresholds:first}
\end{equation}

\textbf{Case 2:} For $\ell \in [\tau_j, \tau_{j-1}]$ with $1 \leq j \leq m$, we consider values $\ell$ within the thresholds $\tau_j$ and $\tau_{j-1}$. Define $U_\ell=\{(t, v) : 0\leq t\leq 1, v\geq \ell\}$, representing the region of where values exceed $\ell$. The measure of $U_\ell$, $\mu(U_\ell)$, equals $\sum_{i=1}^n \Prob{X_i \geq \ell}$.

Consider $\mathcal{B}_j = \{(t, v) : 0\leq t\leq \frac{j-1}{m}, v\geq \tau(t) \}$, which captures the region before time $\frac{j-1}{m}$ with values above the threshold function $\tau(t)$. The measure of $\mathcal{B}_j$ is computed as 
\begin{align*}
    \mu(\mathcal{B}_j) &= \sum_{i=1}^n \Prob{X_i \text{ arrives in }B_j } \\
    &= \sum_{i=1}^n \int_0^{(j-1)/m} \Prob{X_i \geq \tau(t)} \dif t \\
    &= \sum_{i=1}^n \sum_{r=1}^{j-1} \frac{1}{m} \Prob{X_i \geq \tau_r} \\
    &= \frac{1}{m}\sum_{r=1}^{j-1}c_r.
\end{align*}

Assuming $\mathcal{B}_j$ contains no realizations, we guess the first value above $\tau(t)$ at time $t^\ast \geq \frac{j-1}{m}$. For $k\geq j$, we guess $t^\ast$ within the interval $\left[\frac{k-1}{m}, \frac{k}{m}\right]$. Consider the regions
\[
\mathcal{C}_{k, t^\ast} = \{(t, v) : \frac{k-1}{m} \leq t \leq t^\ast, v\geq \tau(t)\}, 
\]
and 
\[
\mathcal{D}_{k, t^\ast} = \{(t, v) : t^\ast \leq t \leq 1, v\geq c_k\}. 
\]
Consider the event $\xi_{t^\ast}$ that the region $\mathcal{B}_k \cup \mathcal{C}_{k, t^\ast}$ is empty and that there is a realization from time $t^\ast$ to time $t^\ast+\dif t^\ast$ above $\ell$, \emph{or} that the region $\mathcal{B}_k \cup \mathcal{C}_{k, t^\ast}$ is empty, that there is a realization from time $t^\ast$ to time $t^\ast+\dif t^\ast$ above $\tau(t^\ast)$ but below $\ell$, and that the region $\mathcal{D}_{k, t^\ast}$ contains at least one value, and the first such value exceeds $\ell$. This event imply the algorithm gets a value at least $\ell$, because it either immediately succeeds in getting a value above $\ell$, or it gets a value above $\tau(t^\ast)$ and below $\ell$, but then selects the first value in $\mathcal{D}_{k, t^\ast}$ which is above $\ell$. 

Consider the probability of the region $\mathcal{D}_{t^\ast}$ being non-empty, and the first realization in the region having value at least $\ell$. This probability can be computed as 

\begin{align*}
    \int_{t^\ast}^1 e^{-(t-t^\ast)c_k} \mu(U_\ell) \dif t = (1-e^{-(1-t^\ast)c_k}) \frac{\mu(U_\ell)}{c_k}
\end{align*}

Hence, the probability of event $\xi_{t^\ast}$ can be computed as 

\begin{align}
    \underbrace{\exp{\left(-\left(\frac{1}{m}\sum_{r=1}^{k-1}c_r + \left(t^\ast - \frac{k-1}{m}\right)c_k\right)\right)}}_{\text{Probability that }\mathcal{B}_k \cup \mathcal{C}_{k, t^\ast} \text{ is empty}} \times \left( \underbrace{\mu(U_\ell) \dif t^\ast}_{\text{value }\geq \ell \text{ from } t^\ast \text{ to } t^\ast + \dif t^\ast} + \underbrace{(c_k - \mu(U_\ell)) \dif t^\ast}_{\text{value in }[\tau_k, \ell) \text{ from }t^\ast \text{ to }t^\ast + \dif t^\ast} \cdot  \underbrace{\left(1 - e^{-(1-t^\ast)c_k}\right) \frac{\mu(U_\ell)}{c_k}}_{\text{first value in }\mathcal{D}_{k, t^\ast} \text{ is }\geq \ell}\right). \eqlab{term:best1of2:dtast}
\end{align}

Integrating \Eqref{term:best1of2:dtast} over $t^\ast$ from $\frac{k-1}{m}$ to $\frac{k}{m}$ yields the probability for a fixed $k \geq j$ as:

\begin{equation*}
    e^{-\frac{1}{m}\sum_{r=1}^{k-1}c_r} \cdot \frac{e^{-\frac{(m+1) c_k}{m}} \left(m e^{c_k} \ell' \left(e^{\frac{c_k}{m}}-1\right) \left(2 c_k - \ell'\right) + c_k \ell' e^{\frac{k c_k}{m}} \left(\ell' - c_k\right)\right)}{m c_k^2}, \eqlab{term:best1of2:dtast:simplified}
\end{equation*}

where $\ell' = \mu(U_\ell) \in [c_{j-1}, c_j]$. Summing the event of $\mathcal{B}_j$ being nonempty, and \Eqref{term:best1of2:dtast:simplified} over $k \geq j$, we obtain the total probability of receiving a value exceeding $\ell$:
\begin{align}
    \Prob{\Alg \geq \ell} &\geq 1-e^{-\frac{1}{m}\sum_{r=1}^{j-1}c_r} + \sum_{k=j}^m \left(e^{-\frac{1}{m}\sum_{r=1}^{k-1}c_r} \cdot  \frac{e^{-\frac{(m+1) c_k}{m}} \left(m e^{c_k} \ell' \left(e^{\frac{c_k}{m}}-1\right) \left(2 c_k - \ell'\right)+c_k \ell' e^{\frac{k c_k}{m}} \left(\ell' - c_k\right)\right)}{m c_k^2}\right) \nonumber \\ 
    &= f_j(c_1, ..., c_m, \ell') = \frac{f_j(c_1, ..., c_m, \ell')}{1-e^{-\ell'}}\Prob{Z \geq \ell'}. \eqlab{best1of2:iid:curvethresholds:second}
\end{align}

Therefore, using stochastic dominance by combining \Eqref{best1of2:iid:curvethresholds:first} and minimizing \Eqref{best1of2:iid:curvethresholds:second} over $\ell'\in [c_{j-1}, c_j]$, the competitive ratio is thereby bounded from below by:

\begin{equation*}
    \min \left\{ 1-e^{-\frac{1}{m}\sum_{i=1}^m c_i}, ~~ \min_{1\leq j \leq m} \inf_{\ell' \in [c_{j-1}, c_{j}]} \frac{f_j(c_1, ..., c_m, \ell')}{1-e^{-\ell}}  \right\}.
\end{equation*}
\end{proof}

\begin{algorithm}
\caption{$1-k^{-k/5}$ competitive algorithm for \IID \textsc{Top-$1$-of-$k$}. }\algolab{best1ofk:iid:lastimprovement}
\begin{algorithmic}
    \State Choose $\tau$ such that $\sum_{i=1}^n \Prob{X_i \geq \tau}=L = e^{\sqrt{k}}$.
    \For{$i = 1, \ldots, n$}
        \If{$X_i > \tau$}
            \State Accept $X_i$. 
            \State $\tau \gets X_i$.
            \State If we have accepted $k$ items, break.
        \EndIf
    \EndFor
\end{algorithmic}
\end{algorithm}

\paragraph{Specializing the algorithm for general $k$.} It is difficult to express the competitive ratio of the algorithms above for general $k$. Here, we give an explicit super-exponential dependence on $k$. \algoref{best1ofk:iid:lastimprovement} sets an initial threshold $\tau$, chosen so that $\sum_{i=1}^n \Prob{X_i \geq \tau} = L = e^{\sqrt{k}}$. Upon encountering a value $v$ that surpasses $\tau$, the algorithm selects $v$ and updates $\tau$ to this new $v$. This selection process continues until either $k$ values have been chosen or all random variables have been examined.

\begin{lemma}
\lemlab{best1ofk:iid:improvement:generalk}
    \algoref{best1ofk:iid:lastimprovement} achieves a competitive ratio of at least $1 - k^{-k/5}$.
\end{lemma}
\begin{proof}
    We apply stochastic dominance to compare $\Prob{\Alg \geq \ell}$ against $\Prob{Z \geq \ell}$.
    
    \paragraph{Case 1: $\ell \in [0, \tau]$.} The presence of any value exceeding $\tau$ ensures the algorithm will choose a value greater than $\tau$, and thus $\ell$. This leads to:
    \begin{equation*}
        \frac{\Prob{\Alg \geq \ell}}{\Prob{Z \geq \ell}} \geq \Prob{\Alg \geq \ell} \geq 1 - e^{-L} \geq 1 - k^{-k/5}.
    \end{equation*}
    
    \paragraph{Case 2: $\ell \in [\tau, \infty]$.} Define $\ell' = \sum_{i=1}^n \Prob{X_i \geq \ell}$, with $\ell'$ ranging in $(0, L)$. Consider the $Y$ random variables with values within $[\tau, \ell]$, labeled as $X_{i_1}, \ldots, X_{i_Y}$. Let $R_j = 1$ if $X_{i_j}$ is the maximum among $X_{i_1}, \ldots, X_{i_j}$, and $0$ otherwise. The sum $M = \sum_{j=1}^Y R_j$ counts the right-to-left maxima. Moreover, $R_1, ..., R_Y$ are independent. If $M < k$ and at least one value exceeds $\ell$, the algorithm will select a value $\geq \ell$. The probability that $M$ reaches $k$ can be bounded by $k^{-k/5}$ using a standard Chernoff bound. Hence:
    \begin{equation*}
        \frac{\Prob{\Alg \geq \ell}}{\Prob{Z \geq \ell}} \geq \frac{\Prob{M < k} (1 - e^{-\ell'})}{1 - e^{-\ell'}} \geq 1 - k^{-k/5}.
    \end{equation*}
    The lemma is then established through stochastic dominance.
\end{proof}

\section{Prophet Secretary Non-\IID Case. }
\label{eta5}
\seclab{main}
We now go back to the non \IID prophet-secretary. In \cite{csz-pstbs-18}, Correa, Saona, and Ziliotto used Schur-convexity to study a class of algorithms known as \textit{blind} algorithms. In particular, they consider discrete blind algorithms. The algorithm is characterized by a \textit{decreasing} threshold function $\alpha: [0, 1]\to [0,1]$. Letting $\ICDFY{Z}{q}$ denote the $q$-th quantile of the maximum distribution (i.e., $\Prob{Z \leq \ICDFY{Z}{q}}=q$), the algorithm accepts realization $v_i$ if $v_i \geq \ICDFY{Z}{\alpha(i/n)}$ (i.e., if it is in the top $\alpha(i/n)$ quantile of $Z$). They characterized the competitive ratio $c$ of an algorithm that follows threshold function $\alpha$ (as $n\to \infty$) as \cite{csz-pstbs-18}
\begin{align}
    c \geq \min \pth{ 1-\int_{0}^1 \alpha(x)\dif x , \min_{x\in [0,1]} \pth{ \int_{0}^x \frac{1-\alpha(y)}{1-\alpha(x)}  \dif y + \int_{x}^1 e^{\int_0^y \log \alpha(w) dw } \dif y}   } \eqlab{Raimundo}.
    \end{align}
Looking at  \Eqref{Raimundo}, the reader might already see many parallels with \Eqref{monstrosity}, even though one is based on quantiles of the maximum, and the other is based on summation thresholds. Correa \etal  resorted to numerically solving a stiff, nontrivial optimal integro-differential equation. They find an $\alpha$ function such that $c \geq 0.665$ (and then resorted to other similar techniques to show the main $0.669$ result). They also showed than no blind algorithm can achieve a competitve ratio above $0.675$.

\subsection*{New analysis for the Non-\IID Case}
\begin{figure}
    \centering
    \includegraphics[width=0.75\textwidth]{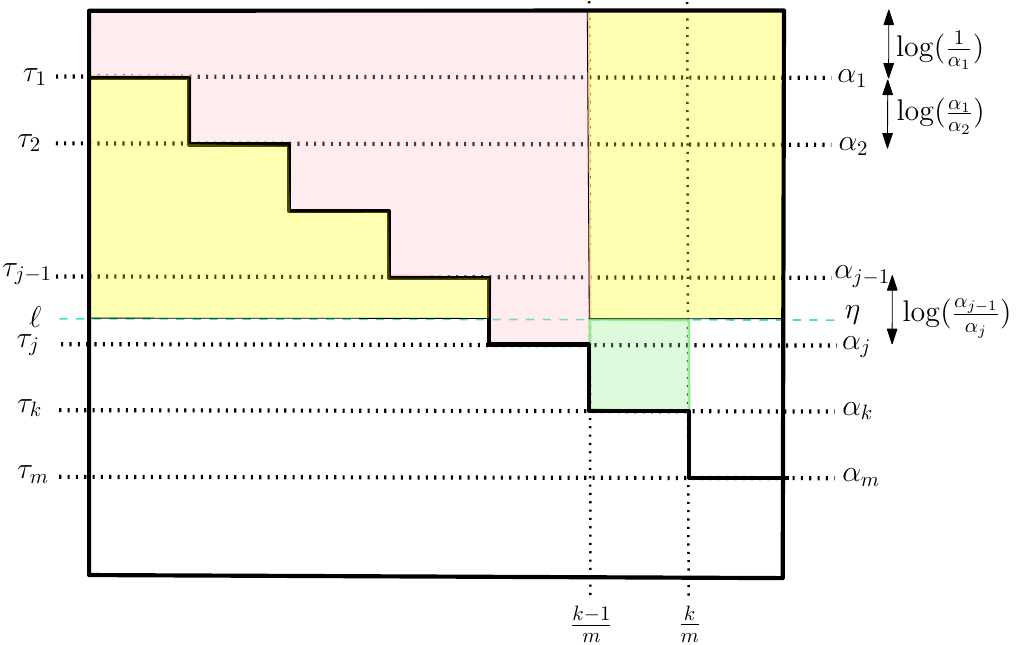}
    \caption{Analysis visualization of \lemref{prophet:secretary:improvement}}
    \figlab{mainproofidea:fig}
\end{figure}

\seclab{mainanalysis}
\paragraph{Algorithm} Using the shards machinery developed thus far, we introduce the new analysis for the prophet secretary. The algorithm employs a straightforward strategy with $m=16$ thresholds $\tau_1 > \ldots > \tau_m$. The threshold function $\tau : [0,1] \to \Re_{\geq 0}$ is defined as $\tau(t) = \tau_{\lceil tm \rceil }$, making $\tau$ a step function. The algorithm selects the first realization $(t_i, v_i)$ satisfying $v_i \geq \tau(t_i)$. It is important to note that the thresholds $\tau_i$ are chosen such that
\begin{equation*}
    \Prob{Z \leq \tau_i} = \alpha_i,
\end{equation*}
where $\alpha_0 = 1 > \alpha_1 > \ldots > \alpha_{m+1} = 0$.

\begin{lemma}
\lemlab{prophet:secretary:improvement}
        The competitive ratio $c$ of the above algorithm satisfies 
    \begin{equation}
            c \geq \min_{1\leq j \leq m+1} \min_{\alpha_{j} \leq \nu \leq \alpha_{j-1}}\frac{f_j({\alpha_1, ..., \alpha_m} , \eta)}{1-\eta}
            \eqlab{prophet-sec-expr},
    \end{equation}
    where
    \begin{align}
    f_j({\alpha_1, ..., \alpha_m}, \eta) &=\frac{1}{m}\sum_{k=1}^{j-1} \pth{1-\alpha_k} + \sum_{k=j}^m \pth{ \prod_{\nu=1}^{k-1} \alpha_{\nu} }^{\frac{1}{m}} w_k q_{t, k} \eqlab{fjterms}~~, \\
    w_k &= \sum_{\nu=0}^{j-1} e^{-s_{\nu} } (1-e^{-r_{\nu} })\frac{1}{m-(k-1)+\nu}  \eqlab{wkexpr}~~, \\
    r_{\nu} &=\frac{m-(k-1)+\nu}{m}\log \frac{\hat{\alpha}_{\nu}}{\hat{\alpha}_{\nu+1}} \eqlab{rkexpr} ~~,\\
    \hat{\alpha_\nu} &= \alpha_\nu \text{ if }\nu\leq j-1 \text{ and }\eta \text{ if }\nu=j \eqlab{alphakexpr}~~, \\
    s_\nu &= \sum_{\beta=0}^{\nu-1} r_\beta \eqlab{skexpr} ~~,\\
    q_{\eta, k} &= \sum_{\beta=0}^{\infty} e^{-\frac{1}{m}\log \frac{\eta}{\alpha_k} }\pth{\frac{1}{m}\log \frac{\eta}{\alpha_k}}^\beta \frac{1}{\beta!} \frac{1}{\beta+1} = \frac{1- \pth{\frac{\alpha_{k}}{\eta}}^{1/m} }{\frac{1}{m}\log \pth{ \frac{\eta}{\alpha_k} } } \eqlab{unstable}.
\end{align}

\end{lemma}

\begin{proof}
    We would like to compare $\Prob{Z\geq \ell}$ vs $\Prob{\Alg \geq \ell}$ as before. For this, we again break the analysis on where $\ell$ lies. 
    
    \paragraph{For $\ell\in [0, \tau_m)$} See \figref{mainproofidea:fig}. We use the trivial upper bound $\Prob{Z\geq \ell}\leq 1$. On the other hand, consider when $\Alg \geq \ell$. Using \lemref{discrete_blind_lb}, we have $\Prob{\Alg \geq \ell} \geq \frac{1}{m}\sum_{i=1}^m \left( 1-\alpha_i\right)$. Hence
\[
\frac{\Prob{\Alg \geq \ell}}{\Prob{Z\geq \ell}} \geq \frac{1}{m}\sum_{i=1}^m \left( 1-\alpha_i\right).
\]
This case is captured in \Eqref{prophet-sec-expr} when $j=m+1$ and $\eta=\alpha_{m+1}=0$. 

\paragraph{The case of $\ell \in [\tau_{j}, \tau_{j-1}]$}
Again, see \figref{mainproofidea:fig}. Let $\eta = \Prob{Z\leq \ell}$. Note $\eta \in [\alpha_j, \alpha_{j-1}]$ and $\Prob{Z> \ell}=1-\eta$. 
 Now, we lower bound $\Prob{\Alg \geq \ell}$. We again give an event $\xi$ on the shards that implies $\Alg \geq \ell$ and with $\Prob{\xi}=f_j({\bf \alpha}, \eta)$. This would imply by stochastic dominance the result. We recommend looking at \figref{mainproofidea:fig} throughout the explanation. 

Formally, $\xi$ consists of a \textit{disjoint union} of $m-j+2$ events. Let $T$ be the time (if any) when we select a value according to our strategy. The first event $\chi$ is the event that $T\leq j-1$. This would imply that the algorithm received a value at least $\tau_{j-1} \geq \ell$. 

    The subsequent $m - j + 1$ events, denoted as $\zeta_k$ for $j \leq k \leq m$, are defined such that the region
    \begin{equation*}
        A_{k} = \{(t, v) : 0 \leq t \leq \frac{k - 1}{m}, v \geq \tau(t) \}, 
    \end{equation*}
    (illustrated in pink in \figref{mainproofidea:fig}) lacks shards, whereas the region
    \begin{equation*}
        B_\ell = \{(t, v) : 0 \leq t \leq 1, v \geq \ell \} \setminus A_k,
    \end{equation*}
    (depicted in yellow in \figref{mainproofidea:fig}) contains at least one shard, and the highest value shard $(t^\ast, v^\ast) \in B_\ell$ appears within the time frame $t^\ast \in [\frac{k - 1}{m}, \frac{k}{m}]$, prior to any shard within the region
    \begin{equation*}
        C_k = \{(t, v) : \frac{k - 1}{m} \leq t \leq \frac{k}{m}, \tau(t) \leq v < \ell \},
    \end{equation*}
    (portrayed in green in \figref{mainproofidea:fig}).

    The event $\zeta_k$ demands further clarification. Absence of shards in $A_k$ (the pink region) permits the algorithm to proceed until time $(k - 1)/m$. The highest value shard $(t^\ast, v^\ast)$ in $B_\ell$ (the yellow region) represents a genuine realization of some $X_i$. This shard should arrive between time $(k - 1)/m$ and $k/m$. However, shards within $C_k$ (the green region) could potentially correspond to genuine realizations of some $X_j$ and precede $v^\ast$, leading the algorithm to opt for a value below $\ell$ in such instances. Therefore, it is imperative for $t^\ast$ to precede all shards within $C_k$, ensuring the algorithm secures an actual realization valued $\geq \ell$. Notably, all these events $\chi, \zeta_j, ..., \zeta_m$ are mutually exclusive.

Invoking \lemref{discrete_blind_lb}, the probability of $\chi$ happening is at least $\frac{1}{m}\sum_{k=1}^{j-1} 1-\alpha_k$. This is the first term in the RHS of \Eqref{fjterms}. 
Next, we compute the probability of $\zeta_k$. The probability that $A_k$ is devoid of shards is at least $\left(\prod_{\nu=1}^{k-1}\alpha_\nu\right)^{1/m}$, as per \lemref{simplified-upper}. This corresponds to the first factor in the RHS of \Eqref{fjterms}. The term $w_k$ in \Eqref{fjterms} and defined in \Eqref{wkexpr} denotes the probability that the highest value shard within $B_\ell$ (the yellow region) arrives between $(k-1)/m$ and $k/m$. We unpack the expression here. Specifically, $\hat{\alpha}_\nu$ is used in lieu of $\alpha_\nu$, defined in \Eqref{alphakexpr}, because the yellow region is bounded below by $\ell$ rather than $\tau_j$, necessitating the use of $\hat{\alpha}_j = \eta$ instead of $\alpha_j$. For $r<j$, $\hat{\alpha}_r = \alpha_r$ is used. Now consider the ``yellow semirow'' 
\begin{equation*}
        R_\nu = \{(t, \hat{v}) : 0\leq t \leq 1 \text{ and } \max(\ell, \tau_{\nu}) \leq \hat{v} \leq \tau_{\nu+1} \} \setminus A_k. 
\end{equation*}

$R_\nu$ has a Poisson rate of $r_\nu$ as defined in \Eqref{rkexpr}, as the length of the yellow semirow is $(m-(k-1)+\nu)/m$, and the entire row's Poisson rate is $\log \frac{\hat{\alpha}_\nu}{\hat{\alpha}_{\nu+1}}$. Similarly, $s_\nu$ defined in \Eqref{skexpr} denotes the Poisson rate of $\bigcup_{j<\nu}R_j$. The probability $w_k$ is derived by requiring some row $R_\nu, 0\leq \nu \leq j-1$, to contain at least one shard, while requiring $\bigcup_{j<\nu}R_j$ to be empty, and the highest value shard in $R_\nu$ to appear between time $(k-1)/m$ and $k/m$. This combined with $A_k$ being empty implies that there is some actual realization $X_i$ that arrives from time $(k-1)/m$ to $k/m$.

Finally, we unpack $q_{\eta, k}$. Since we have already conditioned that $A_k$ (red region) has no shards, the maximum value shard $(t^\ast, v^\ast)$ in $B_\ell$ (yellow region) arrives in time $(k-1)/m$ to $k/m$, then we only need to ensure that $t^\ast$ arrives before all shards in the region $C_k$, which is the expression $q_{\eta, k}$. The region $C_k$ has Poisson rate $\frac{1}{m}\log \frac{\eta}{\alpha_k}$. We count how many shards $0\leq \beta \leq \infty$ are in $C_k$, and require that $v^\ast$ arrives before all of them, which happens with probability $1/(\beta+1)$. 
\end{proof}

\paragraph{Optimization} The right hand side of \Eqref{prophet-sec-expr} can be maximized for $\boldsymbol \alpha$ satisfying $\alpha_0=1 > \alpha_1 > \ldots > \alpha_m > \alpha_{m+1}=0$. We used Python to optimize the expression and report $m=16$ alpha values in \apdxref{C} with $c\geq \constant$. All computations were done with doubles using a precision of 500 bits (instead of the default 64). We finally obtain the main result. 

\begin{theorem}
    There exists an $m=16$ threshold blind strategy for the prophet secretary problem that achieves a competitive ratio of at least $\constant$. 
\end{theorem}

\begin{remark}
    The function $\frac{1- \pth{\frac{\alpha_{k}}{\eta}}^{1/m} }{\frac{1}{m}\log \pth{ \frac{\eta}{\alpha_k} } } $ in \Eqref{unstable} is  numerically unstable for close values of $\alpha_k, \eta$. To resolve this, we \textit{lower bound} it by truncating the summation on the RHS to $30$ terms (instead of $\infty$) and use that as a lower bound on $q_{\eta,k}$. This is referred to as ``stable\_qtk'' in the code. 
\end{remark}

\paragraph{Parameter optimization is not sufficient} Why does the above analysis yield a better competitive ratio for continuous blind strategies? It is important to stress that the set of $m=16$ parameters we derive would \emph{not} improve the analysis from \cite{csz-pstbs-18} from $0.669$ to $0.6724$; in fact, they give a \emph{worse} bound of $0.6675$! Thus it would be incorrect to suggest that we obtain a better competitive ratio because we simply found a better set of parameters. In particular, the constants $f_j(\alpha_1, ..., \alpha_m, \eta)$ we derive are significantly tighter than the $f_j(\alpha_1, ..., \alpha_m)$ that Correa \etal derive. This is because the new bounds utilize all aspects of the geometry involved as seen in the proof. In contrast, the work in \cite{csz-pstbs-18} do this separately using algebraic tools.  Hence we are optimizing for different objectives. 

\section{\IID \textsc{Semi-Online}.}
\label{eta6}
\seclab{iidsemionline}
In this section, we improve the $\approx 0.869$ competitive ratio result from \cite{hs-ttsopi-23} and give a $\approx 0.89$ competitive ratio algorithm for the \IID \textsc{Semi-Online} problem.  As a reminder from the introduction,  in this variant of the prophet inequality problem, the actual values of the variables remain undisclosed. Instead, the gambler is allowed to make $n$ queries, each asking whether ``$X_i \geq \tau_i$'' for a chosen $\tau_i$, which can be determined adaptively. Each random variable is eligible for only \emph{one} query. After all $n$ queries have been exhausted, the gambler selects the variable that holds the highest conditional expectation. Here, as in \cite{hs-ttsopi-23}, we are assuming $X_1, ..., X_n$ are \IID and $n\to \infty$. 

It is worth taking a moment to recap the algorithm from \cite{hs-ttsopi-23}. As a reminder, their algorithm defines thresholds $\tau_1< \tau_2 < \tau_3 < \tau_4=\infty$. It then runs \algoref{semionline:1}. 

\begin{algorithm}
\caption{\IID Semi-Online Algorithm \cite{hs-ttsopi-23}}\algolab{semionline:1}
\begin{algorithmic}[1]
    \State Set $r \gets 1$ and $i^\ast \gets 1$
    \For{$i = 1, \ldots, n$}
        \If{$X_i \geq \tau_r$}
            \State $r \gets r + 1$
            \State $i^\ast \gets i$
        \EndIf
    \EndFor
    \State \Return $X_{i^\ast}$
\end{algorithmic}
\end{algorithm}

Intuitively, the algorithm "raises" its threshold every time a positive response to a query is received, targeting a higher conditional expectation. In their work, \cite{hs-ttsopi-23} optimize the parameters as quantiles of the maximum, selecting $\tau_1 = \Xi(2.035135)$, $\tau_2 = \Xi(0.5063)$, and $\tau_3 = \Xi(0.05701)$, which results in an algorithm that is approximately $0.869$ competitive. The analysis for $\Prob{\Alg \geq \ell}$ is nuanced, considering that the presence of a realization exceeding $\ell$ does not guarantee that subsequent realizations won't fall between a higher threshold and $\ell$, potentially leading to the last successful realization being below $\ell$. Essentially, it's crucial to ensure that the final realization that passes (i.e., receives a "yes" response) is indeed above $\ell$.

A limitation of the current algorithm is its performance when the initial $n/2$ tests fall below $\tau_1$ (happening with a constant probability), and thus reducing the likelihood of later realizations surpassing $\tau_1$. Consequently, the algorithm may fail to achieve any success in its later stages, rendering $\tau_2$ and $\tau_3$ unused.

To address this issue, we employ a similar strategy but perhaps counter intuitively with non-increasing functions. Specifically, we define $k$ non-increasing \textbf{functions} $\tau_1, \tau_2, \ldots, \tau_k: [0, 1] \to \Re_{\geq 0}$, with $\tau_{k+1} = \infty$. The algorithm is then adapted by replacing line 3 in \algoref{semionline:1} with "If $X_i \geq \tau_r(t_i)$", where $t_i$ represents the arrival time of $X_i$.

A standard analysis of this algorithm would be exceedingly tedious, necessitating case-by-case analysis due to the dependencies that arise upon conditioning on the presence of $t$ points within a certain quantile range of the distribution, leading to several complicated nested summations. Indeed, even the application of constant functions (as proposed by \cite{hs-ttsopi-23}) introduces technical challenges, even with just two thresholds.

In this section, we demonstrate how utilizing Poissonization and dynamic programming enables us to establish a lower bound on the competitive ratio. It's important to note that the results from \cite{hs-ttsopi-23} assumes $n \to \infty$, an assumption we also adopt here.

\subsection{Dynamic programming to compute the competitive ratio}
\paragraph{The algorithm} To simplify the exposition, we will have $k$ threshold functions $\tau_1, ..., \tau_k$ which are all decreasing \textbf{step} functions. In particular, for some $p\in \mathbb{N}_{\geq 1}$, from time $(i-1)/p$ to time $i/p$ for $1\leq i\leq p$, the threshold for $\tau_j(x)$ will be $\Xi(c_{ij})$. Hence, we are optimizing for $kp$ parameters $\{c_{i,j}\}$. 

\paragraph{Finely discretizing time} Even with Poissonization, the \textbf{exact} analysis of such strategy would still be painful and involve several nested summations. To counter this, we break time into discretized chunks of $1/m$ using a clock (with $m\in \mathbb{N}_{\geq 1}$). We use the modified \algoref{semionline:2}. 
\begin{algorithm}
\caption{Modified \IID Semi-Online Algorithm}\algolab{semionline:2}
\algolab{}
\begin{algorithmic}[1] 
    \State set $r \leftarrow 1$ and $i^\ast \leftarrow 1$
    \State \textsc{clock} $\leftarrow 0$
    \For{$i = 1, \ldots, n$}
        \If{$X_i \geq \tau_r(t_i)$ and $t_i \geq \textsc{clock}$}
            \State $r \leftarrow r + 1$
            \State $i^\ast \leftarrow i$
            \State \textsc{clock} $\leftarrow \lceil m t_i \rceil / m$
        \EndIf
    \EndFor
    \State \Return $X_{i^\ast}$
\end{algorithmic}
\end{algorithm}

In particular, once we see a value above $\tau_r(t_i)$, we ``skip'' the time to the next multiple of $1/m$. As $m$ increases, the performance of the algorithm should mimic the continuous counterpart. We also insure that $ p | m$ ($p$ divides $m$) so that the discretized times are aligned with the $p$ phases of any of the functions $\tau_j$. 

\paragraph{Dynamic Program} Let us fix $\ell$ and aim to compute $\Prob{\Alg \geq \ell}$ for applying stochastic dominance. Define $\textsc{Prob}[b, j, i]$ for $b \in \{\textsc{T}, \textsc{F}\}$, $0 \leq i < m$, and $1 \leq j \leq k+1$ as the probability that the last successful test among variables arriving from time $t = i/m$ to $t = 1$ is above $\ell$, under the current use of threshold function $\tau_j$, conditioned on whether the last successful query we saw (if any) before time $i/m$ was (or was not) above $\ell$ (indicated by $b = \textsc{T}$ or $b = \textsc{F}$, respectively).

For instance, $\textsc{Prob}[\textsc{T}, 2, 120]$ represents the probability that the last successful test among variables arriving from time $t = 120/m$ to $t = 1$ is above $\ell$, given that we are currently using threshold function $\tau_2$ and that we have seen a successful query above $\ell$ before time $120/m$. Notably, $\textsc{Prob}[\textsc{F}, 1, 0]$ equals $\Prob{\Alg \geq \ell}$.

\paragraph{Recurrence} 
\begin{lemma}
\lemlab{recurrence}
Define $C[j, i] = \sum_{\beta=1}^n \Prob{X_\beta \geq \tau_j(i/m)}$ and $\ell' = \sum_{\beta=1}^n \Prob{X_\beta \geq \ell}$. The computation of $\textsc{Prob}[b, j, i]$ is subject to the following recurrence: if $i \geq m$ or $j = k+1$, then $\textsc{Prob}[\textsc{T}, j, i] = 1$ and $\textsc{Prob}[\textsc{F}, j, i] = 0$. Otherwise, we have
\begin{align*}
    \textsc{Prob}[b, j, i] = \begin{cases}
    e^{-C[j, i]/m} \textsc{Prob}[b, j, i+1] + (1 - e^{-C[j, i]/m}) \left( \frac{\ell'}{C[j, i]} \textsc{Prob}[\textsc{T}, j+1, i+1] \right. \\ 
    \left. + \frac{C[j, i] - \ell'}{C[j, i]} \textsc{Prob}[\textsc{F}, j+1, i+1] \right) & \text{if } \ell' \leq C[j, i], \\
    \\
    e^{-C[j, i]/m} \textsc{Prob}[b, j, i+1] + (1 - e^{-C[j, i]/m}) \textsc{Prob}[\textsc{T}, j+1, i+1] & \text{if } \ell' > C[j, i].
\end{cases}
\end{align*}
In particular, we can compute $\textsc{Prob}(\textsc{F}, 1, 0)$ in $O(mk)$ time. 
\end{lemma}

Before we prove \lemref{recurrence}, we need the following auxiliary lemma.  
\begin{lemma}
\lemlab{poisson_nice_property}
    Let $\tau_1=\Xi(\ell_1)$ and $\tau_2=\Xi(\ell_2)$ for $\ell_1 < \ell_2$. Then 
    \[
    \ProbCond{\text{The first realization above $\tau_2$ is also above }\tau_1}{\text{There is a realization above }\tau_2} = \frac{\ell_1}{\ell_2} 
    \]
\end{lemma}
\begin{proof}
    The conditional probability is 
    \[
    \frac{\int_0^1 e^{-\ell_2 x }\ell_1 \dif x}{1-e^{-\ell_2}} = \frac{\ell_1}{\ell_2}.
    \]
\end{proof}

Finally, we are able to prove \lemref{recurrence}
\begin{proofof}{ \lemref{recurrence}}
    The base cases are straightforward. Consider the scenario where $\ell' \leq C[j, i]$. This presents us with three distinct cases.
    \begin{enumerate}
        \item In the absence of any realization from $i/m$ to $(i+1)/m$ exceeding $\tau_j(i/m)$, which occurs with probability $e^{-C[j, i]/m}$, the probability in question is simply $\textsc{Prob}[b, j, i+1]$.
        \item If there is a realization exceeding $\tau_j(i/m)$, with the first such realization surpassing $\ell$—an event with probability $(1-e^{-C[j, i]/m})\frac{\ell'}{C[j, i]}$ as per \lemref{poisson_nice_property}—the last successful test will be above $\ell$. Consequently, we transition to $\tau_{j+1}$ and time $(i+1)/m$ immediately, in line with the clock mechanism. This scenario aligns with $\textsc{Prob}[\textsc{T}, j+1, i+1]$.
        \item If there is a realization surpassing $\tau_j(i/m)$, but the first realization falls below $\ell$, then the most recent successful test is now deemed to be below $\ell$. Hence, we progress to $\tau_{j+1}$ starting from time $(i+1)/m$, which is represented by $\textsc{Prob}[\textsc{F}, j+1, i+1]$.
    \end{enumerate}
    If $\ell'>C[j, i]$, then if there is no realization above $\tau_j(i/m)$, we continue to $\textsc{Prob}[b, j, i+1]$. Finally, if there is a realization, then the last realization is above $\ell$ now, and we proceed with $\textsc{Prob}[\textsc{T}, j+1, i+1]$. 
\end{proofof}

\paragraph{Optimization} 

\begin{figure}
    \centering
    \includegraphics[width=0.7\textwidth]{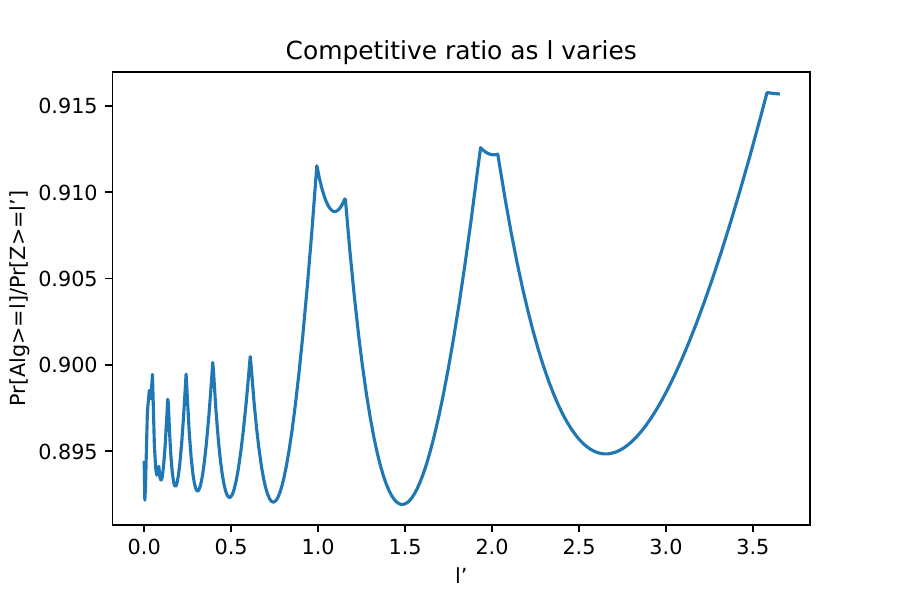}
    \caption{$\Prob{\Alg \geq \ell}/\Prob{Z\geq \ell}$ for $\ell'$ from $0$ to $\max_{i,j}(c_{i,j})$}
    \figlab{fig:iid-semionline}
\end{figure}
\begin{figure}
    \centering
    \includegraphics[width=0.7\textwidth]{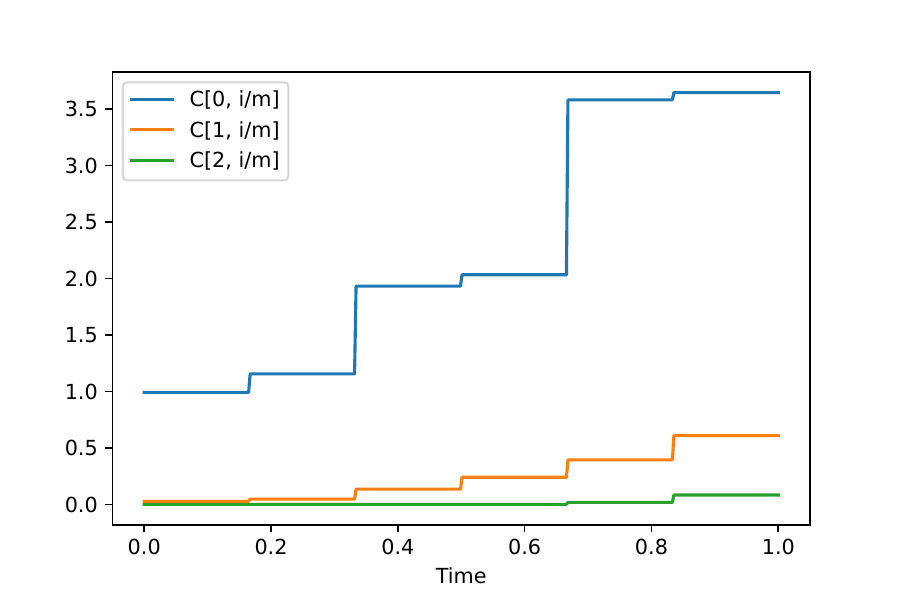}
    \caption{The threshold functions as time varries. Note that the values displayed here are the $c_{j, i}$ (i.e., on expectation, this is how many points should be above the threshold function at this time). The blue, orange, and green thresholds correspond to $\tau_1, \tau_2, \tau_3$ respectively.}
    \figlab{fig:iid-semionline-thresholds}
\end{figure}

We set $m=420$, $p=6$, and $k=3$. Hence we are optimizing for $kp=18$ parameters $c_{i,j}$. In \apdxref{D}, we provide the set of parameters we use along with the code. See \figref{fig:iid-semionline} for the plot of the competitive ratio as $\ell'$ varies from $0$ to $\max_{i,j}(c_{i,j})$. The minimum is at least $0.89$. 

\paragraph{Formal error verification.} For a fixed $\ell' \in [0, \max_{i, j}c_{i,j}]$, define $f(\ell') = \textsc{Prob}(\textsc{F}, 1, 0) - 0.8901 \cdot (1-e^{-\ell'})$. We would like to bound $\frac{d^2f}{d \ell'^2}$. To this effect, define $D_1[b, j, i], D_2[b, j, i]$ as the first and second derivative respectively of $f(\ell')$ with respect to $\ell'$. We can compute both using the recurrence, and bound $|D_2(\ell')|$. Formal error bounds can then be done by standard discretizing ideas. We verified the bounds on $(0, \max_{i,j}c_{i,j}]$ with discrete intervals of size $\epsilon =0.00015$, which ensures an error of $<10^{-6}$ in our claimed competitive ratio.  

\section{\IID and non \IID \textsc{Semi-Online-Load-Minimization}. }
\label{eta7}
\seclab{loadminimization}
We briefly recap the problem. In this setting, we are allowed to ask $n$ queries in total, but a variable can be asked multiple queries. The maximum time any variable is asked is the load. The objective is to find a $1-o(1)$ competitive algorithm in this setting while minimizing the load. \cite{hs-ttsopi-23} give an algorithm with $O(\log n)$ load for \IID random variables, and leave the non \IID case as a future problem. 

In this section, we give an $O(\log^\ast n)$ load algorithm for the non \IID case, hence also improving on the \IID load. 

\paragraph{Bruteforce} If we have a small number of random variables $Y_1, ..., Y_r$, then we can find the maximum with $O(r)$ expected queries and $O(r)$ expected maximum load. We can find which of two random variables (say $Y_1, Y_2$) are larger using $O(1)$ calls on expectation. We query with $\tau$, set to be the median of $Y_1$. Then with probability $1/2 \Prob{Y_2\geq \tau} + 1/2 \Prob{Y_2 \leq \tau} = 1/2$, the realizations are on different sides and we are done in one iteration. However, if the query answers ``yes'' or ``no'' to both, then we update $Y_1, Y_2$ to be the new conditional distributions on this information (for example, if both are ``yes'', then we update the variables to be $Y_1 | Y_1 \geq \tau, Y_2 | Y_2 \geq \tau$), and repeat this process. With probability $1/2^i$, we are done in $i$ iterations. So after $2=O(1)$ expected calls, we know which random variable is larger. Now we apply this process iteratively to $Y_1, ..., Y_r$ using on expectation $O(r)$ queries and load. 

\paragraph{Algorithm} Uniformly sample $\sigma\in \mathbb{S}_n$. First, we throw away $n'=\lceil \sqrt{n} \rceil $ variables, $X_{\sigma(1)}, ..., X_{\sigma(n')}$. We now have $n-n'$ random variables $X_{\sigma(n'+1)}, ..., X_{\sigma(n)}$ and an extra budget of $n'$ queries to use for these random variables. Next, we shard the random variables $X_{\sigma(n')}, ..., X_{\sigma(n)}$ into $\{Y_{\sigma(i) j}\}$. We define $\tau_1$ such that $\sum_{i=n'+1}^n \sum_{j=1}^K \Prob{Y_{\sigma(i),j} \geq \tau_1}= c\log n$ for a sufficiently large constant $c$. 

\paragraph{Log reduction} For $X_{\sigma(n'+1)}, ..., X_{\sigma(n)}$ we first use the threshold $\tau_1$ described above. If at least one query answer is ``yes'', then we continue to the next iteration by including only the random variables that answered yes. In iteration $t$, we use the threshold $\Xi(c \log^{(t)}n) $, the  $\log$ function nested $t$ times (for example $\log^{(2)}n=\log \log n$). By sharding and Poissonization, if we are in iteration $t$, then with probability $1-e^{-c \log^{(t)}n} = 1-\frac{1}{\left( \log^{(t-1)}n \right) ^c}$, we continue to the following iteration, and with probability $\frac{1}{\left( \log^{(t-1)}n \right) ^c}$, the answer will be ``no'' for all random variables being considered (since none are above the new threshold). In that case, we run the bruteforce solution using $O(\log^{(t)} n)$ queries and load on expectation.  So in total, the maximum load on any random variable is on expectation

\[
O(\log^\ast n) + \sum_{t=1}^{O(\log^\ast n)} \frac{O(\log^{t}n)}{O(\log^{(t-1)}n)^c} = O(\log^\ast n).
\]

Clearly, the algorithm always succeeds if $X_{\sigma(n'+1) }, ..., X_{\sigma(n)}$ contains the maximum realization from $X_1, ..., X_n$, which happens with high probability. We now make this more formal. 

\begin{lemma}
    \[
    \Ex{\Alg} \geq (1-\frac{1}{n^{O(1)}}) \Ex{Z}~,
    \]
    where $\Alg$ is the value returned by the algorithm. 
\end{lemma}
\begin{proof}
We have that 
    \[
    \Ex{\Alg} = \sum_{z\in [0, \infty)} \mathds{1}_{z\geq \tau_1}z \Prob{\Alg \text{ selects the maximum }|Z=z}.
    \]
For $z\geq \tau_1$, with probability at least $\geq 1-n'/n = 1-1/n^{O(1)}$), the maximum is in $X_{\sigma(n'+1)}, ..., X_{\sigma(n)}$ and the algorithm succeeds in finding it. So we have    
\[
    \Ex{\Alg} \geq (1-1/n^{O(1)}) \sum_{z\in [0, \infty)} \mathds{1}_{z\geq \tau_1}z  = (1-1/n^{O(1)}) \Prob{Z\geq \tau}\Ex{Z} \geq (1-1/n^{O(1)})(1-1/n^{O(1)}) \Ex{Z}.
    \]
The result follows. 
\end{proof}

\section{Conclusion and future work. }
\seclab{conclusion}
The main ingredient in all our analysis is breaking the non \IID random variables into shards (in the case of non \IID random variables), and arguing about the competitive ratio of the algorithm using events on the shards, rather than on the random variables directly. This is possible due to our application of Poissonization technique. This analysis gives significantly simpler proofs of known results, but also better competitive ratios for several well studied prophet inequalities. 

A conjecture in the field is that the optimal competitive ratio for the non \IID prophet inequality with order-selection is the same as the optimal prophet-inequality ratio for \IID random variables (i.e., $\approx 0.745$). One possible way of achieving this is choosing a different time of arrival distribution for each random variable. This is an idea that was employed in the recent result by Peng and Teng \cite{bz-ospi-22}. Together with the shards point of view, it might be possible to argue that the behavior of the shards (with different time of arrival distributions) can mimic the realizations more closely than otherwise using a uniform time of arrival, allowing the results for the \IID case to go through. We leave this as a potential future direction.  

\bibliographystyle{alpha}
\bibliography{references}

\newcommand{\etalchar}[1]{$^{#1}$}
 \providecommand{\CNFX}[1]{ {\em{\textrm{(#1)}}}}
\begin{thebibliography}{GMTS23}

\bibitem[ACK18]{ack-pss1b-18}
Yossi Azar, Ashish Chiplunkar, and Haim Kaplan.
\newblock Prophet secretary: Surpassing the 1-1/e barrier.
\newblock In {\'{E}}va Tardos, Edith Elkind, and Rakesh Vohra, editors, {\em Proceedings of the 2018 {ACM} Conference on Economics and Computation, Ithaca, NY, USA, June 18-22, 2018}, pages 303--318. {ACM}, 2018.

\bibitem[AEE{\etalchar{+}}17]{aeehk-b1op-17}
Melika Abolhassani, Soheil Ehsani, Hossein Esfandiari, MohammadTaghi Hajiaghayi, Robert~D. Kleinberg, and Brendan Lucier.
\newblock Beating 1-1/e for ordered prophets.
\newblock In Hamed Hatami, Pierre McKenzie, and Valerie King, editors, {\em Proceedings of the 49th Annual {ACM} {SIGACT} Symposium on Theory of Computing, {STOC} 2017, Montreal, QC, Canada, June 19-23, 2017}, pages 61--71. {ACM}, 2017.

\bibitem[AGSC02]{ags-rpiwm-02}
David Assaf, Larry Goldstein, and Ester Samuel-Cahn.
\newblock Ratio prophet inequalities when the mortal has several choices.
\newblock {\em The Annals of Applied Probability}, 12(3):972--984, 2002.

\bibitem[ASC00]{as-srpimm-00}
David Assaf and Ester Samuel-Cahn.
\newblock Simple ratio prophet inequalities for a mortal with multiple choices.
\newblock {\em Journal of Applied Probability}, 37(4):1084--1091, 2000.

\bibitem[BC23]{bc-piosb-22}
Archit Bubna and Ashish Chiplunkar.
\newblock Prophet inequality: Order selection beats random order.
\newblock In {\em Proceedings of the 24th ACM Conference on Economics and Computation}, EC '23, page 302–336, New York, NY, USA, 2023. Association for Computing Machinery.

\bibitem[Cam60]{c-atpbd-60}
Lucien~Le Cam.
\newblock An approximation theorem for the poisson binomial distribution.
\newblock {\em Pacific Journal of Mathematics}, 10:1181--1197, 1960.

\bibitem[CFH{\etalchar{+}}21]{cfhov-ppmot-21}
Jos{\'{e}}~R. Correa, Patricio Foncea, Ruben Hoeksma, Tim Oosterwijk, and Tjark Vredeveld.
\newblock Posted price mechanisms and optimal threshold strategies for random arrivals.
\newblock {\em Math. Oper. Res.}, 46(4):1452--1478, 2021.

\bibitem[CSZ21]{csz-pstbs-18}
Jose Correa, Raimundo Saona, and Bruno Ziliotto.
\newblock Prophet secretary through blind strategies.
\newblock {\em Mathematical Programming}, 08 2021.

\bibitem[dH12]{h-ptcm-12}
Frank den Hollander.
\newblock Probability theory : The coupling method.
\newblock 2012.

\bibitem[EFN18]{efn-pso-18}
Tomer Ezra, Michal Feldman, and Ilan Nehama.
\newblock Prophets and secretaries with overbooking.
\newblock In {\em Proceedings of the 2018 ACM Conference on Economics and Computation}, EC '18, page 319–320, New York, NY, USA, 2018. Association for Computing Machinery.

\bibitem[EHLM17]{ehlm-ps-17}
Hossein Esfandiari, MohammadTaghi Hajiaghayi, Vahid Liaghat, and Morteza Monemizadeh.
\newblock Prophet secretary.
\newblock {\em SIAM Journal on Discrete Mathematics}, 31(3):1685--1701, 2017.

\bibitem[EHLM19]{ethlm-psmpc-19}
Hossein Esfandiari, Mohammad~Taghi Hajiaghayi, Brendan Lucier, and Michael Mitzenmacher.
\newblock Prophets, secretaries, and maximizing the probability of choosing the best.
\newblock {\em International Conference on Artificial Intelligence and Statistics. AISTATS}, 2019.

\bibitem[GM66]{gm-rms-66}
John~P. Gilbert and Frederick Mosteller.
\newblock Recognizing the maximum of a sequence.
\newblock {\em Journal of the American Statistical Association}, 61(313):35--73, 1966.

\bibitem[GMTS23]{gmts-pisrofos-23}
Giordano Giambartolomei, Frederik Mallmann-Trenn, and Raimundo Saona.
\newblock Prophet inequalities: Separating random order from order selection.
\newblock {\em ArXiv}, abs/2304.04024, 2023.

\bibitem[{Gur}23]{gurobi}
{Gurobi Optimization, LLC}.
\newblock {Gurobi Optimizer Reference Manual}, 2023.

\bibitem[HK82]{hk-csrse-82}
T.~P. Hill and Robert~P. Kertz.
\newblock Comparisons of stop rule and supremum expectations of i.i.d. random variables.
\newblock {\em Ann. Probab.}, 10(2):336--345, 05 1982.

\bibitem[HKS07]{hks-aomdpi-07}
Mohammad~Taghi Hajiaghayi, Robert Kleinberg, and Tuomas Sandholm.
\newblock Automated online mechanism design and prophet inequalities.
\newblock In {\em Proceedings of the 22nd National Conference on Artificial Intelligence - Volume 1}, AAAI'07, page 58–65. AAAI Press, 2007.

\bibitem[HS23]{hs-ttsopi-23}
Martin Hoefer and Kevin Schewior.
\newblock {Threshold Testing and Semi-Online Prophet Inequalities}.
\newblock In Inge~Li G{\o}rtz, Martin Farach-Colton, Simon~J. Puglisi, and Grzegorz Herman, editors, {\em 31st Annual European Symposium on Algorithms (ESA 2023)}, volume 274 of {\em Leibniz International Proceedings in Informatics (LIPIcs)}, pages 62:1--62:15, Dagstuhl, Germany, 2023. Schloss Dagstuhl -- Leibniz-Zentrum f{\"u}r Informatik.

\bibitem[JMZ22]{jwj-tgmpiosk-22}
Jiashuo Jiang, Will Ma, and Jiawei Zhang.
\newblock {\em Tight Guarantees for Multi-unit Prophet Inequalities and Online Stochastic Knapsack}, pages 1221--1246.
\newblock 2022.

\bibitem[KS77]{ks-sfv-77}
Ulrich Krengel and Louis Sucheston.
\newblock Semiamarts and finite values.
\newblock {\em Bull. Amer. Math. Soc.}, 83(4):745--747, 07 1977.

\bibitem[KS78]{ks-sapfv-78}
Ulrich Krengel and Louis Sucheston.
\newblock On semiamarts, amarts, and processes with finite value.
\newblock {\em Probability on Banach spaces}, 4:197--266, 1978.

\bibitem[KW12]{kw-mpi-12}
Robert Kleinberg and Seth~Matthew Weinberg.
\newblock Matroid prophet inequalities.
\newblock In {\em Proceedings of the Forty-Fourth Annual ACM Symposium on Theory of Computing}, STOC '12, page 123–136, New York, NY, USA, 2012. Association for Computing Machinery.

\bibitem[KW19]{kw-mpiam-19}
Robert Kleinberg and S.~Matthew Weinberg.
\newblock Matroid prophet inequalities and applications to multi-dimensional mechanism design.
\newblock {\em Games Econ. Behav.}, 113:97--115, 2019.

\bibitem[PT22]{bz-ospi-22}
Bo~Peng and Zhihao~Gavin Tang.
\newblock Order selection prophet inequality: From threshold optimization to arrival time design.
\newblock In {\em 2022 IEEE 63rd Annual Symposium on Foundations of Computer Science (FOCS)}, pages 171--178, 2022.

\bibitem[R\'62]{lefttoright}
Alfred R\'enyi.
\newblock Th\'eorie des \'el\'ements saillants d'une suite d'observations.
\newblock {\em Annales de la facult\'e des sciences de l'universit\'e de Clermont. Math\'ematiques}, 8(2):7--13, 1962.

\bibitem[SC84]{s-ctsrm-84}
Ester Samuel-Cahn.
\newblock Comparison of threshold stop rules and maximum for independent nonnegative random variables.
\newblock {\em The Annals of Probability}, 12(4):1213--1216, 1984.

\bibitem[SHL24]{oracle-augmented}
Har-Peled Sariel, Elfarouk Harb, and Vasilis Livanos.
\newblock Oracle-augmented prophet inequalities.
\newblock 2024.

\bibitem[Sin18]{s-couup-18}
Sahil Singla.
\newblock {\em Combinatorial Optimization Under Uncertainty: Probing and Stopping-Time Algorithms}.
\newblock PhD thesis, CMU, 2018.
\newblock \url{http://reports-archive.adm.cs.cmu.edu/anon/2018/CMU-CS-18-111.pdf}.

\bibitem[Wan86]{w-cma-86}
Y.~H. Wang.
\newblock Coupling methods in approximations.
\newblock {\em The Canadian Journal of Statistics / La Revue Canadienne de Statistique}, 14(1):69--74, 1986.

\end{thebibliography}

\begin{appendices}

\section{Missing proofs}
\apdxlab{A}
\subsection{Proof of \lemref{folklore}}
\label{proof-folklore}
\begin{proof}
For $x\in [0,1]$, the process that independently chooses a time $t_i$ uniformly at random from $[0,1]$ has $\Prob{t_i \leq x}=x$. 

For the second process, let $\sigma$ be the random permutation drawn from $\mathbb{S}_n$. For $x\in [0,1]$,
\[
\Prob{T_i \leq x} = \sum_{j=1}^n \Prob{t_{(j)}\leq x }\Prob{\sigma(i)=j}.
\]
Where $t_{(j)}$ is the $j$-th order statistic of $t_1, \ldots, t_n$ generated by the algorithm. But then
\begin{align*}
    \Prob{T_i \leq x} &= \sum_{j=1}^n \Prob{t_{(j)}\leq x }\Prob{\sigma(i)=j} = \sum_{j=1}^n \frac{1}{n} \sum_{\beta=j}^n {n \choose \beta} x^\beta (1-x)^{n-\beta} \\
&= \frac{1}{n}\sum_{\beta=1}^n {n \choose \beta} x^\beta (1-x)^{n-\beta} \beta = \frac{1}{n} n x = x.
\end{align*}
To show independence, we have for $a,b \in [n]$ such that $a\neq b$, and $x,y\in [0,1]$ such that $x\leq y$
\begin{align*}
    \Prob{T_a\leq x, T_b \leq y} &= \sum_{i=1}^n \sum_{j=i+1}^n \Prob{t_{(i)}\leq x, t_{(j)}\leq y}\Prob{\sigma(a)=i, \sigma(b)=j}  \\ 
    &= \frac{1}{n(n-1)}\sum_{i=1}^n \sum_{j=i+1}^n \Prob{t_{(i)}\leq x, t_{(j)}\leq y} \\
    &= \frac{1}{n(n-1)} \sum_{i=1}^n \sum_{j=i+1}^n \frac{n!}{(i-1)!(j-i-1)!(n-j)!} \int_0^x \int_0^y u^{i-1}(v-u)^{j-i-1}(1-v)^{n-j} dv du\\
    &= \frac{1}{n(n-1)} \int_0^x \int_0^y \sum_{i=1}^n \sum_{j=i+1}^n \frac{n!}{(i-1)!(j-i-1)!(n-j)!} u^{i-1}(v-u)^{j-i-1}(1-v)^{n-j} dvdu\\
    &= \frac{1}{n(n-1)} \int_0^x \int_0^y \frac{n!}{(n-2)!} dvdu = xy = \Prob{T_a\leq x}\Prob{T_b\leq y}.
\end{align*}
Where the interchange of summation and integral follows by Fubini's theorem. Higher order independence follows similarly as above. 
\end{proof}

\subsection{Proof of \lemref{replace:poisson}}
\label{proof-replace:poisson}
\begin{proof}
    Consider the categorical random variable $Y_r \in \Re^{k\times k}$ for which canonical box (if any) realization $r$ arrives in. Hence, it is a categorical random variable parameterized by $p_i \in \Re^{k\times k}$. We have that $\hat{p}_i=\Prob{X_i\geq \tau_k}$. But recall that $\sum_{i=1}^n \Prob{X_i \geq \tau_k} = q$ and so by \IID symmetry and continuity, we have $\hat{p}_i = \frac{q}{n}$. Hence, by \lemref{distance}
    \[
    d(S_n, T_n) \leq \sum_{i=1}^n \frac{2q^2}{n^2} = \frac{2q^2}{n}. 
    \]
    The final remark follows by the additivity of Poisson distributions (i.e., if $X\sim \Poisson(\lambda_1), Y\sim \Poisson(\lambda_2)$, then $X+Y\sim \Poisson(\lambda_1+\lambda_2)$). Taking $k,n \to \infty$, then the variational distance is $0$, and the number of realizations that falls into $\circledcirc$ is the sum of the realizations in the canonical boxes inside $\circledcirc$ (that are coupled with the Poisson variables). 
\end{proof}

\section{Code for \IID prophet inequality getting $\approx 0.7406$}
\apdxlab{B}
\begin{enumerate}
    \item numpy (Tested with version 1.21.5), Scipy (Tested with version 1.7.3)
\end{enumerate} 
To copy the code directly, use \href{https://ideone.com/5gsEqx}{this link}
\begin{python}
import numpy as np
from scipy.optimize import minimize
import scipy

m = 10 #m parameter from paper

def lamb(j, cs):
    return 1/m * sum(cs[i] for i in range(1, j))

#Computes f_j(alphas, alphat) in time O(m^2)
def fj(j, cs, l):
    part1 = 1-np.exp(-lamb(j, cs))
    part2 = 0
    for k in range(j, m+1):
        part2 += np.exp(-lamb(k, cs)) * (1-np.exp(-cs[k]/m)) * l/cs[k]
    return part1+part2

def evaluate_competitive_ratio(cs):
    for i in range(1, len(cs)):
        if cs[i]<cs[i-1]:
            raise Exception("Values are not increasing")
    
    competitive_ratio = 1-np.exp(-1/m * float(sum(cs)) )
    competitive_ratio = min(sum([np.exp(-lamb(k, cs)) * (1-np.exp(-cs[k]/m))/cs[k]  for k in range(1, m+1)]), competitive_ratio)

    for j in range(2, m+1): 
        alphat_bounds = [(cs[j-1],cs[j])]
        x0 = (cs[j-1]+cs[j])/2.0
        
        res = minimize(lambda l: fj(j, cs, l[0])/(1-np.exp(-l[0])), 
                       x0=x0, 
                       bounds=alphat_bounds)
        """As a sanity check, make sure res.fun <= a few values in the middle to make sure minimization worked"""
        for xx in np.linspace(alphat_bounds[0][0], alphat_bounds[0][1], 1000):
            assert res.fun <= fj(j, cs, xx)/(1-np.exp(-xx)), (alphat_bounds, xx, res)
        
        competitive_ratio = min(competitive_ratio, res.fun)
    return competitive_ratio
    
cs = [0.        , 0.07077646, 0.2268947 , 0.42146915, 0.60679691,
       0.8570195 , 1.17239753, 1.51036256, 1.9258193 , 2.88381902,
       3.97363258]
c = evaluate_competitive_ratio(cs)
print(c)
\end{python}

\newpage 
\section{Code for Prophet Secretary}
Requires libraries:
\apdxlab{C}
\begin{enumerate}
    \item numpy (Tested with version 1.21.5)
    \item scipy (Tested with version 1.7.3)
    \item mpmath (Tested with version 1.2.1)
\end{enumerate} 
To copy the code directly, use \href{https://ideone.com/AuSVbj}{this link}

\begin{python}
import numpy as np
from scipy.optimize import minimize
import mpmath as mp
from mpmath import mpf
import scipy

m = 16 #m parameter from paper
mp.dps = 500 #This will force mpmath to use a precision of 
             #500 decimal places, just as a sanity check

comp_ratio = mpf('0.6724') #This is the competitive ratio we claim

def stable_qtk(x):
    #The function (1-e^(-x))/x is unstable for small x, so we will 
    # Lower bound it using the summation in Equation 13 in the paper
    ans = mpf('0')
    for beta in range(30):
        ans += mp.exp(-x) * x**beta / mp.factorial(beta) * 1/(beta+1)
    return ans

#Computes f_j(alphas, eta) in time O(m^2)
def fj(j, alphas, eta):
    part1 = mpf('0')
    for k in range(1, j): #Goes from 1 to j-1 as in paper
        part1 += mpf('1')*1/m * (1-alphas[k])

    #alphas_hat[nu]=alphas[nu] if nu<=j-1 and eta if nu==j
    alphas_hat = [alphas[nu] for nu in range(j)] + [eta]    
    part2 = mpf('0')
    for k in range(j, m+1): #Goes from j to m as in paper
        product = mpf('1')
        for nu in range(1, k): #Goes from 1 to k-1
            product *= (alphas[nu]**(1/m))
        
        wk = mpf('0')
        s_nu = mpf('0')
        for nu in range(j): #from 0 to j-1
            r_nu = (m-(k-1)+nu)/m * mp.log(alphas_hat[nu]/alphas_hat[nu+1])
            wk += mp.exp(-s_nu)*(1-mp.exp(-r_nu)) * 1/(m-(k-1)+nu)
            s_nu += r_nu
            
        q_t_k = stable_qtk( 1/m * mp.log(eta/alphas[k]) )
        part2 += product * wk * q_t_k
    
    return part1 + part2

def verify_competitive_ratio(alphas):
    assert np.isclose(np.float64(alphas[0]), 1) #first should be 1
    assert np.isclose(np.float64(alphas[-1]), 0) #Last should be 0
    assert len(alphas)==(m+2)
    
    defeciency = mpf('1')  #This quantity is the mninimum of
                    # f_j(alpha1, ..., alpha m, eta)- comp_ratio*(1-eta)
                    #This needs to be >=0 at the end of the code

    for j in range(1, m+2): #Goes from 1 to m+1 as in paper
        eta_bounds = [(np.float64(alphas[j]),np.float64(alphas[j-1]))] 
        
        x0 = [np.float64((alphas[j]+alphas[j-1])/2)]
        res = minimize(lambda alphat: 
                           fj(j, alphas, alphat[0]) - comp_ratio*(1-alphat[0]), 
                       x0=x0, 
                       bounds=eta_bounds)
        """
        As a sanity check, we will evaluate fj(alphas, x) - comp_ratio*(1-x) for x in 
        eta_bounds and assert that res.fun (the minimum value we got) is <= that. 
        This is just a sanity check to increase the confidence that the minimizer 
        actually got the right minimum
        """
        opt = fj(j, alphas, res.x[0]) - comp_ratio*(1-res.x[0])
        trials = np.linspace(eta_bounds[0][0], eta_bounds[0][1], 200) #200 breaks 
        min_in_trials = min([ fj(j, alphas, x) - comp_ratio*(1-x) for x in trials ])
        assert opt <= min_in_trials
        """
        End of sanity check
        """
        defeciency = min(defeciency, opt)
            
    if defeciency>=mpf('0'):
        print("Claimed bound is True")
    else:
        print("Claimed bound is False")

alphas = [mpf('1.0'), mpf('0.66758603836404173'), mpf('0.62053145929311715'), 
mpf('0.57324846512425975'),
 mpf('0.52577742556626594'), mpf('0.47816906417879007'), mpf('0.43049233470891257'),
 mpf('0.38283722646593055'), mpf('0.33533950489086961'), mpf('0.28831226925828957'),
 mpf('0.23273108361807243'), mpf('0.19315610994691487'), mpf('0.16547915613363387'),
 mpf('0.13558301500280728'), mpf('0.10412501367635961'), mpf('0.071479537771643828'),
 mpf('0.036291830527618585'), mpf('0.0')]

verify_competitive_ratio(alphas) #Takes roughly 20 seconds
\end{python}

\newpage 
\section{Code for \IID \textsc{Semi-Online}}
\apdxlab{D}
To copy the code directly, use \href{https://ideone.com/6G7hou}{this link}

\begin{python}
import numpy as np
from scipy.optimize import minimize
import scipy
import math
np.random.seed(0)

k = 3
m = 420
p = 6
comp_ratio = 0.8901 #The competitive ratio we claim.

def verify_competitive_ratio(cs, eps):
    assert m
    
    C = [[cs[outer + inner] for inner in range(p-1, -1, -1) for _ in range(m//p)] for outer in range(0, len(cs)-1,p)]
    C = np.array(C)
    
    Prob = [[[0 for i in range(m+5)] for j in range(k+5)] for b in range(2)]
    
    def cost(l):
        l = l[0]
        for i in range(m+1):
            for j in range(k+1):
                for b in range(2):
                    if j>=k or i>=m:
                        Prob[b][j][i]=b
        
        for i in range(m+1, -1, -1):
            for j in range(k+1, -1, -1):
                for b in range(2):
                    if j>=k or i>=m:
                        continue 
                    if l<=C[j, i]:
                        Prob[b][j][i] = np.exp(-C[j,i]/m)*Prob[b][j][i+1] + (1-np.exp(-C[j,i]/m))*(l/C[j, i] * Prob[1][j+1][i+1] + (C[j, i]-l)/C[j, i] * Prob[0][j+1][i+1])
                    else:
                        Prob[b][j][i] = np.exp(-C[j,i]/m)*Prob[b][j][i+1] + (1-np.exp(-C[j,i]/m))*Prob[1][j+1][i+1]
        
        return Prob[0][0][0] - comp_ratio*(1-np.exp(-l))
    
    defeciency = 1 #This is the minimum of Prob[0][0][0] - comp_ratio*(1-e^(-l))
                    #Needs to be >=0 at end of execution to that claimed ratio is True
        
    defeciency = min(defeciency, 1-np.exp(-sum(C[0])/m) - comp_ratio) #Case of l'>max(C)
    
    """Run global optimization on cost in the range (0, max(C)]"""
    res = scipy.optimize.shgo(cost, bounds=[(0,C.max())], iters=10,  
                              options={'disp':False, 'f_tol':1e-9})
    defeciency = min(defeciency, res.fun)

    "As a sanity check, make sure that global minimizer succeeded"
    ls = np.linspace(0.0, C.max(), math.ceil(C.max()/eps))
    for l in ls:
        assert res.fun <= cost([l])
    """End of sanity check"""
    
    if defeciency>=0:
        print(f"Claimed bound of {comp_ratio} is True")
    else:
        print(f"Claimed bound of {comp_ratio} is False")

def MonteCarlo(cs):
    C = [[cs[outer + inner] for inner in range(p-1, -1, -1) for _ in range(m//p)] for outer in range(0, len(cs)-1,p)]
    C = np.array(C)

    N = 4000 #Large n, bound converges for n->Infinity
    epochs = 10000 
    prophet = 0
    alg = 0
    for _ in range(epochs):
        X = np.random.uniform(0, 1, (N, 2)) #First dim=time, second dim=Xi~U(0, 1)
        X = X[X[:, 0].argsort()] #Sort by time of arrival 
        r = 0
        clock = 0
        i_star = None
        for i in range(len(X)):
            ti, vi = X[i][0], X[i][1]
            if r<k and vi>=1- C[r][math.floor(ti*m)]/N  and ti>=clock:
                r = r + 1
                i_star = i
                clock = math.ceil(ti*m)/m
        prophet += X[:, 1].max()
        if i_star:
            alg += X[i_star][1]
    
    print(f"The competitive ratio on n={N} IID U(0, 1) random variables is {alg/prophet} using {epochs} epochs.")

cs0 = [3.64589394e+00, 3.58116098e+00, 2.03323633e+00, 1.93319241e+00,
       1.15603731e+00, 9.92652855e-01, 6.10147568e-01, 3.94833386e-01,
       2.41093283e-01, 1.36659577e-01, 4.80563875e-02, 2.83455285e-02,
       8.39298670e-02, 1.91858842e-02, 0.00133218127, 1.33218127e-03,
       1.05769060e-03, 1.05769044e-03]
  
verify_competitive_ratio(cs0, eps=0.0001) #Takes  3-4 minutes
MonteCarlo(cs0) #Takes ~1 min
\end{python}

\newpage 
\section{Code for \IID \textsc{Top-$1$-of-$2$}.}
Requires libraries:
\apdxlab{E}
\begin{enumerate}
    \item numpy (Tested with version 1.21.5)
    \item scipy (Tested with version 1.7.3)
\end{enumerate} 
To copy the code directly, use \href{https://ideone.com/pnbjas}{this link}
\begin{python}
import numpy as np
from scipy.optimize import minimize 
import scipy 

m = 10 #m parameter from paper
comp_ratio = 0.883 #The competitive ratio we claim.

def lamb(j, cs):
    return 1/m * sum(cs[i] for i in range(1, j))

def fj(j, cs, l):
    part1 = 1-np.exp(-lamb(j, cs))
    part2 = 0
    for k in range(j, m+1):
        term = np.exp(- lamb(k, cs))
        term *= np.exp(-1*(m+1)/m*cs[k])
        inner = (m*np.exp(cs[k])*l*(np.exp(cs[k]/m)-1)*(2*cs[k]-l) 
                    + cs[k]*l*np.exp(k*cs[k]/m)*(l-cs[k]))
        term *= inner 
        term /= (m*cs[k]**2)
        part2 += term
    return part1+part2

def verify_competitive_ratio(cs):
    for i in range(1, len(cs)):
        if cs[i]<cs[i-1]:
            raise Exception("cs should be sorted in ascending order")
    
    defeciency = 1 # This is the min value of f_j(c1, ..., cm, \ell') - comp_ratio*(1-e^(-\ell'))
                    #This needs to be >=0 by the end of execution so that the bound we claim
                    #Is truthful.
    
    defeciency = min(defeciency, 1-np.exp(-1/m * float(sum(cs))) - comp_ratio)
    for j in range(1, m+1): 
        ell_bounds = [(cs[j-1],cs[j])]
        
        res = scipy.optimize.shgo(lambda l: fj(j, cs, l[0]) - comp_ratio*(1-np.exp(-l[0])), 
                                  iters=10,
                       bounds=ell_bounds, options={'disp':False, 'f_tol':1e-9})
        
        """As a sanity check, make sure res.fun <= some values in the middle to make sure minimization worked"""
        for xx in np.linspace(ell_bounds[0][0], ell_bounds[0][1], 10000):
            assert res.fun <= fj(j, cs, xx) - comp_ratio*(1-np.exp(-xx)), (j, ell_bounds, xx, res)
        """End of sanity check"""
        defeciency = min(defeciency, res.fun)
    
    if defeciency>=0:
        print(f"Claimed bound of {comp_ratio} is True")
    else:
        print(f"Claimed bound of {comp_ratio} is False")

cs = [0., 0.35598315, 0.56202538, 0.86407969, 1.22558122, 1.65459166,
       2.14361195, 2.5868228 , 3.07922161, 4.0722262 , 5.21637928]

verify_competitive_ratio(cs)
\end{python}

\end{appendices}

\end{document}